\documentclass{article}
\usepackage[dvipsnames]{xcolor}
\usepackage{psfrag}
\usepackage{epsfig}
\usepackage{amsmath,amsfonts,amsthm,graphicx}
\usepackage{ulem}
\usepackage{mathrsfs}
\usepackage{amsthm}
\usepackage{dashrule}
\usepackage{hyperref}

\usepackage{fullpage}
\usepackage{wrapfig}
\usepackage{pict2e}
\usepackage[vcentermath]{youngtab}

\usepackage[colorinlistoftodos,prependcaption,textsize=tiny,backgroundcolor=blue!20]{todonotes}

\usepackage{comment}
\usepackage{rotating}

\numberwithin{equation}{section}

\newtheorem{theorem}{Theorem}

\theoremstyle{definition}

\newtheorem{proposition}[theorem]{Proposition}

\newcommand\mdash{\hdashrule[0.5ex]{10pt}{1pt}{1pt}}

\newcommand\es{\emptyset}

\newcommand{\com}[1]{(*{\textbf{#1}}*)}

\newcommand{\bal}{\begin{align}}
\newcommand{\eal}{\end{align}}
\newcommand{\beq}{\begin{equation}}
\newcommand{\eeq}{\end{equation}}
\newcommand\beqa{\begin{eqnarray}}
\newcommand\eeqa{\end{eqnarray}}
\newcommand\bea{\begin{array}}
\newcommand\eea{\end{array}}

\newcommand\wb\nu
\newcommand\wf\lambda
\newcommand\yb\mu

\newcommand\yf\tau

    \newcommand{\mtwo}[4]{\left(%
    \begin{array}{cc}
    #1 & #2 \\
    #3 & #4 \\
    \end{array}%
    \right)}

\newcommand{\su}{\mathfrak{su}}
\renewcommand{\sl}{\mathfrak{sl}}
\newcommand{\psl}{\mathfrak{psl}}
\newcommand{\psu}{\mathfrak{psu}}
\newcommand{\gl}{\mathfrak{gl}}

    \newcommand{\nn}{\nonumber}
    
    \newcommand{\COMMENT}[1]{}
    
    \newcommand{\neqa}{\nonumber\end{eqnarray}}
    
\def\o{{        \omega}}
\def\a{{\alpha}}

\def\[{\left[}
\def\]{\right]}

\def\l{\lambda}
\def\e{\epsilon}

\def\s{\sigma}
\def\a{\alpha}
\def\b{\beta}

\def\<{\langle}
\def\>{\rangle}

\def\i2{\frac{i}{2}}

\def\<{\langle}
\def\>{\rangle}

\def\cF{{\cal F}}

\def\i2{\frac{i}{2}}

\def\1h{\hat 1}
\def\2h{{\hat 2}}
\def\3h{{\hat 3}}
\def\4h{{\hat 4}}

\def\be{\begin{eqnarray}}
\def\ee{\end{eqnarray}}
\def\no{\nonumber}
\newcommand{\eqn}{\begin{eqnarray}}
\newcommand{\enn}{\end{eqnarray}}

    \def\CA{{\cal A}}
    \def\CB{{\cal B}}
    
    \def\CD{{\cal D}}
    \def\CE{{\cal E}}
    \def\CF{{\cal F}}
    
    \def\CH{{\cal H}}

    \def\CK{{\cal K}}
    
    \def\CM{{\cal M}}
    \def\CN{{\cal N}}
    \def\CO{{\cal O}}
    \def\CP{{\cal P}}

    \def\CU{{\cal U}}

    \def\<{\left\langle\,}
    \def\>{\, \right\rangle}
    \def\[{\left[}
    \def\]{\right]}

   \def\gl{{\mathfrak{gl}}}
   \def\GL{{\mathsf{GL}}}
   \def\SL{{\mathsf{SL}}}
   \def\SU{{\mathsf{SU}}}
   \def\SO{{\mathsf{SO}}}
   \def\su{{\mathfrak{su}}}
   \def\sl{{\mathfrak{sl}}}
   \def\so{{\mathfrak{so}}}

   \def\groupn{\mathsf{n}}
   \def\groupm{\mathsf{m}}
      \def\groupp{\mathsf{p}}
   \def\groupq{\mathsf{q}}
      \def\groups{\mathsf{s}}
   \def\groupk{\mathsf{k}}
   \def\groupr{\mathsf{r}}

   \def\algu{\mathfrak{u}}
   \def\algg{\mathfrak{g}}
   
   \def\algh{\mathfrak{h}}
    \def\glnm{\gl(\groupn|\groupm)}
    \def\sunm{\su(\groupn|\groupm)}

\newcommand{\En}{{\mathsf E}}

\usepackage{amssymb}
\usepackage{subfig}
\usepackage{overpic}

\usepackage{color}

\newcommand{\zo}{{\overline{0}}}
\newcommand{\zi}{{\overline{1}}}
\newcommand{\ii}{{\mathbf{i}}}
\newcommand{\gE}{{\mathrm{E}}}
\newcommand{\cE}{{\mathcal{E}}}
\newcommand{\ch}{{\mathrm{h}}}

\renewcommand{\a}{\alpha}

\renewcommand{\b}{\beta}

\newcommand{\HWS}{|{\rm HWS \rangle}}
\newcommand{\HWSbra}{\langle{\rm HWS |}}
\newcommand{\LWS}{|{\rm LWS \rangle}}
\newcommand{\fvac}{|{\rm 0} \rangle}
\newcommand{\act}{\triangleright}
\usepackage{bbm}
\renewcommand{\ii}{\ensuremath{\mathbbm i}}
\newcommand{\trnp}{^{\text{\tiny{T}}}}

\begin{document}
\begin{flushright} IGC -17/11-1 \end{flushright} 
\Yboxdim4pt
\begin{center}
{ \bf The complete unitary dual of non-compact Lie superalgebra $\su(\groupp,\groupq|\groupm)$ via the generalised oscillator formalism, and non-compact Young diagrams}
\vspace{2cm}

{\bf Murat G\"unaydin$^\groupp$ and Dmytro Volin$^{\groupp,\groupq,\groupm}$ } \\
$^\groupp$ Institute for Gravitation and the Cosmos \& Physics Department  \\ The Pennsylvania State University,
University Park, PA 16802, USA  \\
$^\groupq$ Nordita, KTH Royal Institute of Technology and Stockholm University,\\
Roslagstullsbacken 23, SE-106 91 Stockholm, Sweden\\
$^\groupm$ School of Mathematics \& Hamilton Mathematics Institute, \\Trinity College Dublin,
College Green, Dublin 2, Ireland

\end{center}
\vspace{1cm}
\begin{center}
\scalebox{0.7}{
\begin{picture}(330,190)(-120,-100)
\color{gray!10}
\thinlines
\multiput(-115,-100)(0,10){20}{\line(1,0){320}}
\multiput(-110,-105)(10,0){32}{\line(0,1){200}}
\thicklines
\thicklines
\color{gray!10}
\polygon*(-115,95)(-90,95)(-90,90)(-60,90)(-60,80)(-20,80)(-20,70)(0,70)(0,0)(-115,0)
\polygon*(205,0)(80,0)(80,-70)(120,-70)(120,-80)(150,-80)(150,-100)(205,-100)
\color{gray}
\drawline(-90,95)(-90,90)(-60,90)(-60,80)(-20,80)(-20,70)(0,70)(0,0)(-115,0)
\drawline(205,0)(80,0)(80,-70)(120,-70)(120,-80)(150,-80)(150,-100)(205,-100)
%
\color{blue!10}
\polygon*(0,0)(0,60)(20,60)(20,50)(40,50)(40,40)(80,40)(80,30)(120,30)(120,20)(150,20)(150,0)(80,0)(80,-60)(40,-60)(40,-50)(20,-50)(20,-40)(0,-40)(0,-30)(-20,-30)(-20,-20)(-60,-20)(-60,-10)(-90,-10)(-90,0)(0,0)
\color{blue}
\drawline(0,0)(0,60)(20,60)(20,50)(40,50)(40,40)(80,40)(80,30)(120,30)(120,20)(150,20)(150,0)(80,0)(80,-60)(40,-60)(40,-50)(20,-50)(20,-40)(0,-40)(0,-30)(-20,-30)(-20,-20)(-60,-20)(-60,-10)(-90,-10)(-90,0)(0,0)
\end{picture}
}
\end{center}
\vspace{0cm}
{\bf Abstract:}  We study the unitary representations of the non-compact real forms of the complex Lie superalgebra $\mathfrak{sl}(\groupn|\groupm)$. Among them, only the real form $\mathfrak{su}(\groupp,\groupq|\groupm)$ ($\groupp+\groupq=\groupn)$ admits nontrivial unitary representations, and all such representations are of the highest-weight type (or the lowest-weight type). We extend the standard oscillator construction of the unitary representations of non-compact Lie superalgebras over standard Fock spaces to generalised Fock spaces which allows us to define the action of oscillator determinants raised to non-integer powers. We prove that  the proposed construction yields all the unitary representations including those with continuous labels. The unitary representations can be diagrammatically represented by non-compact Young diagrams.  We apply our general results to the physically important case of four-dimensional conformal superalgebra $\mathfrak{su}(2,2|4)$ and show how  it yields readily its unitary representations including those corresponding to supermultiplets  of conformal fields with continuous (anomalous) scaling dimensions.

\newpage
\section*{Notations}

\begin{tabular}{rl}
Lie algebras: & $\mathfrak{g}$; $\gl$, $\sl$, $\su$, \ldots
\\
Universal enveloping algebra: & $\mathcal{U}(\mathfrak{g})$
\\
Lie groups:  &  $\mathsf{G}$; $\GL$, $\SL$, $\SU$, \ldots
\\
Ranks of (sub)algebra: & $\groupp,\groupq,\groupm,\groupn,\ldots$
\\
\\
$\mathbb{Z}_2$-numbers: & $\zo$ and $\zi$.
\\
parity grading function: & $p_i$
\\
c-grading function : & $c_i$
\\
\\
Generic index: & $i,j,k,\ldots$
\\
Generic p-even index: & $\mu,\wb,\ldots$
\\
p-even, c-odd index: & $\dot\alpha,\dot\beta,\ldots$
\\
p-even, c-even index: & $\alpha,\beta,\ldots$
\\
Generic p-odd index: & $a,b,\ldots$
\\
\\
$\gl(\groupm|\groupn)$ generators: & $\gE_{ij}$, $\gE_{\mu a}$, \ldots
\\
Cartan generators: & $\ch_i\equiv \gE_{ii}-(-1)^{p_i+p_{i+1}}\gE_{i+1,i+1}$
\\
Generic weight: & $m_i$ (eigenvalue of $\gE_{ii}$)
\\
Eigenvalue of $\gE_{\mu\mu}$:  & $\wb_{\mu}$ ($\wb$ without indices denotes set of all these eigenvalues)
\\
Eigenvalue of $\gE_{aa}$: & $\wf_a$ ($\wf$ without indices denotes set of all these eigenvalues)
\\
Fundamental weight: & $[\wb;\wf]$ or $[\wb_{L};\wf;\wb_{R}]$ (evaluated on highest-weight state)
\\
Repetition of the same value in weights: & $\mdash a\mdash\equiv a,a,\ldots,a$
\\
Dynkin label: & $\omega_i$ (eigenvalue of $\ch_i$ on highest-weight state)
\\
Dynkin weight: & $\langle \omega \rangle\equiv \langle \omega_1,\omega_2,\ldots\rangle$
\\
\\
Nodes of Kac-Dynkin-Vogan diagram: &
\\
$p_i=p_{i+1}$, $c_{i}=c_{i+1}$: & \raisebox{-0.4em}{\includegraphics[width=1.1em]{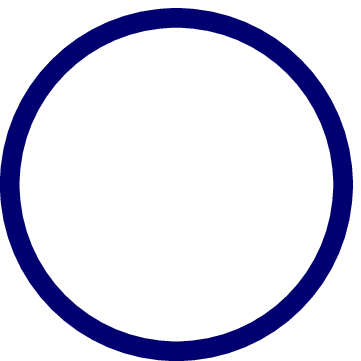}}
\\
$p_i=p_{i+1}$, $c_{i}\neq c_{i+1}$:  & \raisebox{-0.4em}{\includegraphics[width=1.1em]{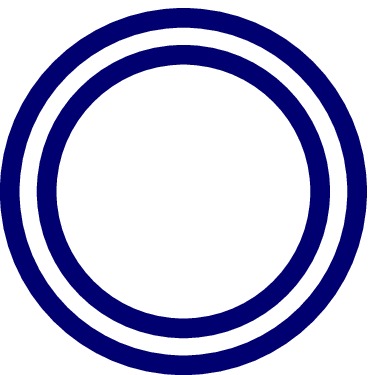}}   $\Leftrightarrow$ $\su(\ldots,\ldots)$
\\
$p_{i}\neq p_{i+1}$, $c_{i}= c_{i+1}$: & \raisebox{-0.4em}{\includegraphics[width=1.1em]{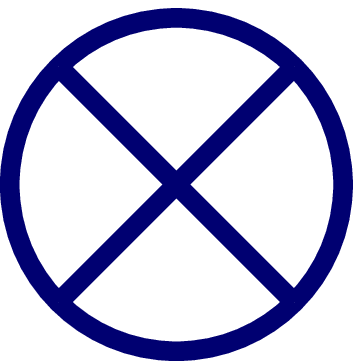}}  $\Leftrightarrow$ $\su(\ldots|\ldots)$
\\
$p_{i}\neq p_{i+1}$, $c_{i}\neq c_{i+1}$: &   \raisebox{-0.4em}{\includegraphics[width=1.1em]{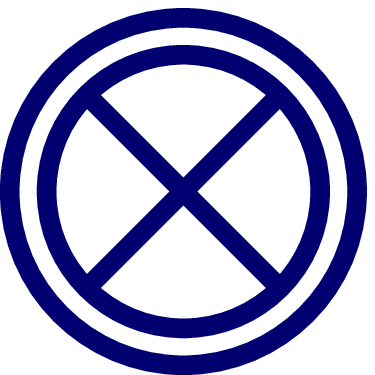}}  $\Leftrightarrow$ $\su(\ldots,\!|\ldots)$
\\
\\
colour index: &  $A,B,C,\ldots$
\\
Total number of colours: & $P$
\\
Sets of colours:  & ${\bf A},{\bf B},{\bf F},{\bf A}_{\!\Delta},{\bf B}_{\!\Delta},{\bf F}_{\!\Delta}\subset \{1,2,\ldots,P\}$
\\
Number of colours in a set: & $|{\bf A}|$, $|{\bf B}|,\ldots $
\\
\\
Ordinary Young diagrams (integer partitions): & $\yb,\yf,\yb_L,\yb_R$; e.g. $\yb\equiv\{\mu_1,\mu_2,\ldots\}$ with all $\mu_i\in\mathbb{Z}_{\geq 0}$ and $\mu_i\geq\mu_{i+1}$\,.
\\
Height of  partition: & $h_{\yb}$ (number of non-zero $\mu_i$)
\\
Size of  partition: & $|\yb|\equiv\yb_1+\yb_2+\ldots+\yb_{h_{\yb}}$
\\
\\
Bosonic oscillators: & $a_{\alpha}^{\vphantom{\dagger}},a_{\alpha}^{\dagger }$, $b_{\dot\alpha}^{\vphantom{\dagger}},b_{\dot\alpha}^{\dagger }$
\\
Fermionic oscillators: & $f_{a}^{\vphantom{\dagger}},f_{a}^{\dagger }$
\\
Fock space: & $\mathcal{F}$
\\
Fock vacuum: & $\fvac$
\\
$\gamma$-deformed Fock space: & $\mathcal{F}_{\gamma}$
\\
Irreducible modules of $\algu(\groupn)$: & $V_{\mu}$ (constructed from bosonic oscillators)\\
& $W_{\tau}$ (constructed from fermionic oscillators)
\end{tabular}
\newpage
\tableofcontents
\newpage
\section{Introduction}
Oscillator methods have a long and distinguished history in the study
of symmetries of physical systems.  Annihilation  and creation operators, that is the operators $a$ and $a^{\dagger}$ that satisfy the oscillator algebra
\be
[a,a^\dagger]=1\,,
\ee
 were introduced
in the early days of quantum mechanics. 
For example, in his seminal paper on relativistic quantum field theory of the electromagnetic field Dirac
introduced annihilation and creation operator for the photon~\cite{Dirac:1927dy}.  State spaces
generated by the action of bosonic (fermionic) creation operators on
the vacuum state were later called bosonic (fermionic) Fock spaces~\cite{Fock:1932vt}.
%

Fundamental importance of annihilation and creation operators for quantum field theory was eloquently stated by Weinberg in his textbook on quantum field theory \cite{Weinberg:1995mt}. Weinberg argues that the cluster decomposition principle requires the Hamiltonian to be a polynomial sum of annihilation and creation operators, with well-behaved dependence of coefficients on physical momenta.  On the other hand, Wigner's theorem implies that the symmetries of quantum mechanical systems  must  be realised unitarily on the corresponding  Hilbert spaces. Hence the physically relevant representations are often unitary and are likely to be realisable in terms of the oscillators, though an explicit realisation can be quite complicated for interacting theories. Classification of unitary representations using oscillator methods is one of  the key themes of this  article.

Due to the mentioned physical heritage, it is typical to associate  creation operators with generation of new physical states from a certain  vacuum state. Then, in the simplest realisation of symmetry generators  as bilinear combinations of  the oscillators, the constructed representations have  integer weights. This does not represent any obstacle  when one deals with compact symmetry groups where discreteness is natural anyway. For instance, Schwinger\footnote{ Schwinger's pioneering work  was published  recently as a book \cite{Schwingerbook}. }
 has applied the oscillator method to  study   the quantum theory
of angular momentum and construct
 representations of the rotation group $\SU(2)$, which was later extended to other compact groups.

 %
 However, when one turns to unitary realisations of Lie superalgebras and non-compact Lie algebras, one finds that non-integer weights become possible as well, while the oscillator approach is not directly applicable to the  construction of  unitary representations with non-integer weights, in general.  For instance, the so-called ladder representations of the covering group $\SU(2,2)$  of the conformal group $\SO(4,2)$, that are labelled by integer weights, were studied using bosonic annihilation
and creation operators in the 1960s  \cite{malkin_manko,nambu1,nambu2,barut_kleinert1,Barut:1967zz,barut_kleinert3,Mack:1969dg}, but a full classification of positive energy  (highest-weight) unitary representations  of $\SU(2,2)$ was given in \cite{Mack:1975je} without employing the oscillator method. A complete classification of all
unitary representations
(unitary dual) of $\SU(2,2)$ was  given in the mathematics literature later \cite{MR645645}.

In the 1960s non-compact Lie groups were studied by theoretical physicists as spectrum generating or as dynamical symmetry groups  \cite{Barut:1990zza,Dothan:1965aa}.
Beginning mid 1970s, non-compact global internal symmetry groups appeared within the framework of extended super-gravity theories as well as matter coupled super-gravity theories. Motivated
partly  by the goal of constructing the unitary representations of
non-compact global symmetry groups of supergravity theories, a general oscillator  method for construction of unitary representations
of non-compact groups was developed in \cite{Gunaydin:1981yq}. It was especially  well-suited for constructing irreducible unitary
lowest- or highest-weight representations with integer weights. Shortly after  the work of \cite{Gunaydin:1981yq}, a complete classification of the unitary highest-weight representations of non-compact simple groups appeared in the mathematics literature \cite{MR733809}.

The novel concept of supersymmetry was introduced in the physics literature in the 1970s and
 led to the theory of Lie superalgebras. Simple finite-dimensional
Lie superalgebras were classified by Kac \cite{Kac:1977em,Kac:1977qb}
who also developed the  representation theory of compact simple Lie superalgebras.
 Oscillator methods for constructing finite-dimensional representations
of compact Lie superalgebras were developed by the authors of \cite{BahaBalantekin:1980tup,BahaBalantekin:1980igq}.
 The general oscillator method of \cite{Gunaydin:1981yq} for constructing unitary representations of non-compact Lie algebras was extended to
the construction of unitary representations of  non-compact Lie  superalgebras in \cite{Bars:1982ep,Gunaydin:1987hb}. It was then applied to  construct explicitly the integer-weight unitary representations of superalgebras $\mathfrak{su}(2,2|\groupn)$ with even subalgebras $\mathfrak{su}(2,2)\oplus \mathfrak{u}(\groupn)$ in \cite{Gunaydin:1984fk,Gunaydin:1998jc,Gunaydin:1998sw}, of  $ \mathfrak{osp}( \groupn|2\groupm, \mathbb{R} ) $ with even subalgebras $\mathfrak{so}(\groupn)\oplus \mathfrak{sp}(2\groupm,\mathbb{R})$ in \cite{Gunaydin:1985tc,Gunaydin:1987hb,Gunaydin:1988kz} and  of $ \mathfrak{osp}(2\groupn^*|2\groupm) $ with even subalgebras $\mathfrak{so}^*(2\groupn)\oplus \mathfrak{usp}(2\groupn)$
in \cite{Gunaydin:1984wc,Gunaydin:1990ag}. The general oscillator methods of \cite{Bars:1982ep,Gunaydin:1984fk,Gunaydin:1985tc,Gunaydin:1987hb,Gunaydin:1988kz} were reformulated
in a more mathematical language to give a classification of holomorphic discrete series representations with integer weights
of $\mathfrak{osp}(\groupm|2\groupn,\mathbb{R})$ and $\mathfrak{su}(\groupp,\groupq|\groupm)$  in  \cite{MR1190748,MR1134934}.

Even with restriction to integer weights, the oscillator method is still very powerful. As  the generators of Lie superalgebras
are realized as bilinears of free bosonic and fermionic oscillators,  tensor product decomposition of the resulting
representations is very straightforward. This feature was used for studying of
the spectra of Kaluza-Klein supergravity theories and  their relations to  those of  certain field theories long before the AdS/CFT correspondence was conjectured within the framework of M/superstring theory \cite{Maldacena:1997re,Witten:1998qj,Gubser:1998bc}. Indeed, the Kaluza-Klein spectrum of IIB supergravity over AdS$_5\times$S$^5$
space was first obtained by the oscillator method by  tensoring
the CPT self-conjugate doubleton supermultiplet of  the symmetry superalgebra
$\mathfrak{su}(2,2|4)$ with itself repeatedly. The Lie superalgebra $\mathfrak{su}(2,2|4)$  is the $\mathcal{N}=8$, AdS$_5$ superalgebra and the CPT self-conjugate
doubleton supermultiplet is the $4d$, $\mathcal{N}$=$4$ Yang-Mills supermultiplet \cite{Gunaydin:1984fk}. The
CPT
self-conjugate doubleton supermultiplet does not have a Poincar\'e limit
in AdS$_5$ and decouples from the Kaluza-Klein
spectrum. The authors of \cite{Gunaydin:1984fk} pointed out  that  its
field theory
 lives on the boundary of AdS$_5$ on which $\SO(4,2)$ acts as a conformal
group and that the unique candidate for this boundary theory is the
four dimensional
$\mathcal{N}\!=\!4$ super Yang-Mills theory. Similarly, spectra of the compactifications
of 11-dimensional supergravity over
AdS$_4\times$S$^7$ and AdS$_7\times$S$^4$ were obtained simply  by tensoring the CPT self-conjugate ultra short supermultiplets of the symmetry superalgebras $\mathfrak{osp}(8|4,\mathbb{R})$ and $\mathfrak{osp}(8^*|4)$  in \cite{Gunaydin:1985tc}
and \cite{Gunaydin:1984wc}, respectively. Again it was pointed out that the ultra short supermultiplets of
$\mathfrak{osp}(8|4,\mathbb{R})$ and $\mathfrak{osp}(8^*|4)$ decouple from the spectrum as gauge modes and their field theories live on the boundaries of AdS$_4$ and AdS$_7$ as superconformal field theories \cite{Gunaydin:1985tc,Gunaydin:1984wc}, respectively.

A complete classification of positive energy unitary representations of four-dimensional superconformal algebra $\mathfrak{su}(2,2|1)$ was given in \cite{Flato:1983te} and for general $\mathfrak{su}(2,2|\groupn)$ $(n>1)$  in \cite{Dobrev:1985vh,Dobrev:1985qv}. Corresponding classifications of positive energy unitary representations  of conformal superalgebras in six, five and three space-time dimensions were studied in \cite{Minwalla:1997ka}. An attempt at classifying all unitarisable highest-weight modules of  classical  superalgebras was made by Jakobsen \cite{MR1214730}, but we should note that his classification statement for $\su(\groupp,\groupq|\groupm)$ case has  errors. For instance, certain  well-known positive energy unitary representations of $\mathfrak{su}(2,2|\groupn)$ (i.e the doubleton supermultiplets of \cite{Gunaydin:1984fk,Gunaydin:1998jc,Gunaydin:1998sw})  do not appear in his classification list.  There are also other issues.  The papers of \cite{Flato:1983te,Dobrev:1985vh,Dobrev:1985qv,Minwalla:1997ka,MR1214730} relied mainly on general aspects of Lie algebra theory for deriving their results and did not use the oscillator method except for the work of  \cite{Minwalla:1997ka} where the oscillator methods were used to construct explicitly some short supermultiplets whose existence could not be established conclusively otherwise. \\[0em]

The main goal of this paper is to construct {\it all} possible unitary representations, i.e. the unitary duals, of different real forms of the complex Lie superalgebra $\sl(\groupn|\groupm)$. Many partial results were already obtained in the past, as was reviewed above. However, to our knowledge, the correct and fully general statement has not yet appeared in the literature. This came as a little bit of a surprise for us as clearly the mathematical tools that could be used to this end existed for quite a while. We  decided to accomplish the classification goal not mainly to fill in the existing gap, but to introduce several new ideas that can be used in this perspective, in particular to further push the oscillator approach beyond integer-weighted representations. We also found that existing mathematical literature on the subject is quite difficult to comprehend by people with  background in physics, so we aimed to introduce the topic in simpler terms and try  to create a bridge between the two communities.

In section~\ref{sec:defs} we cover the basic definitions of the $\sl(\groupn|\groupm)$ Lie algebra, of its real forms, and of the unitary representations. Despite being almost textbook material, this subject has several confusing points where application of the ordinary Lie algebra intuition to  the supersymmetric case fails.

One of the confusing points is relation between the Lie algebra real form and the Hermitian conjugation properties of the representation module which we decided to recast in more abstract terms of *-algebras. Whereas in non-supersymmetric case there is a one-to-one correspondence between the two structures, this becomes rather a two-to-two correspondence in supersymmetric case which results in two distinct classes of unitary representations. Mixing them up leads to non-unitary representations which apppear for instance already in the works of  \cite{BahaBalantekin:1980tup,BahaBalantekin:1980igq}  and which are labelled by mixed supertableaux there. We stress that mixed
 supertableaux of  \cite{BahaBalantekin:1980tup,BahaBalantekin:1980igq} cannot describe unitary representations in principle, because they  combine representations that have different properties under Hermitian conjugation.

Another caveat is that Lie superalgebras allow for non-equivalent Borel decompositions. This is not the case for ordinary compact  Lie algebras. It is tempting to choose a Borel decomposition as close as possible to the ordinary Lie algebra case, i.e. to reduce the number of fermionic nodes on the Kac-Dynkin diagram thus getting the distinguished  diagram (in terminology of Kac \cite{Kac:1977em}, see also \cite{Kac:1977qb}). Whereas it is indeed a useful choice for compact Lie superalgebras, it is not the most convenient one for  the non-compact case. Moreover, it proves to be useful to consider interplay between all possible Borel decompositions. Our strategy to classify the  unitary representations relies a lot on this interplay dubbed as duality transformations in section~\ref{sec:necescond}. We therefore allow for arbitrary Borel decompositions from the very beginning, in terms of parity $p$ and conjugation $c$ gradings assigned to $\sl(\groupn|\groupm)$ algebra generators in section~\ref{sec:defs}. A  Kac-Dynkin-Vogan diagram is associated to each choice of these two gradings.

Finally we revisit in section~\ref{sec:defs} another supersymmetric feature that is known to specialists but is less known to wider community: presence of fermionic generators significantly reduces the variety of possible unitary representations. Namely, the most general real form of $\sl(\groupn|\groupm)$ with $\groupn,\groupm\neq 0$ that admits nontrivial unitary representations is $\su(\groupp,\groupq|\groupm)$. Moreover these representations are necessarily of highest-weight type in an appropriate choice of grading. So, in fact, study of unitary duals is simpler in supersymmetric case than in a non-supersymmetric case, despite the more complicated structure of the superalgebra.

Section~\ref{sec:necescond} is devoted to derivation of the necessary conditions for a representation to be unitary. These conditions are derived from comparison of norms between all possible highest-weight vectors that are linked by duality transformations. It is easy to encode these conditions on the square lattice, so we call them plaquette constraints.  These constraints for the case of integer weights are naturally related to the Young diagrams -- objects originating from supersymmetric version of Schur-Weyl duality and which were well explored in the literature for compact case, see e.g. \cite{ChenWangBook} and references therein. The plaquette constraints are also very informative about shortenings in supersymmetric multiplets, as is discussed in section~\ref{sec:shortening}. While these necessary unitarity constraints are simple to derive, they are very restrictive. Later on we conclude that they are also sufficient. Also, as explained in appendix~\ref{sec:appa}, we can use a supersymmetric extension trick combined with these constraints to derive the necessary (and, by inspection, sufficient) conditions for  unitarity of highest-weight representations of ordinary non-compact Lie algebras $\su(\groupp,\groupq)$, alternatively to how this case was studied historically.

In section~\ref{sec:repos} we develop a novel approach that allows one to use oscillators for constructing representations with non-integer weights. This is not the first proposal of how to use oscillators to this end, but it is a different one from what was used before, as we outline in the conclusion section.  We propose to keep a simple bilinear realisation of the symmetry algebra, e.g. $\gE_{ij}=a_{i}^{\dagger}\cdot a_{j}$ for $\gl(\groupn)$ etc.  However, we consider that certain oscillators act on a deformation of Fock space $\CF_{\gamma}$, where $\gamma$ controls non-integrality of the weight. To get several continuous parameters, one can design a certain tensor product of deformed Fock spaces.

Although $\CF_{\gamma}$ has a non-positive Hermitian norm, we can identify an important subspace in it which is positive-definite. The quest in proving unitarity  becomes elementary if we are able to prove that our Lie algebra representation is confined to this subspace. Section~\ref{sec:classification} is devoted to explicit oscillator-based constructions which demonstrate that indeed all representations singled out by the plaquette constraints can be realised inside the positive-definite subspace of $\CF_{\gamma}$. Note that we use the same oscillator realisations of Lie algebra as in \cite{Gunaydin:1984fk,Gunaydin:1998jc,Gunaydin:1998sw}. Simply changing the representation space of the oscillator algebra allows us to get now all possible unitary representations.

We finalise section 5 with introduction of a new combinatorial object -- non-compact Young diagram. In contrast to the weight labelling of a unitary representation, which is sensitive to the choice of Borel decomposition, Young diagram is an invariant object that can be used for labelling purposes. For the case of integer-weighted representations, the non-compact Young diagram can be extended. The extended Young diagrams do not depend on the choice of the Lie algebra anymore but the same diagram can be used to describe a unitary representation of $\su(\groupp,\groupq|\groupm)$ with various values of $\groupp,\groupq,\groupm$. This is analogous to the fact that ordinary Young diagrams describe representations of $\su(\groupn)$ with arbitrary value of $\groupn$ larger or equal to the diagram's height.

We also added three  appendices to this paper. Appendix A collects some technical add-ons to the main discussion, namely it discusses the necessary conditions for a highest-weight representation of $\mathfrak{su}(\groupp,\groupq)$ using a supersymmetric extension trick and the structure of monomial shortenings. The other two appendices are devoted to the  application of the developed ideas to the physically most important case --  the representations of four dimensional conformal algebra $\su(2,2)$ and its supersymmetric extension $[\mathfrak{p}]\su(2,2|4)$. More specifically,  we show in appendix~\ref{sec:confa} how the generalised oscillator method yields very simply all the positive energy unitary representations of $\su(2,2)$ as classified by Mack \cite{Mack:1975je} including those with anomalous dimensions. The representations of its supersymmetric extension $[\mathfrak{p}]\su(2,2|4)$, including those with continuous (anomalous) weights  are analysed in appendix~\ref{sec:superconformal}.

\section{\label{sec:defs}Definitions and basic properties}
The main focus of this paper is classification of unitary representations for all real forms of $\sl(\groupn|\groupm)$ algebra~\footnote{Classification for $\mathfrak{psl}(\groupn|\groupn)$ case is  obtained by imposing a constraint on representations of  $\sl(\groupn|\groupn)$, see Section \ref{sec:supqm}.}. In this section we introduce the basic definitions of real forms and unitarity, show that only highest-weight representations (UHW) or lowest-weight representations (ULW) are possible,  and provide a detailed account of Kac-Dynkin-Vogan diagrams, duality transformations and their usage for describing UHW/ULW modules.


\subsection{Parity grading}

One can view many properties of $\sl(\groupn|\groupm)$  as coming from the ones of  $\glnm$, and we will  use this fact extensively. The Lie superalgebra $\gl(\groupn|\groupm)$ is defined as follows. It is spanned by the generators $\gE_{ij}$, with indices $i,j\in\overline{1,\groupm+\groupn}$. One introduces a parity grading function $p$ which assigns a $\mathbb{Z}_2$-number for each index: $p_i=\overline 0$ or $\overline 1$. Let $\groupn$ be the number of $i$-s with $p_i={\overline 0}$ and $\groupm$ be the number of $i$-s with $p_i={\overline 1}$.

The  generator $\gE_{ij}$ is odd  if  $p_i+p_j=\overline{1}$ and it is even if $p_i+p_j=\overline{0}$. The Lie bracket of the $\glnm$ superalgebra is defined as
\be\label{comrel1}
        [\gE_{ij},\gE_{kl}]=\delta_{jk}\,\gE_{il}-(-1)^{(p_i+p_j)(p_k+p_l)}\delta_{li}\,\gE_{kj}\,.
\ee
Since parity of the generators depends only on  the sums $p_i+p_j$,  one can shift all $p$'s by $\overline 1$ with no effect on \eqref{comrel1}. In other words, $\gl(\groupn|\groupm)$ and $\gl(\groupm|\groupn)$ are isomorphic.

In principle, it is always possible to make the so-called distinguished choice for the parity grading: $p_{i}={\overline 0}$ for $1\leq i\leq \groupn$,  $p_i={\overline 1}$ for $\groupn+1\leq i \leq \groupn+\groupm$. However, flexibility in the choice of $p$ will play an important role in this article, hence we  keep $p$ to be general.

Kac-Dynkin diagram is an alternative equivalent way to describe the parity distribution. The diagram is a set of $\groupn+\groupm-1$ nodes arranged in a line~\footnote{For arbitrary Lie algebra, Kac-Dynkin diagram is introduced as the graph determining lengths of simple roots and angles between them. For the considered case of $A$-series we use a simplified language which is sufficient for our goal.}. Its $i$-th node is denoted by a crossed circle if the generator $\gE_{i,i+1}$ is odd and it  is denoted by a blank circle  if the generator $\gE_{i,i+1}$ is even.
\\ \ \\
{\it Example 1:} The following  Kac-Dynkin diagram for $\sl(2|3)$ algebra
\be\label{KD1}
\raisebox{-.3em}{\includegraphics[width=0.3\textwidth]{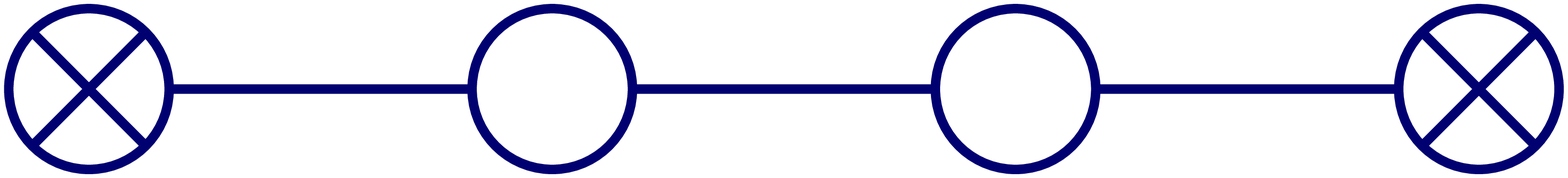}}
\ \  \ \
\ee
 corresponds to the choice $p_i={\overline 0}$ for $i\in {1,5}$ and $p_i={\overline 1}$ for $i\in {2,3,4}$. To reflect this particular choice of the parity distribution, we may denote $\sl(2|3)$ algebra with the Kac-Dynkin diagram (\ref{KD1}) as $\sl(1|3|1)$.
 \\ \ \\
 {\it Example 2:}  $\sl(2|3)$, understood as the notation for parity distribution $p_i={\overline 0}$ for $i\in {1,2}$ and $p_i={\overline 1}$ for $i\in {3,4,5}$, corresponds to the diagram
\be\label{KD2}
\raisebox{-.3em}{
\includegraphics[width=0.3\textwidth]{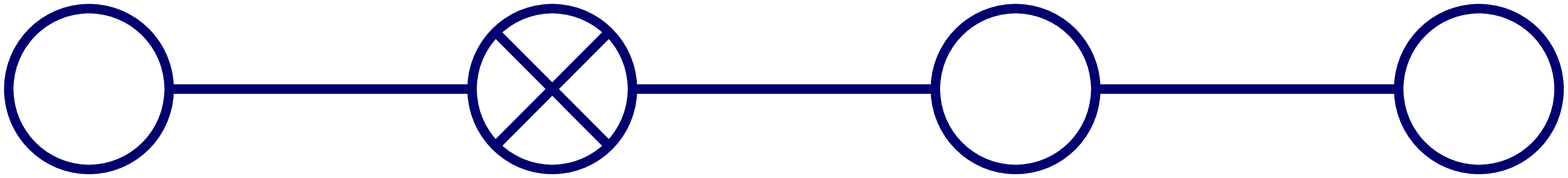}
}
\,.
\ee
In the following, it will be clear from the context whether $\sl(\groupn|\groupm)$ refers only to a choice of the algebra or also to the choice of parity distribution.

\vspace{1em}
\noindent The Cartan subalgebra of $\gl(\groupn|\groupm)$ is spanned by $\gE_{ii}$. To describe the Cartan subalgebra of $\sl(\groupn|\groupm)$, we use the generators
\be
        \ch_{i,j}&\equiv& \gE_{ii}-(-1)^{p_i+p_j}\gE_{jj}\,,\\
        \ch_i&\equiv& \ch_{i,i+1},
\ee
with $\ch_i$, $i\in\overline{1,\groupn+\groupm-1}$ forming a basis.

For unitary representations, one can always  diagonalise the action of Cartan subalgebra on the representation space. Then each of the eigenstates is naturally assigned the weight which is the set of the eigenvalues of the Cartan generators. For the $\gl(\groupn|\groupm)$ weight (eigenvalues of $\gE_{ii}$), we use notation
\be\label{fundaweight}
[m_1,m_2,\ldots,m_{\groupm+\groupn}]\,,
\ee
and for the $\sl(\groupn|\groupm)$ weight (eigenvalues of $\ch_i$) we use notation
\be\label{weight}
\langle\o_1,\ldots,\o_{\groupn+\groupm-1}\rangle\,.
\ee
 $\omega_i$ shall be called odd if $\gE_{i,i+1}$ is an odd generator and even otherwise.

Finally, we introduce the Borel decomposition as
\be\label{boreldecomposition}
        \gl=\mathfrak{n}^-\oplus\algh\oplus\mathfrak{n}^+,
\ee
with $\mathfrak{n}^+={\rm span}\{\gE_{ij},i<j\}$, $\algh={\rm span}\{\gE_{ii}\}$,
$\mathfrak{n}^-={\rm span}\{\gE_{ij},i>j\}$.

A  representation is said to be of the highest-weight type if there exists a vector $\HWS$ (highest-weight state) for which the condition
\be
\mathfrak{n}^+\HWS=0
\ee
is satisfied. Similarly, with the  condition $\mathfrak{n}^-\LWS=0$ one defines  a representation of the lowest-weight type. The highest-weight state, if it exists, is unique for an irreducible unitary representation. Then its weight can be used to label the representation and it is called the weight of the representation. The weight \eqref{fundaweight} of $\HWS$ is called the fundamental weight and $\eqref{weight}$ is called the Dynkin weight, $\omega_i$'s are also known as the Dynkin labels.

Note however that the sub-algebras $\mathfrak{n}^{\pm}$  depend on the choice of the parity distribution, the notion of the highest and lowest weights also depends on this choice. Therefore, labelling of a representation by its weight is only meaningful  in conjunction with prescription what particular parity function $p$ is used.

\subsection{\label{sec:realforms}Real form and Hermitian basis}
A real Lie superalgebra $\algg_0$ is a Lie superalgebra over $\mathbb R$. In particular, all its structure constants are real. $\algg_0$ is called a real form of complex superalgebra $\algg$ if $\algg$ is isomorphic to $\mathbb{C}\otimes_{\mathbb R}\algg_0$.

Hermitian basis  of a complex Lie superalgebra $\algg$ is chosen such that the commutation relations have the form
\be\label{HermCom}
        [\mathfrak{b}_i,\mathfrak{b}_j]=\ii\,f_{ij}^k \mathfrak{b}_k\,,\ \ [\mathfrak{b}_i,\mathfrak{f}_\a]=\ii\,c_{i\a}^\b\mathfrak{f}_\beta\,,\ \ [\mathfrak f_\a,\mathfrak f_\b]=h_{\a\b}^i\mathfrak{b}_i\,,
\ee
where $\ii=\sqrt{-1}$ and $f_{ij}^k$, $c_{i\a}^\b$, $h_{\a\b}^i$ are real numbers. Here $\mathfrak{b_i}$ are  even generators and $\mathfrak{f_\a}$ are odd generators.

If we choose the generators  $e^{\ii\pi/2}\mathfrak b_i$ and $e^{\ii\pi/4}\mathfrak{f}_\a$  as a basis in $\algg$ then all structure constants become real. Hence we established a bijective way to associate  a real form to a Hermitian basis, and we will speak about the Hermitian basis of a given real form in this sense. This bijection is not the only possible one,  but we need to pick some bijection to remove ambiguity in the follow-up discussions.

 \subsection{\label{sec:defunitary}Definition of unitarity of a representation}
We begin by recalling the definition in the non-supersymmetric case. A unitary representation $T$ of a real Lie group  is a representation which has a  positive-definite
invariant Hermitian form $\langle\cdot,\cdot\rangle$ in the representation space $V$. Invariance means that the Lie group acts on $V$ by unitary operators, hence the name for such representations.

We will consider unitarity on the level of algebras only~\footnote{Integration to the level of group can  in principle be done, but this would require a  discussion of when to introduce universal covers of groups, and, later on, how to  define supermanifolds. These technical details do not contribute to the statement of the classification theorem and hence we omit their study.}. The above definition is reformulated on the level of algebra as follows

\paragraph{Definition 1.}   A representation $T$ of a real  Lie algebra $\algg_0$ is called unitary if one can define a positive-definite Hermitian form $\langle\cdot,\cdot\rangle$ in the representation space $V$ such that generators $\mathfrak{b}$ in Hermitian basis of  $\algg_0$ are represented by Hermitian operators  in $V$:
\be
 T(\mathfrak{b})=T(\mathfrak{b})^\dagger \,,
\ee
where $\mathcal{O}^\dagger $ means  the Hermitian conjugate of $\mathcal{O}$ with respect to $\langle\cdot,\cdot\rangle$.

This justifies the name for the Hermitian basis introduced above. Note that the property of hermiticity is consistent with commutation relations~\footnote{At this point we speak only about even generators and hence only about first commutator in  \eqref{HermCom}.}.

Working with real forms of Lie algebras is not always convenient. Hence it is important to determine how  the properties of $T$ are  extended to complexification $\algg\simeq\mathbb{C}\otimes_{\mathbb R}\algg_0$. One can introduce a $^*$-operation on $\algg$ as follows: It acts as $\mathfrak{b}^*=\mathfrak{b}$ on generators of the Hermitian basis and then it is continued to any generator by anti-linearity, i.e. $(c\,\mathfrak{j})^*=\overline{c}\,\mathfrak{j}^*$, where $c$ is a complex number, $\overline{c}$ is its complex conjugate, and $\mathfrak{j}\in\algg$. One can check that  the $^*$-operation is an anti-isomorphism of the Lie algebra $\algg$: $[\mathfrak{j},\mathfrak{j}']^*=-[\mathfrak{j^*},\mathfrak{j^*}']$.

The reason for introducing such an operation is the following property:   $T(\mathfrak{j}^*)=T(\mathfrak{j})^\dagger $. Hence $^*$-operation  is a formal "Hermitian conjugation" at the level of a Lie algebra.

In general, an anti-linear anti-isomorphism on a complex Lie algebra shall be called  a $^*$-operation. Complex Lie algebra endowed  with a $^*$-operation is called a $^*$-algebra.
Now we can reformulate the notion of unitary representation in a more convenient form:

\ \\
\noindent {\bf{Definition 2.}} Representation $T$ of  a $^*$-algebra is called unitary if one can construct a positive-definite Hermitian form $\langle\cdot,\cdot\rangle$ on the representation space such that for all generators $\mathfrak{j}$
\be
        T(\mathfrak{j}^*)=T(\mathfrak{j})^\dagger \,.
\ee

\ \\
\noindent Note that one can use the $^*$-operation to define a real form of the Lie  algebra. Indeed,  Hermitian basis of a real form is spanned by generators satisfying $\mathfrak j^*=\mathfrak j$. Therefore  both  definitions of  unitary representations given above are equivalent.\\[1em]

\noindent The  second definition of unitary representations generalises in a  straightforward manner to the case of Lie superalgebras: We introduce a $^*$-operation as an operation that  has the basic  properties of Hermitian conjugation. Namely, it is antilinear and  "commutes" with the Lie superalgebra bracket as follows:
\be
        [\mathfrak j_1,\mathfrak j_2]^*=-(-1)^{p(\mathfrak j_1)p(\mathfrak j_2)}[\mathfrak j_1^*,\mathfrak j_2^*].
\ee
Then, we define a unitary representation  of a $^*$-superalgebra by using literally the Definition 2 from above.

Let us now analyse the construction that uses the first-type definition. To this end one has first to identify a real form corresponding to the  $^*$-algebra. Superficially  there are two real forms that can be constructed: The one having Hermitian basis with generators satisfying $\mathfrak{b}^*=\mathfrak{b},\ \mathfrak{f}^*=\mathfrak{f}$,
and the other having Hermitian basis with generators satisfying $\mathfrak{b'}^*=\mathfrak{b'},\
\mathfrak{f'}^*=-\mathfrak{f'}$ \footnote{A mixed case with different sign choices for different odd generators  is impossible; it would be inconsistent with commutation relations.}.  However, both real forms are isomorphic. Indeed, $\gl(\groupn|\groupm)$ enjoys the outer automorphism
\be\label{outer}
        \varphi_{\rm out}:&\gE_{ij}\mapsto -\gE_{ji}\,,&\ \ p_i+p_j=0\,\no\\
        \varphi_{\rm out}:&\gE_{ij}\mapsto \ii\, \gE_{ji}\,,&\ \ p_i+p_j
=1\,,
\ee
whereas the two Hermitian bases are related by a complex transformation $\mathfrak{f}\to\ii\,\mathfrak{f}$. Combining the latter with the outer automorphism, one gets a real transformation relating the two real forms.

Therefore,  the $^*$-algebra defines a unique, up to an isomorphism, real form. This statement is similar to the one from non-supersymmetric case, however uniqueness emerges for a more subtle reason which relies on \eqref{outer}. We now note that \eqref{outer} is an automorphism of $\glnm$ Lie algebra but it does not preserve the $^*$-operation. Therefore, while both real forms defined above are isomorphic, the representations subduced on them are not isomorphic. Indeed, in one case odd generators are represented by Hermitian operators. In the other case they are represented by anti-Hermitian operators.

\ \\
{\it Example:} Consider a compact real form of $\gl(1|1)$ algebra defined by $\gE_{i,j}^*=\gE_{j,i}$ ($i,j\in \overline{1,2}$, $p_0={\overline 0}$, $p_1={\overline 1}$). Two Hermitian bases, distinguished by $\sigma=\pm 1$, are given by
\be
        \mathfrak{b}_+=\gE_{11},\ \mathfrak{b}_-=\gE_{22},\ \mathfrak{f}_1=\gE_{1,2}+\sigma\, \gE_{2,1}\,\ \mathfrak{f}_2=\ii\,(\gE_{1,2}-\sigma\, \gE_{2,1})\,.
\ee
One has $\mathfrak{\mathfrak{f}}_\a^*=\sigma\,\mathfrak{f}_\alpha$. Explicit commutation relations are
\be
        &&[\mathfrak{b}_\pm,\mathfrak{f}_1]=- \ii\,\sigma\,\mathfrak{f}_2,\ [\mathfrak{b}_{\pm},\mathfrak{f}_2]= \ii\,\sigma\, \mathfrak{f}_1,\ \ [\mathfrak{f}_1,\mathfrak{f}_2]=0,\no\\
        &&[\mathfrak{f}_1,\mathfrak{f}_1]=[\mathfrak{f}_2,\mathfrak{f}_2]=2\sigma\,(\mathfrak{b}_++\mathfrak{b}_-)
\ee
The real transformation $\mathfrak{f}_1\to \mathfrak f_{2},\ \mathfrak{f}_2\to \mathfrak f_1, \mathfrak{b_\pm}\to-\mathfrak{b_\pm}$ establishes isomorphism between these two Hermitian bases, but their $^*$-structure remains different.\\[1em]

We can now introduce a refined definition of a unitary representation similar to the Definition 1 from above.

\paragraph{Definition 1 (supersymmetric case).} A representation $T$ of a real Lie superalgebra $\algg_0$ is called unitary if one can define a positive-definite Hermitian form $\langle\cdot,\cdot\rangle$ in the representation space such that
\begin{subequations}
\be
T(\mathfrak{b})^{\dagger }=T({\mathfrak{b}})\,,\ \   T(\mathfrak{f})^{\dagger }=+T({\mathfrak{f}})\,,
\ee
or
\be
T(\mathfrak{b})^{\dagger }=T({\mathfrak{b}})\,,\ \  T(\mathfrak{f})^{\dagger }=-T({\mathfrak{f}})\,.
\ee
\end{subequations}
for all bosonic $\mathfrak{b}$ and fermionic $\mathfrak{f}$ generators of its Hermitian basis.

Hence unitary representations in the sense of this definition split in two classes. There is no possible isomorphism between representations of these two classes (i.e. there is no unitary transformation of the representation space that maps a representation from one class to a representation from the other class). However   the outer automorphism \eqref{outer} allows one to generate  one class of unitary representations from the other class, in the sense that representation $T_{\varphi}\equiv T\cdot \varphi_{\rm out}$  belongs to the opposite class than representation $T$ itself.

Each $^*$-algebra representation can be used to define the real Lie algebra representation belonging to either of the classes, depending on how we decide to construct the real form. In the following, we will classify  unitary representations with respect to all possible $^*$-superalgebras of $\sl(\groupn|\groupm)$. The unitary representations of corresponding real form, which come in pairs, can be easily subduced afterwards.

\subsection{Real forms that admit only trivial unitary representations}
\label{sec:otherreal}
All real forms of $\sl(\groupn|\groupm)$ algebra were classified in \cite{Parker:1980af}. They are: $\su(\groupp,\groupn-\groupp|\groupr,\groupm-\groupr)$, $\sl(\groupn,\mathbb{R}|\groupm,\mathbb{R})$ and  $\su^*(2\groupn|2\groupm)$. Additionally, in the case $\groupn=\groupm$, one has an extra real form  $\psl'(\groupn|\groupn)$ of complex algebra $\psl(\groupn|\groupn)$. Among them, only the first real form allows non-trivial unitary representations, and even in this case there are severe restrictions which will be discussed in section~\ref{sec:allarehw}.

In this subsection, we prove why other three forms admit only trivial unitary representations. The reader may skip this part, it will not be required for the rest of the paper.

It suffices to consider the distinguished choice of the parity grading till the end of this subsection. First, we list the $^*$-algebra structure for these cases\footnote{In \cite{Parker:1980af}, all real forms are given in terms of an involutive semimorphism $C$ which is simply related to the $^*$-operation:  $\mathfrak{b}^*=-C\,{\mathfrak{b}}$ for even generators, and  $\mathfrak{f}^*= \ii\,C\,\mathfrak{f}$ for odd generators.}:
\begin{itemize}
\item $\sl(\groupn,\mathbb{R}|\groupm,\mathbb{R})$ is defined by
\be
\gE_{ij}^{*}=d_{p_i+p_j}\, \gE_{ij}\,,
\ee
where $d_{\bar 0}=1$ and $d_{\bar 1}=\ii$.

The even subalgebra of this form is $\sl(\groupn,\mathbb{R}) \oplus \sl(\groupm,\mathbb{R})\oplus \mathbb{R}$.
\item To define $\su^*(2\groupn|2\groupm)$, introduce an extra grading function $c_i$. Its values are: $c_{\mu}=\bar 0$ for $1\leq \mu\leq \groupn$ and $c_{\mu}=\bar 1$ for $\groupn <\mu\leq 2\groupn$;  $c_{a}=\bar 0$ for $1\leq a\leq \groupm$ and $c_{a}=\bar 1$ for $\groupm <a \leq 2\groupm$. Then
\begin{align}
\gE_{\mu\wb}^*&=(-1)^{c_{\mu}+c_{\wb}}\gE_{\mu+\groupn,\wb+\groupn}\,,
&
\gE_{\mu a}^* &=\ii\,(-1)^{c_{\mu}+c_{a}}\gE_{\mu+\groupn,a+\groupm}\,,
\\
\gE_{ab}^*&=(-1)^{c_{a}+c_{b}}\gE_{a+\groupm,b+\groupm}\,,
&
\gE_{a\mu}^* &=\ii\,(-1)^{c_{\mu}+c_{a}}\gE_{a+\groupm,\mu+\groupn}\,,
\end{align}
where addition of $\groupn$ (resp. $\groupm$) is understood modulo $2\groupn$ (resp. $2\groupm$).

The even subalgebra of this form is  $\su^*(2\groupn) \oplus \su^*(2\groupm)\oplus \mathbb{R}$.

\item $\psl'(\groupn|\groupn)$ is defined by
\be
\gE_{ij}^*=d_{p_i+p_j}\,\gE_{i+\groupn,j+\groupn}\,,
\ee
where addition of $\groupn$ is understood modulo $2\groupn$.

The even subalgebra of $\psl'(\groupn|\groupn)$ is $\sl(\groupn,\mathbb{C})$ (realification of complex Lie algebra $\sl(\groupn)$).

\end{itemize}

For all the three discussed forms one can find an odd generator in $\mathfrak{n}^+$ and an odd generator in $\mathfrak{n}^-$ that have the property $\mathfrak{f}\,\mathfrak{f}^*+\mathfrak{f}^*\,\mathfrak{f}=0$. In fact, for $\sl(\groupn,\mathbb{R}|\groupm,\mathbb{R})$ and $\su^*(2\groupn|2\groupm)$ all the odd generators enjoy this property. Then, for any vector $|v\rangle$ of the representation space, one computes
\be
0=\langle v|\mathfrak{f}\,\mathfrak{f}^*+\mathfrak{f}^*\,\mathfrak{f}|v\rangle=|| \mathfrak{f}|v\rangle ||^2+|| \mathfrak{f}^*|v\rangle ||^2\,.
\ee
Hence $|| \mathfrak{f}|v\rangle ||=0$ but, as one does not allow zero norms in a unitary representation, $\mathfrak{f}|v\rangle=0$. Hence $\mathfrak{f}$ is represented by $0$ on the whole representation space. On the other hand, odd subspaces of $\mathfrak{n}^+$ and $\mathfrak{n}^-$ form irreducible representations under adjoint action of the even sub-algebra.  Hence the property $\mathfrak{f}|v\rangle=0$ should hold for all odd generators. Finally, since by commuting odd generators one can generate the whole algebra we conclude that the whole algebra is represented by zero, that is the representation is trivial.

\subsection{c-grading}
We will consider only the real forms $\su(\groupp,\groupn-\groupp|\groupr,\groupm-\groupr)$ for the remainder of the article as only they can lead to non-trivial unitary representations. We commence by an explicit construction of a $^*$-algebra.
Define a $c$-grading by assigning a  $\mathbb{Z}_2$-number $c_i={\overline 0}$ or ${\overline 1}$ for $i\in\overline{\\ 1,\groupm+\groupn}$. With respect to the  $c$-grading, we call a generator $\gE_{ij}$   $c$-even  if
$c_i+c_j={\overline 0}$ and $c$-odd  if $c_i+c_j={\overline 1}$. Then we define the $^*$-operation  by
\be\label{defstar}
        \gE_{ij}^*=(-1)^{c_i+c_j}\gE_{ji}\,.
\ee
We therefore introduced a new structural function $c$, in addition to the parity decomposition function $p$.  Define four invariants $n_{\a,\b}$, with $\alpha,\beta\in\mathbb{Z}_2$, as the number of $i$-s with $p_i=\a$ and $c_i=\b$. Let us label them by
\be
n_{\overline{0},\overline{0}}=\groupp\,,\ n_{\overline{0},\overline{1}}=\groupq\,,\ n_{\overline{1},\overline{1}}=\groupr\,,\  n_{\overline{1},\overline{0}}=\groups\,.
\ee
The real form defined by \eqref{defstar} is $\su(\groupp,\groupq|\groupr,\groups)$.  In particular $\groupp+\groupq=\groupn$ and $\groupr+\groups=\groupm$.

Consider any permutation $\sigma$ of numbers from the set $\{1,2,\ldots,\groupn+\groupm\}$. It defines a natural isomorphism
\be\label{permiso}
E_{ij}'= E_{\sigma(i)\sigma(j)},\ \ p_i'= p_{\sigma(i)},\ \ c_i'=c_{\sigma(i)}
\ee
between *-algebras with the same $n_{\alpha,\beta}$. This explains why we use only invariants ${n_{\alpha,\beta}}$ to label real forms associated with different $p$- and $c$-gradings.

Additionally, since only sums $c_i+c_j$ and $p_i+p_j$ are relevant for the commutation relations and the $^*$-operation, one has  extra identifications obtained by the shifts $c_i\to c_i+\overline{1}$ and $p_i\to p_i+\overline{1}$:
\be\label{genericisoint}
 \su(\groupp,\groupq|\groupr,\groups)=\su(\groupq,\groupp|\groups,\groupr)=\su(\groups,\groupr|\groupq,\groupp)\,.
\ee
Finally, there is also the Lie algebra outer automorphism (\ref{outer}) which results in the change of $c$-parity of $p$-odd generators only:
\be\label{genericisoout}
 \su(\groupp,\groupq|\groupr,\groups)\stackrel{\text{out}}{\simeq}\su(\groupp,\groupq|\groups,\groupr)\,.
\ee

\subsection{Kac-Dynkin-Vogan diagrams}
Similarly to the case of $p$-grading, we will exhibit the $c$-grading on the Dynkin diagrams. The $i$-th Dynkin node will be called $c$-odd if $c_i\neq c_{i+1}$; we will denote such a node by an extra circle around it. This diagrammatic representation is a straightforward generalisation of Vogan's labeling of  real forms of Lie algebras to the supersymmetric case~\footnote{We use double circles instead of typically used black nodes so as  not to interfere with crossed notation for $p$-odd nodes.}. For original definitions of Vogan diagrams see e.g. \cite{KnappBook}.

If the $i$-th Dynkin node is $c$-odd  and $p$-even then the corresponding rank-1 sub-algebra (spanned by $\ch_i,\gE_{i,i+1}, \gE_{i+1,i}$) is the non-compact $\su(1,1)$. The $p$-odd node always corresponds to the compact $\su(1|1)$, independently of $c$-parity, due to \eqref{outer}. Finally, the $p$-even and $c$-even case corresponds to $\su(2)$.

\ \\
{\it Example 1:} The diagram
\be
\label{KD3}
\raisebox{-.3em}{
\includegraphics[height=1.5Em]{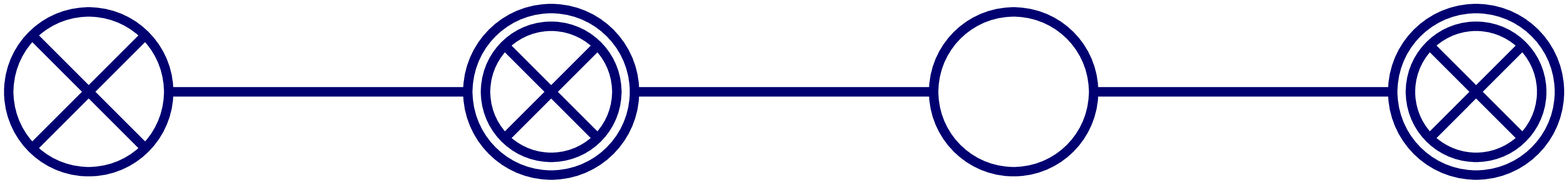}
}
\ee
denotes the following choice of grading functions for the real form $\su(2,1|2)\equiv \su(2,1|2,0)$:
$$p_1=p_3=p_4=\overline{0}\,,\ p_2=p_5=\overline{1}\ \ {\rm  and}\ \ c_1=c_2=c_5=\overline{0}\,,\ c_3=c_4=\overline{1}\,.$$ To explicitly reflect both the $c$  and $p$  grading, we may denote this diagram as $\su(1|1,\!|2,\!|1)$.

\ \\
\noindent{\it Example 2:} The non-supersymmetric case can be also studied as an interesting sub-case. The diagram
\be
\raisebox{-.3em}{
\includegraphics[height=1.5Em]{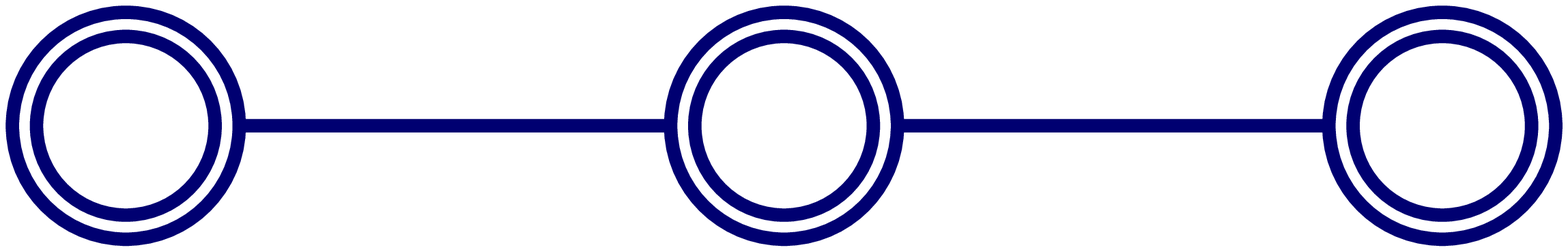}
}
\ \  \ \
\ee
denotes the following choice of grading functions for the real form $\su(2,2)\equiv \su(2,2|0,0)$:
$$p_1=p_2=p_3=p_4=\overline{0}\ \ {\rm and}\ \ c_1=c_3=\overline{0}\,,\ c_2=c_4=\overline{1}\,.$$ To explicitly reflect both gradings  we may denote this diagram as  $\su(1,1,1,1)$\,.

It is useful to consider the extended  diagrams \cite{VdJ} which are obtained by adding the node corresponding to the sub-algebra  $\{\ch_{\groupn+\groupm}\equiv \gE_{\groupn+\groupm,\groupn+\groupm}-(-1)^{p_{\groupn+\groupm}+p_1}\gE_{11}, \, \gE_{\groupn+\groupm,1}, \, \gE_{1,\groupn+\groupm}\}$. The $p$- and $c$-grading of this additional node follow from the obvious property that number of $c$-odd and number of $p$-odd nodes of extended Dynkin diagram are even.

For the two examples  above the extended diagrams are
\be\label{KD3}
\raisebox{-.7em}{
\includegraphics[height=3EM]{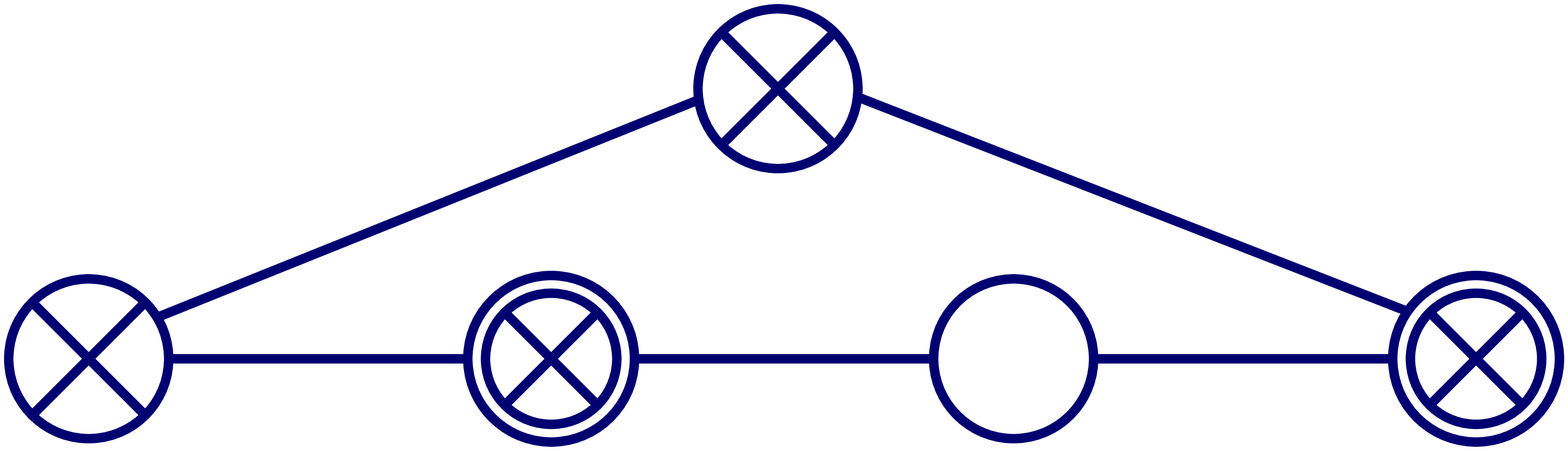}
}
\hspace{2em}
{\rm and}
\hspace{2em}
\raisebox{-.7em}{
\includegraphics[height=3EM]{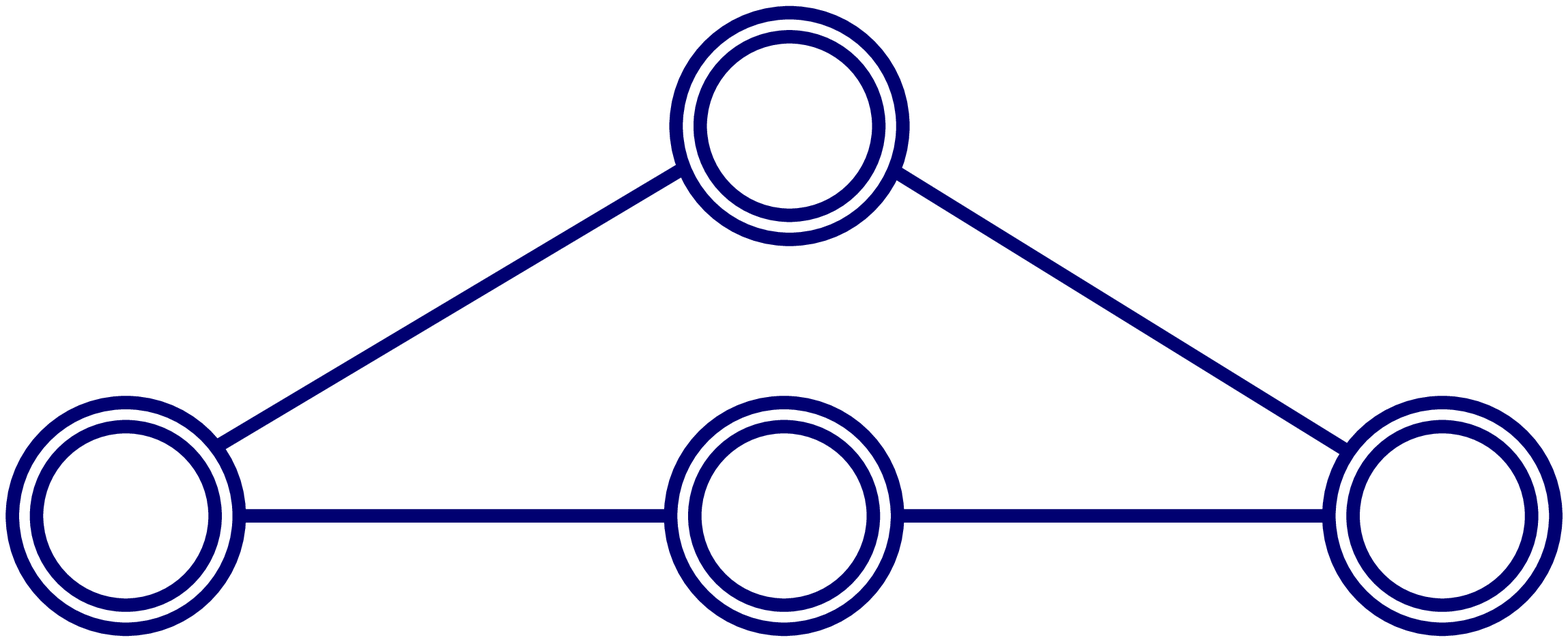}
}\,.
\ \  \ \
\ee

We now discuss how isomorphisms of type \eqref{permiso} are realised on the diagrams. A useful option is to view them as compositions of the elementary ones with  $\s_{(i,i+1)}:\{1,2,\ldots,i,i+1,\ldots,\groupn+\groupm\}\mapsto :\{1,2,\ldots,i+1,i,\ldots,\groupn+\groupm\}$. Since these permutations interchange the values of $p_i$ and $p_{i+1}$ it is easy to see what happens with the $p$-grading of the Kac-Dynkin-Vogan diagram: If the $i$-th node is $p$-even then the $p$-grading is unaffected; but if the $i$-th node is $p$-odd then the $p$-gradings of the nodes at positions $i\pm 1$ change to their opposites. The $c$-grading of the diagram obeys the same rules.

For $p$-even node the isomorphism defined by $\sigma_{(i,i+1)}$ is nothing but the Weyl reflection. Sometimes one also speaks about "odd" Weyl reflections \cite{LeSaSe,Frappat1989} in the case of $p$-odd node. The elementary isomorphisms are called duality transformations in the literature on quantum integrability. The duality transformations on $\su(1|1)$ or $\su(2)$ node map a set of Bethe Ansatz equations to an equivalent but generically not identical set of equations \cite{Woynarovich,Bares,Pronko:1998xa}.

Other useful remarks apply to two specific permutations: Cyclic permutation $\sigma=(123\ldots(\groupn+\groupm))$ applied several times can be used to declare any node as the extended one. Direction reversal permutation $\sigma=(1(\groupn+\groupm))(2(\groupn+\groupm-1))\ldots$ can be thought as a possibility to  read the diagrams in both clockwise and counterclockwise direction.

\ \\
{\it Example:} Let us illustrate this discussion by demonstrating two isomorphisms: $\su(1|1,\!|2,\!|1)\simeq \su(2,1|2)$ and $\su(1,1,1,1)\simeq \su(2,2)$. First, one considers extended diagrams for $\su(1|1,\!|2,\!|1)$ and $\su(1,1,1,1)$, as was already done in \eqref{KD3}. Then one performs the duality transformation on the left-most node and declare the right-most node as the extended one:
\be\label{KD4}
\raisebox{-.5em}{
\includegraphics[width=0.8\textwidth]{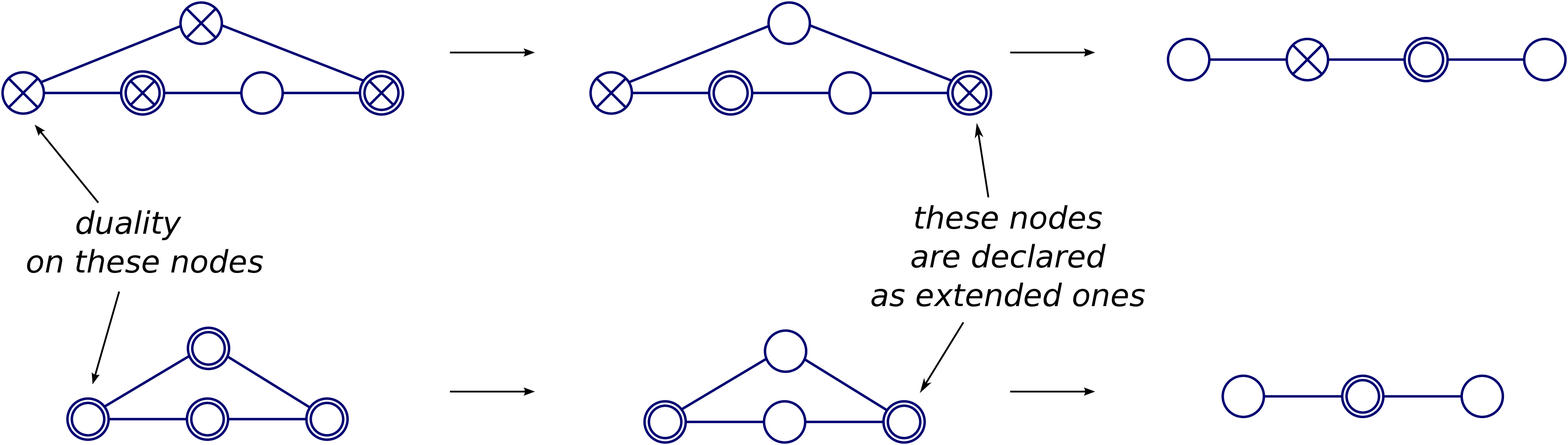}
}\,.
\ee
In the first case, one gets the $\su(2,1|2)$ grading if one reads the Kac-Dynkin-Vogan diagram from right to left. In the second case, one gets the $\su(2,2)$ as required.\\[1em]

\noindent It is a simple combinatorial exercise to show that by isomorphisms \eqref{permiso} one can always bring the extended diagram to a form with at most two $c$-odd and at most two $p$-odd nodes ($p$-odd and $c$-odd nodes may coincide). All possible non-equivalent cases are summarized in Fig.~\ref{fig:possiblerealform}.

\begin{figure}[t]
\begin{centering}
\begin{tabular}{|c|c|c|c|c|c|}
\hline
\parbox{6.4em}{\ \ \includegraphics[width=4.32em]{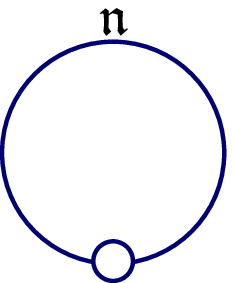}} &
\parbox{6.4em}{\includegraphics[width=6.24em]{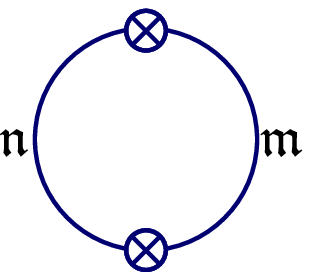}} &
\parbox{6.24em}{\includegraphics[width=6.24em]{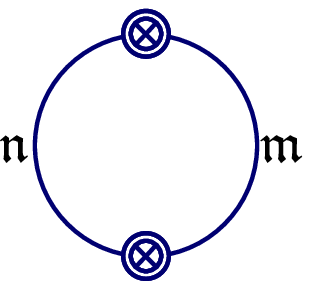}} &
\parbox{6.4em}{\includegraphics[width=6.24em]{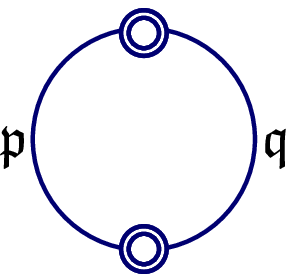}} &
\parbox{6.24em}{\includegraphics[width=6.24em]{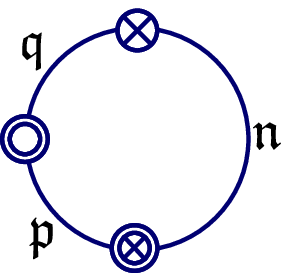}} &
\parbox{6.24em}{\includegraphics[width=6.24em]{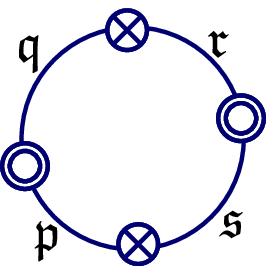}} \\\hline
\!$\su(\groupn)$\! & \!$\su(\groupn|\groupm)$\! & \!$\su(\groupn,\!|\groupm)$\!
&  \!$\su(\groupp,\groupq)$\! & \!$\su(\groupp,\groupq|\groupn)$\! & \!$\su(\groupp,\groupq|\groupr,\groups)$\!
\\ \hline
\end{tabular}
\caption{\label{fig:possiblerealform}List of $^*$-algebras not related by
isomorphisms \eqref{permiso}. Number $\groupk$  over the line represents
$\groupk-1$ $p$-even and $c$-even nodes. Despite not being related by \eqref{permiso},  $\su(\groupn|\groupm)$ and $\su(\groupn,\!|\groupm)$ are nevertheless isomorphic (by \eqref{outer}).}
\end{centering}
\end{figure}
Under Lie algebra outer automorphism (\ref{outer})  we have further isomorphism of $^*$-algebras:
\be
        \su(\groupn,\!|\groupm)\stackrel{\text{out}}{\simeq}\su(\groupn|\groupm)\,.
\ee
Of course, this is a particular case of \eqref{genericisoout}, given the notational meaning $\su(\groupn|\groupm)\equiv \su(0,\groupn|\groupm,0)$.

The forms of the extended diagrams in Fig.~\ref{fig:possiblerealform} will be called canonical.
Now we get the first justification of
\begin{wrapfigure}{r}{0.44\textwidth}
\captionsetup{width=0.35\textwidth}
  \begin{center}
    \includegraphics[width=0.2\textwidth]{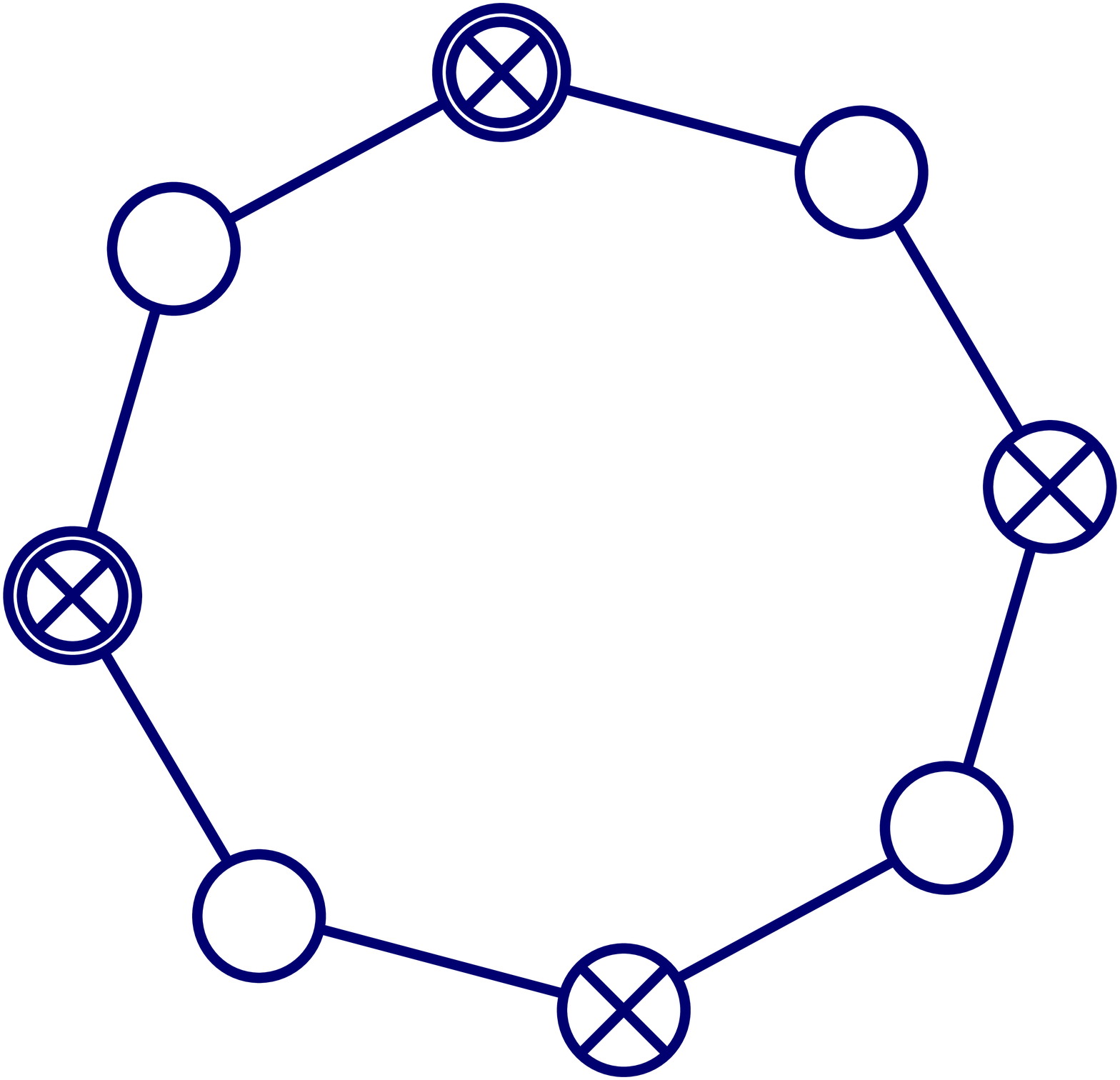}
  \caption{\label{fig:adscft}A convenient choice of the grading in  AdS/CFT integrability.}
  \end{center}
  \vspace{-10pt}
\end{wrapfigure}
why we kept the $c$- and $p$-grading functions general. Even for the canonical forms of the diagrams, there are different choices for declaring which node is the extended one. Each of them is interesting for practical purposes, but different choices lead to different grading functions. For instance, if one considers the 4-dimensional superconformal algebra with $\mathcal{N}$ supersymmetries, the distinguished Kac-Dynkin diagram would correspond to the $\su(2,2|\mathcal{N})$ grading. However, a better choice for describing unitary representations appears to be $\su(2,\!|\mathcal{N}|2)$ or $\su(2|\mathcal{N},\!|2)$. Also non-canonical forms can be interesting. Let us mention one  example coming from integrability: The asymptotic Bethe Ansatz equations for the AdS/CFT integrable system \cite{Beisert:2005fw} have their simplest form for the non-canonical choice of the gradings shown in Fig.~\ref{fig:adscft}.

\subsection{\label{sec:allarehw}All unitary representations are  of highest- or lowest-weight type}
In the following we assume that $p$-odd subspace is not empty. In this case there is a strong restriction on the type of possible unitary representations. Indeed, suppose we consider a representation with Hermitian odd  generators $\mathfrak{f}$ in the Hermitian basis. Their spectrum is real,  hence  even  generators $\mathfrak{b}$ obtained by $\mathfrak{b}=[\mathfrak{f},\mathfrak{f}]=2\,\mathfrak{f}^2$  have a non-negative-definite spectrum. These bosonic generators can belong to Cartan subalgebra, hence we get a strong restriction on possible weights.

Let us now make  precise statements which are generally known in the literature and can be found e.g. in \cite{MR1134934}.

\begin{proposition}
The real form $\su(\groupp,\groupq|\groupr,\groups)$  with all $\groupp,\groupq,\groupr,\groups$ being non-zero has only the trivial unitary representation.
\end{proposition}
\begin{proof}
Consider an $\su(1,1|1,1)$ subalgebra with $(1,1|1,1)$ choice of $p$- and $c$-gradings. Then, for $\mu=1$ or $2$ and $a=3$ or $4$, one has $p_\mu=\bar 0$ and $p_a=\bar 1$ and
\be\label{relationEaf}
       0\leq (\gE_{\mu a}+\gE_{\mu a}^*)^2=[\gE_{\mu a},\gE_{\mu a}^*]=(-1)^{c_\mu+c_a}(\gE_{aa}+\gE_{\mu\mu})\,.
\ee
Since $c_1=c_4=\bar 0$ and $c_2=c_3=\bar 1$, we can explicitly write
\be
\gE_{11}+\gE_{44}\geq0\,,\ \gE_{11}+\gE_{33}\leq0\,,\ \gE_{22}+\gE_{44}\leq0\,,\ \gE_{22}+\gE_{33}\geq0\,.
\ee
These inequalities are only satisfied if $\gE_{11}=\gE_{22}=-\gE_{33}=-\gE_{44}$.

By picking different $\su(1,1|1,1)$ subalgebras and repeating the logic  we conclude that $\gE_{ii}=(-1)^{p_i+p_1} \gE_{11}$ for all $i$. But this implies that all weights $\omega_i$ are zero.
\end{proof}

\begin{proposition}
\label{thr:abouthw}
$\su(\groupp,\groupq|\groupm)$ admits only unitary representations which are of highest-weight type (UHW) for the $(\groupp,\groupq|\groupm)$ choice of grading.
\end{proposition}

\ \\
\noindent{\it Proof.} Let $\dot\beta\in\overline{1,\groupp}$, $\beta\in\overline{\groupp+1,\groupp+\groupq}$ and $a\in\overline{\groupp+\groupq+1,\groupp+\groupq+\groupm}$. Then  one has from relation \eqref{relationEaf} that
\be\label{constr1}
        \gE_{\dot\beta\dot\beta}+\gE_{aa}\leq 0\,,\ \ \gE_{\beta\beta}+\gE_{aa}\geq 0 \ \ \  \Rightarrow\ \ \ \gE_{\dot\beta\dot\beta}-\gE_{\beta\beta}\leq0\,.
\ee
\begin{wrapfigure}{l}{0.35\textwidth}
\captionsetup{width=0.28\textwidth}
  \begin{center}

\begin{picture}(80,90)(0,0)

\thicklines

\color{BurntOrange}
\dashline{4}(0,80)(0,100)
\dashline{4}(30,80)(30,100)
\dashline{4}(55,80)(55,100)
\dashline{4}(80,80)(80,100)
\put(0,90){
\put(15,0){\vector(1,0){15}}
\put(15,0){\vector(-1,0){15}}
\put(12,4){$\groupp$}
}
\put(30,90){
\put(10,0){\vector(-1,0){10}}
\put(10,0){\vector(1,0){15}}
\put(10,4){$\groupq$}
}
\put(55,90){
\put(10,0){\vector(-1,0){10}}
\put(10,0){\vector(1,0){15}}
\put(10,4){$\groupm$}
}

\color{Blue}
\put(0,0){\line(1,0){80}}
\put(0,0){\line(0,1){80}}
\put(0,80){\line(1,0){80}}
\put(80,0){\line(0,1){80}}
\dottedline{4}(80,0)(0,80)
\dashline{4}(0,50)(80,50)
\dashline{4}(30,0)(30,80)
\dashline{4}(55,50)(55,80)

\color{black}

\put(15,70){$\tiny\gE^{+}$}
\put(60,30){$\tiny\gE^{+}$}
\put(35,65){$\tiny\gE_{\dot\beta\beta}$}
\put(60,65){$\tiny\gE_{\dot\beta a}$}

\end{picture}

  \caption{\label{fig:raising}Notations for raising operators used in the proof.}
  \end{center}
\end{wrapfigure}
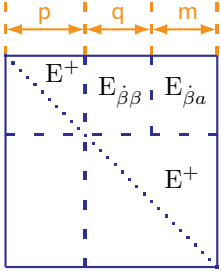
Due to $[\gE_{\dot\beta\dot\beta}-\gE_{\beta\beta},\gE_{\dot\beta\beta}]=+2\,\gE_{\dot\beta\beta}$, action of $\gE_{\dot\beta\beta}$ on any state, if non-zero, increases the eigenvalue of $\gE_{\dot\beta\dot\beta}-\gE_{\beta\beta}$. All $\gE_{\dot\beta\beta}$ commute among themselves. Hence one can always find a state $|\es ''\rangle$ for which $\gE_{\dot\beta\beta}|\es ''\rangle=0$ simultaneously for all pairs of  $\dot\beta,\beta$. Furthermore, one can act on this state by $\gE_{\dot\beta a}$. Because $\gE_{\dot\beta a}^2=0$ and $[\gE_{\dot\beta \beta},\gE_{\dot\beta a}]=0$ one eventually gets the state $|\es'\rangle$ which is annihilated by both $\gE_{\dot\beta \beta}$ and $\gE_{\dot\beta a}$.

Finally, act on $|\es'\rangle$ by all the raising generators ($\gE_{ij}$, $i<j$) that are denoted as $\gE^+$ in Fig.~\ref{fig:raising}. Since they belong to the compact sub-algebra $\su(\groupp)\oplus\su(\groupq|\groupm)$, they will eventually create a state $|\es\rangle$ that is annihilated by them. But because the commutators $[\gE_{\dot\beta\beta},\gE^+]$ and $[\gE_{\dot\beta a},\gE^+]$ result in either $\gE_{\dot\beta \beta}$ or $\gE_{\dot\beta a}$, $|\es\rangle$ is still annihilated by $\gE_{\dot\beta \beta}$ and $\gE_{\dot\beta a}$, hence $\gE_{ij}|\es\rangle=0$ for any $\gE_{ij}\in\mathfrak{n}^+$, so it is the highest-weight state.
\qed

\ \\
\noindent Note that if one considers the $(\groupp,\groupq,\!|\groupm)$ choice of grading, the signs in \eqref{constr1} would change to opposite and we would get the lowest-weight unitary representations (ULW). Furthermore, when $\groupp=0$ or $\groupq=0$,  one gets a finite-dimensional representation which has both  highest- and  lowest-weight states.

\section{\label{sec:necescond}Necessary conditions for unitarity and related topics}
\subsection{Closer look at the duality transformations}
In the proposition~\ref{thr:abouthw}, we dealt with a particular choice of $\su(\groupp,\groupq|\groupm)$ grading. Let us now show how the statement transforms if one wants to change to another grading.

Our main interest is a duality transformation, i.e. the map \eqref{permiso} with $\sigma$ being an elementary permutation $\sigma_{(i,i+1)}$. Sub-algebra $\mathfrak{n}^+$ is not invariant under its action, hence, generically, $\HWS$ is no longer a highest-weight state after such a transformation. To construct the new highest-weight state, consider the rank-1 sub-algebra $\{\ch_i,\gE_{i,i+1},\gE_{i+1,i}\}$. If it is compact, {\it i.e.} $\su(2)$ or $\su(1|1)$, then we can construct the state $\HWS'$ which is annihilated by $\gE_{i+1,i}$, simply by acting with $\gE_{i+1,i}$ on $\HWS$  sufficient number of times. It is now  easy to check that $\HWS'$ is the highest-weight state of $\su(\groupp,\groupq|\groupm)$-module in the grading obtained after \eqref{permiso} with $\sigma=\sigma_{(i,i+1)}.$

The conclusion is that a representation remains of the highest-weight type if the duality transformation is performed on a compact Dynkin node. One can check that in the case of $p$-even node, when the duality transformation is the Weyl reflection, the value of the fundamental weight does not change under the duality. The functions $p$ and $c$ are not affected either, hence we will not consider the duality transformation on the compact $p$-even node in the following.

If the node is non-compact, i.e. $c$-odd and $p$-even, then generically we lose the a clear notion of highest or lowest weight after duality transformation and hence we will not consider such transformations either\footnote{This does not mean that we exclude certain representations. Duality transformations do not create new representations, they only change the way we describe them. So we simply do not  consider certain inconvenient descriptions.}.

Hence we are left with dualities on the $p$-odd nodes.  They are the interesting ones and, in the following, the words "duality transformation" will refer  only to this case.

To further study duality transformations, let us introduce a two-dimensional lattice and  consider  the Dynkin diagram as a path on it\footnote{The idea to consider Kac-Dynkin diagrams as two-dimensional paths was proposed in \cite{Kazakov:2007fy} to parameterise possible B\"acklund flows that appear in the study of supersymmetric integrable spin chains.}, see Fig.~\ref{fig:2dlattice}. One decodes the grading from the path choice as follows: If the $j$'th line of the path is horizontal then $p_{j}=\bar 1$  and $c_{j}=\bar 0$. If this line is vertical then $p_j=\bar 0$. Additionally, if it is above the dashed line then $c_{j}=\bar 0$, otherwise $c_{j}=\bar 1$. One can also assign the eigenvalue of $\gE_{jj}$ on $\HWS$  in this particular grading to the line. Hence, by choosing different paths and $j$'s, one eventually assigns a certain weight to each line of the lattice. It will be clear below that this procedure is unambiguous. Therefore the designed  lattice  contains information about fundamental weights for all possible choices of gradings~\footnote{Given that we do not perform duality transformation on non-compact nodes, we restrict gradings to the case when at most one non-compact node is present.} and it shall be called the weight lattice~\footnote{Not to be  confused with a lattice in the root space of a Lie (super)algebra.}. Given that the two-dimensional image clarifies the role of various gradings, it will be convenient to use the following notation for the fundamental weight:
\be
[\wb_L;\wf;\wb_R]\equiv [\wb_L^1,\ldots,\wb_L^{\groupp};\wf_1,\ldots,\wf_\groupm;\wb_R^1,\ldots \wb_R^{\groupq}]\,,
\ee
where $\wb_{L}^{\dot\alpha}$ are eigenvalues of $\gE_{\dot\alpha\dot\alpha}$ with $p_{\dot\alpha}=\bar 0$ and $c_{\dot\alpha}=\bar 1$; $\wb_{R}^{\alpha}$ are eigenvalues of $\gE_{\alpha\alpha}$ with $p_{\alpha}=\bar 0$ and $c_{\alpha}=\bar 0$; and $\wf_{a}$ are eigenvalues of $\gE_{aa}$ with $p_a=\bar 1$ and $c_{a}=\bar 0$. Unlike in \eqref{fundaweight}, appearance of these eigenvalues does not necessarily respect the total order on the set of indices, but the explicit values of $\wb$'s and $\wf$'s will depend on the choice of the Kac-Dynkin path.

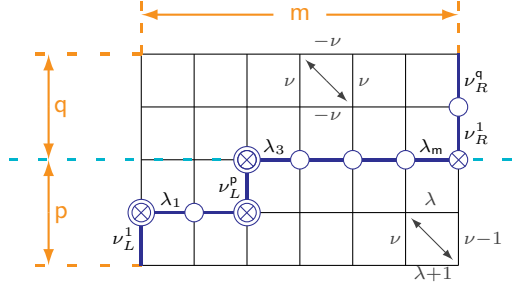
\begin{figure}[t]
\begin{center}
\begin{picture}(160,110)(-40,0)
\color{black}
\thinlines
\multiput(0,0)(0,20){5}{\line(1,0){120}}
\multiput(0,0)(20,0){7}{\line(0,1){80}}
\thicklines

\color{BurntOrange}
\dashline{4}(-40,80)(0,80)
\put(-35,40){
\put(0,20){\vector(0,1){20}}
\put(0,20){\vector(0,-1){20}}
\put(2,18){$\groupq$}
}
\dashline{4}(-40,0)(0,0)
\put(-35,0){
\put(0,20){\vector(0,1){20}}
\put(0,20){\vector(0,-1){20}}
\put(2,18){$\groupp$}
}

\dashline{4}(0,80)(0,100)
\dashline{4}(120,80)(120,100)
\put(50,95){\vector(-1,0){50}}
\put(70,95){\vector(1,0){50}}
\put(56,93){$\groupm$}

\color{Turquoise}
\dashline{4}(-50,40)(140,40)

\color{Blue}
\drawline(0,0)(0,20)(40,20)(40,40)(120,40)(120,80)

\color{white}
\put(0,20){\circle*{10}}
\put(0,20){\circle*{7}}
\put(20,20){\circle*{7}}
\put(40,20){\circle*{10}}
\put(40,20){\circle*{7}}
\put(40,40){\circle*{10}}
\put(40,40){\circle*{7}}
\put(40,40){\circle*{7}}
\put(60,40){\circle*{7}}
\put(80,40){\circle*{7}}
\put(100,40){\circle*{7}}
\put(120,40){\circle*{7}}
\put(120,60){\circle*{7}}

\thinlines
\color{Blue}
\put(0,20){\circle{10}}
\put(0,20){\circle{7}}
\put(20,20){\circle{7}}
\put(40,20){\circle{10}}
\put(40,20){\circle{7}}
\put(40,40){\circle{10}}
\put(40,40){\circle{7}}
\put(40,40){\circle{7}}
\put(60,40){\circle{7}}
\put(80,40){\circle{7}}
\put(100,40){\circle{7}}
\put(120,40){\circle{7}}
\put(120,60){\circle{7}}

\put(0,20){
\put(-2.6,-2.6){\line(1,1){5.2}}
\put(-2.6,2.6){\line(1,-1){5.2}}
}
\put(40,20){
\put(-2.6,-2.6){\line(1,1){5.2}}
\put(-2.6,2.6){\line(1,-1){5.2}}
}
\put(40,40){
\put(-2.6,-2.6){\line(1,1){5.2}}
\put(-2.6,2.6){\line(1,-1){5.2}}
}
\put(120,40){
\put(-2.6,-2.6){\line(1,1){5.2}}
\put(-2.6,2.6){\line(1,-1){5.2}}
}

\color{black}
\put(-11,8){$\scriptstyle\tiny\wb_L^1$}
\put(29,28){$\scriptstyle\wb_L^{\groupp}$}
\put(122,48){$\scriptstyle\wb_R^1$}
\put(122,68){$\scriptstyle\wb_R^{\groupq}$}
\put(7,23){$\scriptstyle\lambda_1$}
\put(46,43){$\scriptstyle\lambda_3$}
\put(105,43){$\scriptstyle\lambda_{\groupm}$}

\color{gray!50!black}
\put(60,60){
\put(5,-5){$\scriptstyle-\nu$}
\put(5,23){$\scriptstyle-\nu$}
\put(-6,8){$\scriptstyle \nu$}
\put(22,8){$\scriptstyle \nu$}
\put(10,10){\vector(1,-1){8}}
\put(10,10){\vector(-1,1){8}}
}

\put(100,0){
\put(3,-5){$\scriptstyle \lambda+1$}
\put(7,23){$\scriptstyle \lambda$}
\put(-6,8){$\scriptstyle \nu$}
\put(22,8){$\scriptstyle \nu-1$}
\put(10,10){\vector(1,-1){8}}
\put(10,10){\vector(-1,1){8}}
}

\end{picture}

\caption{\label{fig:2dlattice} Lattice of all possible weights of highest-weight vectors of $\su(\groupp,\groupq|\groupm)$. The lines adjoining the lower-left and upper-right corners are dropped from the resulting Kac-Dynkin diagrams. }
\end{center}
\end{figure}

The duality transformations is a way to change one path to another, as illustrated by the following example
\be\label{weightlattice}
\raisebox{-.1em}{\includegraphics[width=0.5\textwidth]{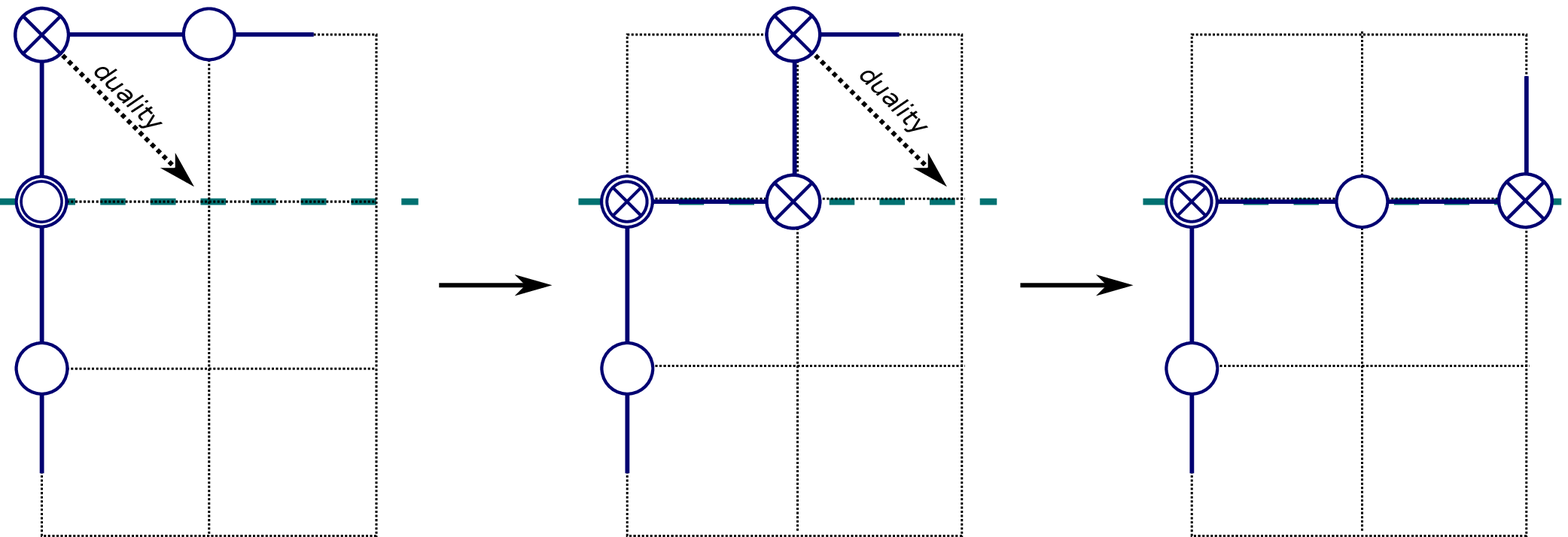}}\,.
\ee
Here, one has $\su(2,1|2)$ grading on the left diagram. By performing two  duality transformations as shown, we end up with the diagram on the right, which corresponds to the $\su(2,\!|2|1)$ grading.\\[0em]

Let $p_\mu\neq p_{a}$ for some $\mu,a$ that are subsequent numbers on the given Kac-Dynkin path. For definiteness, assume that $\mu<a$ and $p_\mu=\bar 0$, $p_{a}=\bar 1$. Let $\gE_{\mu\mu}\HWS=\wb\HWS$ and $\gE_{aa}\HWS=\wf\HWS$. By definition of the highest-weight state, $\gE_{\mu a}\HWS=0$. Compute now the norm of $\gE_{a\mu}\HWS$:

\be\label{normchange}
||E_{a\mu}\HWS||^2=\HWSbra\,E_{a\mu}^{\dagger }\, E_{a\mu}\, \HWS=(-1)^{c_a+c_\mu}(\wf+\wb)||\,\HWS||^2\,.
\ee
We see that if $\wf+\wb=0$, the norm of $\gE_{a\mu}\HWS$ is zero and hence it must be, if we are to construct a unitary representation, that $\gE_{a\mu}\HWS=0$. Therefore $\HWS$ is annihilated both by $\gE_{\mu a}$ and $\gE_{a\mu}$ and hence it is the highest-weight state before and after the duality transformation $\sigma_{(i,i+1)}$. Hence we conclude \\[0em]
\begin{itemize}
\item
If $\wf+\wb=0$ then $\HWS'=\HWS$ and  one has
\begin{subequations}
\label{dualityrules}
\be
\raisebox{-.5em}{\includegraphics[height=5EM]{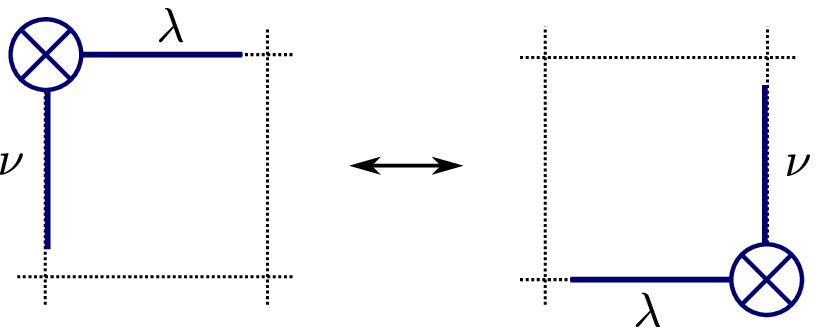}}\,,
\label{eq:dualityrule1}
\ee

\item If $\wf+\wb\neq 0$ then  $\gE_{a\mu}\HWS\neq 0$. However, $\gE_{a\mu}^2\HWS=\frac 12[\gE_{a\mu},\gE_{a\mu}]\HWS=0$. Therefore
 $\HWS'=\gE_{a\mu}\HWS$, and we have
\be
\label{eq:dualityrule2}
\raisebox{-.5em}{\includegraphics[height=5EM]{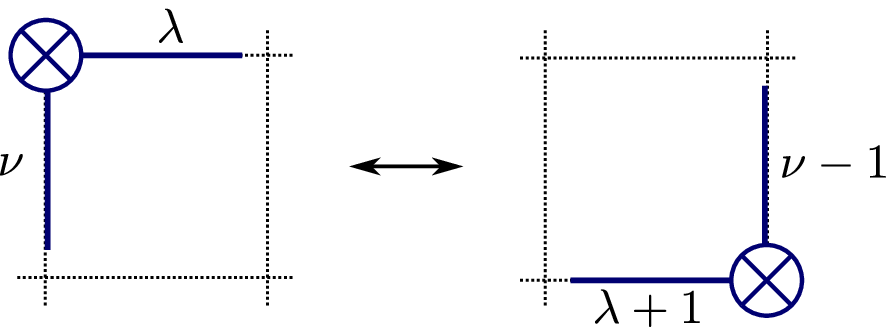}}\,.
\ee
\end{subequations}
\end{itemize}
Note that the duality transformation can only affect the eigenvalues of $\gE_{aa}$ and $\gE_{\mu\mu}$ and does not affect eigenvalues of other $\gE_{ii}$.

Relations \eqref{dualityrules} give us the explicit rule how to map fundamental weights between different gradings\footnote{For the sake of clarity, we drew $c$-even $p$-odd node in \eqref{dualityrules}. The same rule applies for $c$-odd $p$-odd node.}. We can easily generate the whole weight lattice in Fig.~\ref{fig:2dlattice} using this rule.

\subsection{\label{sec:plaquette}Plaquette unitarity constraints}
As it follows from \eqref{normchange} and the subsequent discussion, the norms of $\HWS$ and $\HWS'$ have the same sign if $(-1)^{c_a+c_\mu}(\wf+\wb)\geq 0$.  This inequality  should hold in any unitary representation, we shall call it the  plaquette constraint. There are $(\groupp+\groupq)\times \groupm$ plaquettes on the weight lattice, and the duality transformation can be associated to either of them. The necessary condition for unitarity is that all the plaquette constraints are satisfied simultaneously.

There are $\groupp\times\groupm$ plaquettes where the duality is performed on a $c$-odd node. For them, one should satisfy $\wf+\wb\leq 0$. And there are $\groupq\times\groupm$ plaquettes were the duality is performed on a $c$-even node. For them, one should satisfy $\wf+\wb\geq 0$.

A remarkable fact is that these simple necessary constraints are also sufficient. We need to develop certain techniques to prove it, which is the goal of sections~\ref{sec:repos} and \ref{sec:classification}. Our proof will be constructive: we  explicitly realise all  the unitary representations in certain deformations of Fock spaces. A reader not interested in these details can jump directly ahead to our main conclusion -- Theorem~\ref{th:final} on page~\pageref{th:final} and the following discussion which puts the unitarity constraints into a more concrete form.

Before starting to explicitly construct representations, we need to gain some extra intuition about the plaquette constraints. To this end, we  analyse a particular example of some $\su(2|4)$ representations in subsection~\ref{sec:su24} and then  make an extended discussion about possible shortening conditions in subsection~\ref{sec:shortening}. Only the compact case will be considered in these subsections, but generalisation to the non-compact case is straightforward, as is discussed in appendix~\ref{sec:Genusho} for the example of superconformal algebra.

\renewcommand{\groupq}{\mathsf{m}}
\renewcommand{\groupm}{\mathsf{q}}
\subsection{\label{sec:su24}$\su(2|4)$ example and generalisation to $\su(\groupq|\groupm)$ }
\begin{wrapfigure}{r}{0.35\textwidth}
\captionsetup{width=0.30\textwidth}
\vspace{-3em}
  \begin{center}
    \includegraphics[width=0.15\textwidth]{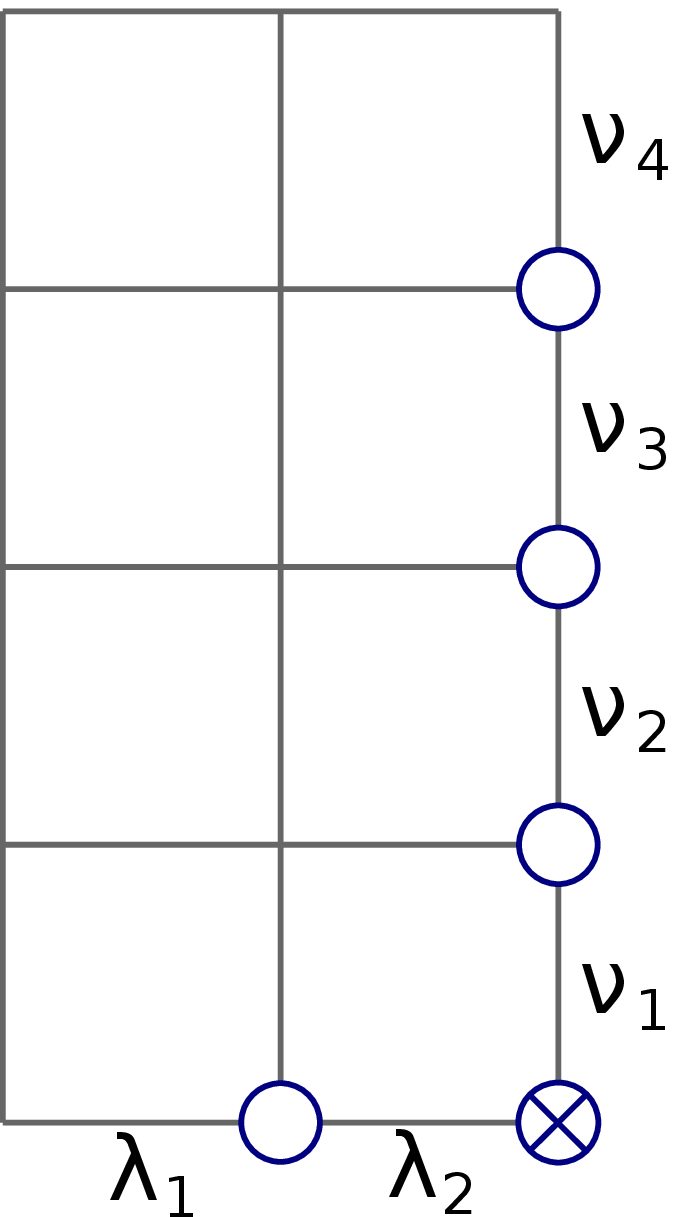}
  \caption{\label{fig:example24}A fundamental weight on the weight lattice for the $\su(0,\!|\groupq|\groupm)$ choice of grading.}
  \end{center}
  \vspace{-10pt}
\vspace{2em}
\end{wrapfigure}
To conform with the general case discussed later, we choose $p_1=\bar 1$ and $p_{\groupq+\groupm}=\bar 0$ in the distinguished grading of $\su(\groupq|\groupm)$ (i.e. we think about it as $\su(0,\!|\groupq|\groupm)$). The UHW irrep of $\su(\groupq|\groupm)$ is parameterised by the fundamental weight
\be
[\wf_1,\ldots,\wf_{\groupq};\wb_1,\ldots,\wb_{\groupm}]\equiv [\wf;\wb]\,,
\ee
see Fig.~\ref{fig:example24}. The weights for all other  gradings are restored by systematic application of \eqref{dualityrules}. Since we are dealing with $\su$ but not $\algu$ algebra, the overall shift $\wf_a\to \wf_a+\Lambda$, $\wb_i\to\wb_i-\Lambda$ does not affect the representation. We use it to set $\wb_{\groupm}=0$. Additionally, let us define $\yf$ by
\be\label{introducingbeta}
\wf_a=\yf_a+\beta\,,
\ee
where $\yf_\groupq=0$. $\yf$ and $\wb$ should be integer partitions. Indeed, we can consider the highest-weight state as the highest-weight state of the $\su(\groupq)$ or $\su(\groupm)$ submodules, and these modules should be unitary as well, $\yf$ and $\wb$  being integer partitions is the well-known condition of unitarity of representations of compact $\su(\groupn)$ algebras.

The only additional information about unitarity of $\su(\groupq|\groupm)$ is contained in the value of $\beta$ and we should determine what restrictions are imposed on $\beta$ by the plaquette constraints.

Consider an explicit  example of $\su(2|4)$. We choose $\yf=(1,0)$ and $\wb=(3,1,0,0)$. Note that the height of the Young diagram $\wb$ (number of non-zero rows) is $h_{\wb}=2$ here. Now we vary the value of $\beta$ and consider several different cases illustrated in Fig.~\ref{latticesample}.

\begin{figure}[t]
\begin{center}
\begin{tabular}{ccccc}
\includegraphics[width=0.15\linewidth]{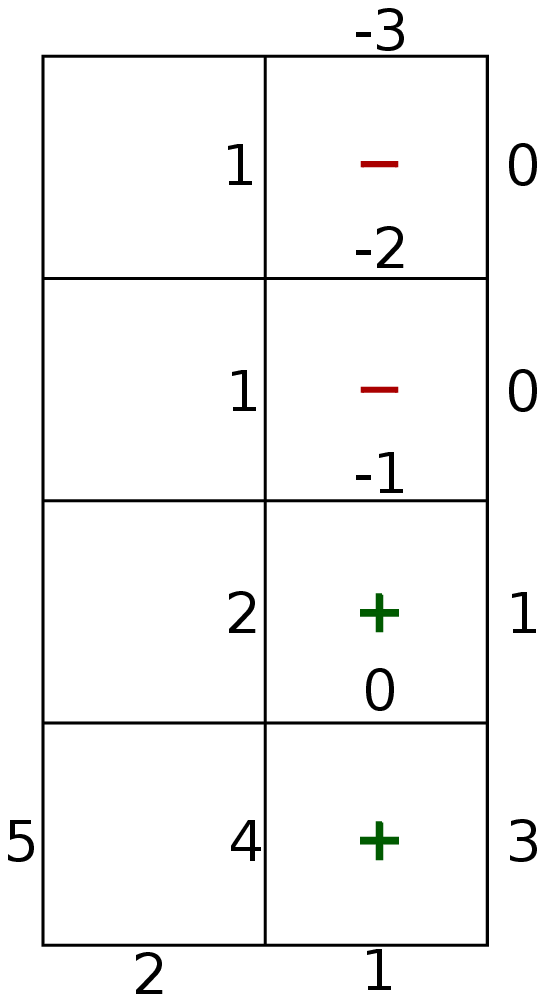} &
\includegraphics[width=0.15\linewidth]{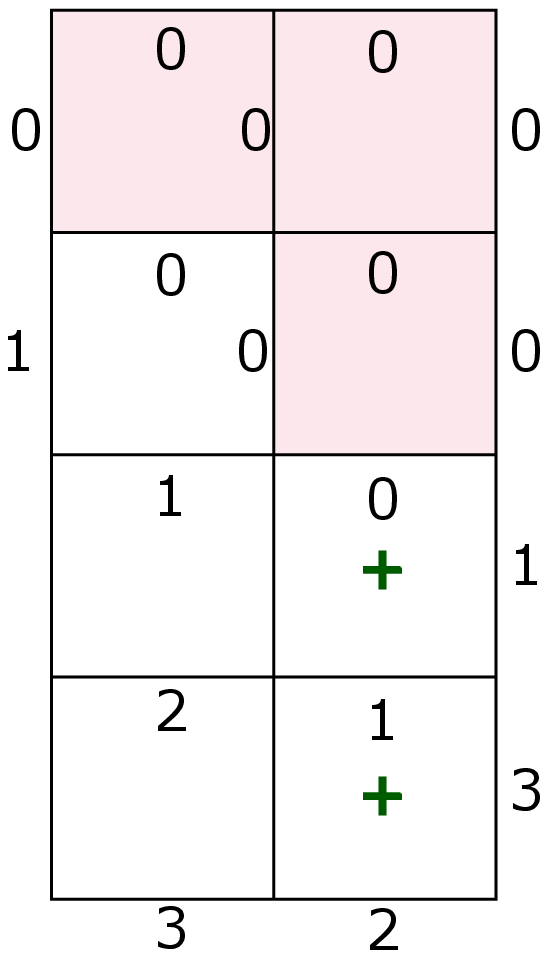} &
\includegraphics[width=0.15\linewidth]{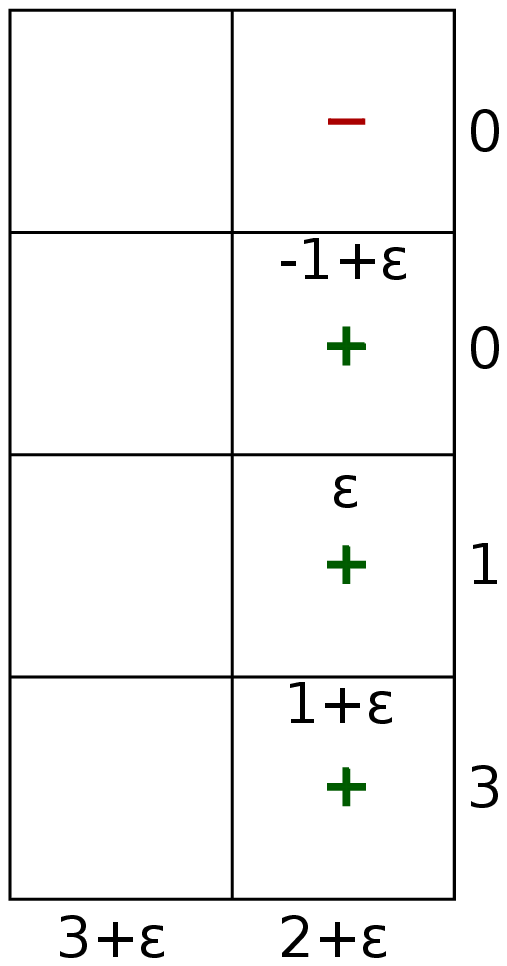} &
\includegraphics[width=0.15\linewidth]{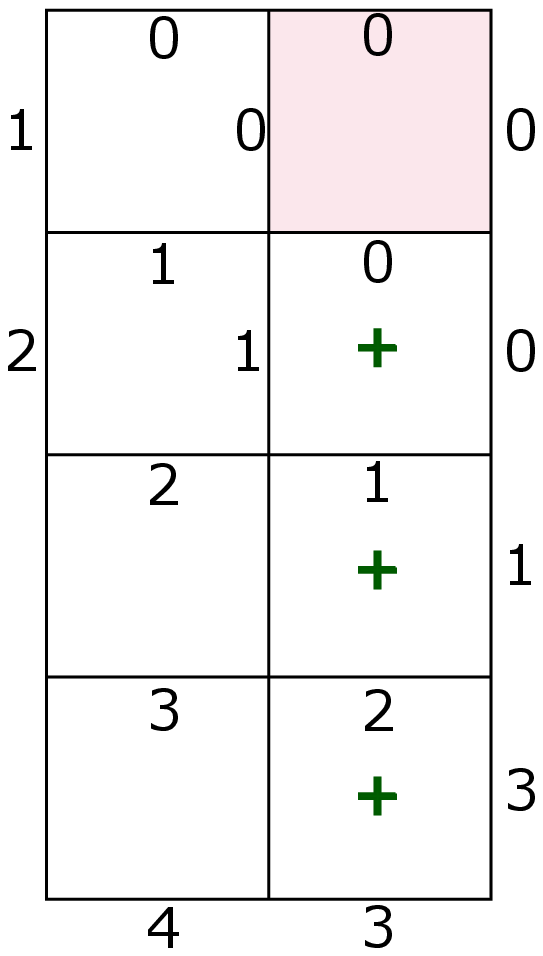} &
\includegraphics[width=0.15\linewidth]{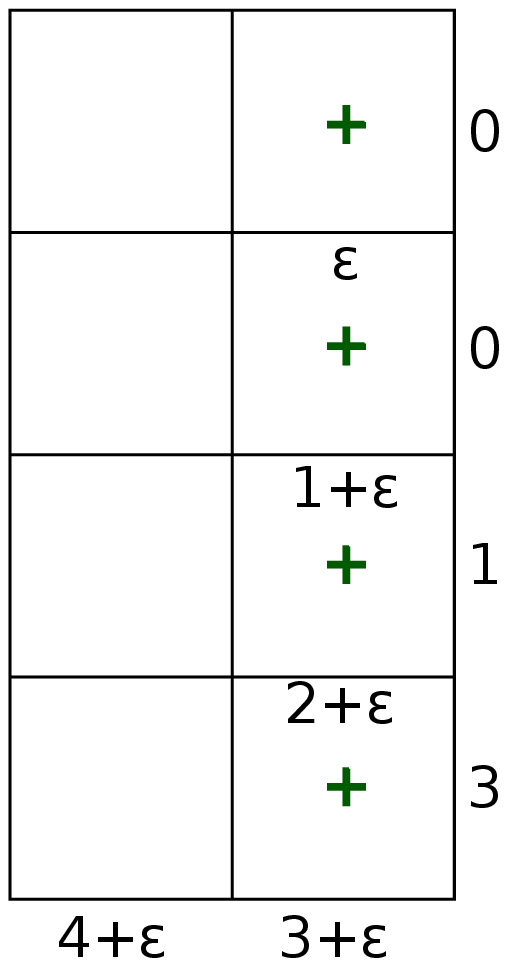}
\\
A: $\beta=1<2$ & B: $\beta=2$ & C: $3>\beta>2$ & D: $\beta=3$ & E: $\beta>3$
\end{tabular}
\caption{\label{latticesample}Lattice of possible weights for a representation of $\su(2|4)$. $\yf=(1,0)$, $\wb=(3,1,0,0)$. Five different choices of the value $\beta$ are shown. Cases A and C are not unitary. Cases B and D are  short unitary representations, shaded plaquettes are those where $\wf+\wb=0$. Case E is a long unitary representation.}
\end{center}
\end{figure}
If $\beta<2=h_{\wb}$ then the representation cannot be unitary as we encounter states with norm of opposite signs when performing chains of duality transformations, see Fig.~\ref{latticesample}A. In this figure, symbols "+" and "-" refer to the sign of $(\wf+\wb)$ in \eqref{normchange}. When $\beta=2$ or $\beta=3$, all the necessary unitarity constraints are satisfied. Fig.~\ref{latticesample}B and Fig.~\ref{latticesample}D illustrate these  cases, with shaded regions being the plaquettes where equalities $\wf+\wb=0$ occur. Fig.~\ref{latticesample}C describes the non-unitary case $2<\beta<3$, the negative norm emerges when performing duality transformation in the upper-right corner.  Finally any real $\beta>3=\groupm-1$  allows one  to satisfy the necessary unitary constraints as well, see Fig.~\ref{latticesample}E, now equality $\wf+\wb=0$ never occurs.

The situation in generic $\su(\groupq|\groupm)$ case should be clear by now: If $\beta<h_{\wb}$, representation is not unitary because we have an explicit duality transformation that creates a negative-norm space. Sign-changing duality transformations also happen for any non-integer $\beta$ with $h_{\nu}\leq \beta\leq\groupm-1$, whereas  we can satisfy unitarity constraints for integer $\beta$ in this range, with  equalities $\wf+\wb=0$  on some of the plaquettes. Finally, $\beta>\groupm-1$ allows one to satisfy the unitarity constraints using only strict inequalities.

Note that this analysis of duality transformations is  done uniquely in the right-most column of the weight lattice. Indeed, if to consider plaquettes in other columns, the unitarity requirements there are  weaker (sums $\wf+\wb$ become more positive) and hence do not introduce any new information. In fact, the unitarity constraints on $\beta$ depend only on $\groupm$ and $h_{\wb}$ and are not sensitive to $\groupq$ or $\yf$. To derive the constraints, one should have $\groupq\neq 0$ though.

\subsection{\label{sec:shortening}Short representations and Young diagrams on fat hook}
We remind that $\wf+\wb=0$ means $\HWS=\HWS'$ for a given duality transformation, this is the so-called shortening condition. Representations that enjoy this property for some plaquette are called short or atypical, as opposed to long or typical representations where this equality is never satisfied. In physics literature, see e.g. \cite{Dolan:2002zh}, one distinguishes BPS (sometimes called fully short) and semi-short representations; these are particular instances of atypical representations as will be made  precise below.

We  discuss now what consistent choices of shortenings are possible. For the distinguished choice of $\su(\groupq|\groupm)$ grading as above, the highest-weight state is annihilated by all the generators  $\gE_{a\mu}$, where $p_a=\bar 1$ and $p_{\mu}=\bar 0$. This state is denoted as $|\es\rangle$ until the end of this section. Consider the action of $p$-odd  generators $\gE_{\mu a}$ on  $|\es\rangle$. Note that $[\gE_{\mu a},\gE_{\wb b}]=0$. Therefore, if this action were free from further constraints, it would generate a $2^{\groupq\times\groupm}$-dimensional vector space equivalent to the space of polynomials in $\groupq\times \groupm$ Grassmann variables. Shortening conditions introduce further constraints in this space making it smaller.\\[0em]

Let $\gE_{\mu a}|\es\rangle=0$ for some $a,\mu$. Then it must be that $\wf_a+\wb_\mu=0$,  the weights are for the distinguished choice of the grading. For any $b>a$ one should have $\wf_b\leq \wf_a$ which immediately implies $\wf_b=\wf_a$,  otherwise $\wf_b+\wb_\mu<0$ and the unitarity is violated. In particular, one should have $\wf_\groupq+\wb_\mu=0$. Similarly $\wb_{\mu}=\wb_{\mu'}$ for any $\mu'>\mu$, and in particular $\wf_\groupq+\wb_\groupm=0$. But $\wb_\groupm=0$ by our choice of an overall shift $\Lambda$, hence $\wf_{\groupq}=0$ and $\wf_a=\wb_\mu=0$.

Let $\wb_1\neq 0$. Then perform the duality transformation in the bottom right plaquette of the weight lattice (we refer to  Fig.~\ref{fig:example24} as an example). It produces $\wf_\groupq\to \wf_\groupq-1$, thus rendering $\wf_\groupq+\wb_\mu<0$ in the new grading choice, which contradicts the  unitarity constraints. Hence the only possibility is that $\wb_1=0$ and then $\wb_{\mu'}=0$ for any value of $\mu'$.

On the other hand $\wf_a+\wb_\mu=0$ always implies $\gE_{\mu\,a}|\es\rangle=0$. We therefore conclude
\be
\gE_{\mu\,a}|\es\rangle=0 \quad \Leftrightarrow \quad \wf_a=0\text{ and } \wb_{\mu'}=0\text{ for any } \mu'\,.
\ee
Let $\wf_\groupq=\wf_{\groupq-1}=\ldots=\wf_{\groupq-r+1}=0$ and all $\wb_{\mu}=0$. Then such a representation should be called  $\alpha-$BPS, where $\alpha=\frac{r}{\groupq}$ is the fraction of lowering operators $\gE_{\mu\,a}$ that annihilate the highest-weight state $|\es\rangle$ when directly acting on it.

It is also possible that a certain product of lowering operators annihilates $|\es\rangle$ while the individual lowering operators from this product do not annihilate $|\es\rangle$. In analogy with the physics literature this situation should be called semi-shortening. By generalisation of the analysis from above we show that
\be\label{shorteningweightvalues}
\gE_{\mu_ra}\ldots \gE_{\mu_2a}\gE_{\mu_1a}|\es\rangle=0 \quad \Leftrightarrow \quad \wf_a={r-1}\text{ and } \wb_{\mu'}=0\text{ for any } \mu'\geq r\,.
\ee
The direction $\Leftarrow$ is proven as follows. Since $[\gE_{\mu a},\gE_{\mu'a}]=0$, we can always arrange $\gE_{\mu_ia}$ in the way that $\mu_r>\mu_{r-1}>\ldots >\mu_1$. Then, consider  $|\Omega\rangle\equiv\gE_{\mu_{r-1}a}\ldots \gE_{\mu_2a}\gE_{\mu_1a}|\es\rangle$. We can think about $|\Omega\rangle$ as the state obtained by the chain of $r-1$ duality transformations\footnote{More precisely, these will be duality transformations changing highest-weight state only after we  restrict ourselves to $\su(r|1)$ subalgebra that contains $\gE_{\mu_ia}$.}. Each duality lowers $\wf_a$ by one. Hence, after $r-1$ dualities, $\wf_a=0$ in the corresponding grading. On the other hand $\mu_r\geq r$, so $\wb_{\mu_r}=0$. Hence $\wf_a+\wb_{\mu_r}=0$ implying that action of $\gE_{\mu_ra}$ annihilates $|\Omega\rangle$. Note also that $|\Omega\rangle\neq 0$ as none of the duality transformations leading to $|\Omega\rangle$ has $\lambda+\nu=0$.

The direction $\Rightarrow$ is more generically formulated as follows: Let $\prod\limits_{i=1}^{r}\gE_{\mu_ia_i}|\es\rangle=0$ and none of the subproducts of $\prod\limits_{i=1}^{r}\gE_{\mu_ia_i}$ annihilates $|\es\rangle$. Then it is only possible if $a_1=a_2=\ldots =a_r=a$, and one can derive the r.h.s.  of \eqref{shorteningweightvalues}. This statement is proved in appendix~\ref{sec:monoshortproof}.

Finally we remark that it is   possible that not only monomial but also certain polynomial combinations of lowering operators  annihilate $|\es\rangle$. As an example consider the $\su(2|2)$ representation with weight $[\wf;\wb]=[1,1;0,0]$ in the distinguished grading. Then, for $a_1=1, a_2=2,\mu_1=3,\mu_2=4$,  $(\gE_{\mu_2a_2}\gE_{\mu_1a_1}-\gE_{\mu_2a_1}\gE_{\mu_1a_2})|\es\rangle=0$ whereas  $\gE_{\mu_2a_2}\gE_{\mu_1a_1}$ and $\gE_{\mu_2a_1}\gE_{\mu_1a_2}$ do not annihilate $|\es\rangle$ separately. This example relies on $ [\gE_{a_2a_1},\gE_{\mu_2a_2}\gE_{\mu_1a_2}]=\gE_{\mu_2a_2}\gE_{\mu_1a_1}-\gE_{\mu_2a_1}\gE_{\mu_1a_2}$ and it, in particular, uses that $\lambda_1=\lambda_2$. Generalisation of this example to arbitrary cases would involve several technicalities sensitive to a choice of the representation weight, and we do not discuss it here. Instead, intuition coming from monomial constraints \eqref{shorteningweightvalues} will be sufficent for our goals.\\[0em]

\begin{figure}[t]
\begin{center}
\begin{picture}(220,180)(0,0)
\thinlines
\color{black}
\drawline(0,0)(0,210)
\drawline(120,100)(120,210)
\drawline(0,0)(210,0)
\drawline(120,100)(210,100)
\color{blue!5}
\polygon*(0,0)(0,180)(20,180)(20,140)(40,140)(40,80)(80,80)(80,40)(100,40)(100,20)(180,20)(180,0)(0,0)

\color{magenta!10}
\polygon*(40,100)(40,80)(80,80)(80,40)(100,40)(100,20)(120,20)(120,100)(40,100)


\color{black}
\thinlines
\multiput(0,0)(0,20){6}{\line(1,0){120}}
\multiput(0,0)(20,0){7}{\line(0,1){100}}

\color{blue}
\thicklines
\drawline(0,0)(0,180)(20,180)(20,140)(40,140)(40,80)(80,80)(80,40)(100,40)(100,20)(180,20)(180,0)(0,0)

\color{black}
\thinlines
\dashline{4}(10,145)(10,180)
\put(10,175){\vector(0,1){5}}
\dashline{4}(10,130)(10,00)
\put(10,05){\vector(0,-1){05}}
\put(6,135){$\wf_1$}
\dashline{4}(30,125)(30,140)
\put(30,135){\vector(0,1){5}}
\dashline{4}(30,110)(30,40)
\put(30,45){\vector(0,-1){5}}
\put(26,115){$\wf_2$}
\dashline{4}(62,30)(20,30)
\put(25,30){\vector(-1,0){5}}
\dashline{4}(75,30)(100,30)
\put(95,30){\vector(1,0){5}}
\put(65,27){$\wb_2$}
\dashline{4}(132,10)(20,10)
\put(25,10){\vector(-1,0){5}}
\dashline{4}(148,10)(180,10)
\put(175,10){\vector(1,0){5}}
\put(135,7){$\wb_1$}
\put(62,87){$\scriptstyle{0}$}
\put(70,102){$\scriptstyle 0$}
\put(90,102){$\scriptstyle 0$}
\put(110,102){$\scriptstyle 0$}

\color{BurntOrange}
\thinlines
\put(55,200){\vector(-1,0){55}}
\put(65,200){\vector(1,0){55}}
\put(58,198){$\groupq$}
\put(200,45){\vector(0,-1){45}}
\put(200,55){\vector(0,1){45}}
\put(196,48){$\groupm$}

\color{Blue}
\thicklines
\drawline(0,0)(20,0)(20,40)(60,40)(60,100)(120,100)

\color{white}
\put(20,0){\circle*{7}}
\put(20,20){\circle*{7}}
\put(20,40){\circle*{7}}
\put(40,40){\circle*{7}}
\put(60,40){\circle*{7}}
\put(60,60){\circle*{7}}
\put(60,80){\circle*{7}}
\put(60,100){\circle*{7}}
\put(80,100){\circle*{7}}
\put(100,100){\circle*{7}}

\thinlines
\color{Blue}
\put(20,0){\circle{7}}
\put(20,20){\circle{7}}
\put(20,40){\circle{7}}
\put(40,40){\circle{7}}
\put(60,40){\circle{7}}
\put(60,60){\circle{7}}
\put(60,80){\circle{7}}
\put(60,100){\circle{7}}
\put(80,100){\circle{7}}
\put(100,100){\circle{7}}

\put(20,0){
\put(-2.6,-2.6){\line(1,1){5.2}}
\put(-2.6,2.6){\line(1,-1){5.2}}
}
\put(20,40){
\put(-2.6,-2.6){\line(1,1){5.2}}
\put(-2.6,2.6){\line(1,-1){5.2}}
}
\put(60,40){
\put(-2.6,-2.6){\line(1,1){5.2}}
\put(-2.6,2.6){\line(1,-1){5.2}}
}
\put(60,100){
\put(-2.6,-2.6){\line(1,1){5.2}}
\put(-2.6,2.6){\line(1,-1){5.2}}
}

\end{picture}

\caption{\label{fig:samplefathook} Example of a Young diagram  on a $(\groupq|\groupm)$ fat hook. Its shape is defined by the fundamental weight $[\wf;\wb]$ in a given choice of the grading specifed by a Kac-Dynkin diagram which is a path from the external to the internal corner of the hook. Whereas the weight depends on the grading, the diagram is an invariant object. The weight lattice is the $\groupq\times\groupm$ rectangle between the corners of the hook. Its plaquettes outside of the Young  diagram are the ones where shortenings happen (condition $\lambda+\nu=0$ is satisfied).}
\end{center}
\end{figure}

As we see from \eqref{shorteningweightvalues}, shortenings are straightforwardly related to the value of the fundamental weight $[\wf;\wb]$. This relation can be naturally depicted graphically as outlined in Fig.~\ref{fig:samplefathook}. A representation with integer weights is associated to a Young diagram \cite{Berele:1987yi}, such that the distance between diagram's boundary and the Kac-Dynkin path is given by the fundamental weight (known in this context as Frobenius weight, see~\cite{ChenWangBook}). Shortening \eqref{shorteningweightvalues} corresponds to  a chain of duality transformations, such that the last duality is done on a plaquette with $\lambda+\nu=0$. All such plaquettes are precisely those situated on the weight lattice outside of the Young diagram.

In summary, for integer weights, the unitarity constraint for representations of $\su(\groupq|\groupm)$ can be reformulated as follows: For representation to be unitary, it should be possible to find an overall shift $\Lambda$ that one can associate a Young diagram to the fundamental weight. For short representations, the choice of such shift is unique, whereas typically several choices are possible for long representations. Note also that the Young diagram is  also a proper combinatorial object \cite{Berele:1987yi}, that is the corresponding representation can be constructed using supersymmetric version of the Young symmetriser, and Littlewood-Richardson rules apply as usual for tensor product of such irreducible representations. Finally we note that representations with weights described by Young diagrams are the ones and the only ones that appear in the (non-deformed) oscillator construction revisited in the next sections.

%
\begin{figure}[t]
\begin{center}
\setlength{\unitlength}{0.5pt}
{
\begin{picture}(400,100)(-400,0)
\put(-200,0){
\thicklines
\color{blue!10}
\polygon*(0,0)(0,50)(10,50)(10,40)(20,40)(20,30)(40,30)(40,30)(105,30)(105,10)(125,10)(125,0)(0,0)
\color{blue!20!BurntOrange!15}
\polygon*(40,0)(75,0)(75,30)(40,30)
\color{blue}
\drawline(0,50)(10,50)
\drawline(10,40)(20,40)
\drawline(20,30)(40,30)
\drawline(105,30)(105,10)
\drawline(125,10)(125,0)
\color{gray}
\drawline(0,0)(0,100)
\drawline(0,0)(140,0)
\drawline(40,30)(40,100)
\drawline(40,30)(140,30)
}
\put(-250,40){
\put(0,0){\vector(1,0){20}}
\put(0,0){\vector(-1,0){20}}
}
\put(-450,0){
\thicklines
\color{blue!10}
\polygon*(0,0)(0,85)(10,85)(10,75)(20,75)(20,65)(40,65)(40,30)(70,30)(70,10)(90,10)(90,0)(0,0)
\color{blue!20!BurntOrange!15}
\polygon*(0,30)(0,65)(40,65)(40,30)
\color{blue}
\drawline(0,85)(10,85)
\drawline(10,75)(20,75)
\drawline(20,65)(40,65)
\drawline(70,30)(70,10)
\drawline(90,10)(90,0)
\color{gray}
\drawline(0,0)(0,100)
\drawline(0,0)(140,0)
\drawline(40,30)(40,100)
\drawline(40,30)(140,30)
}
%
%
\end{picture}
}
\setlength{\unitlength}{1pt}
{
\begin{picture}(200,50)(0,0)
\put(25,0){
\thicklines
\color{blue!10}
\polygon*(0,0)(0,47)(10,47)(10,37)(20,37)(20,27)(38,27)(38,10)(58,10)(58,0)(0,0)

\color{gray}
\drawline(0,0)(0,50)
\drawline(0,0)(70,0)
\drawline(40,30)(40,50)
\drawline(40,30)(70,30)
\color{black}
\drawline(30,30)(40,20)
\thinlines
\color{gray}
\multiput(0,0)(0,10){4}{\line(1,0){40}}
\multiput(0,0)(10,0){5}{\line(0,1){30}}
\thicklines
\color{blue}
\put(0,-3){
\drawline(0,50)(10,50)
\drawline(10,40)(20,40)
\drawline(20,30)(40,30)
}
\put(-68,0){
\drawline(105,30)(105,10)
\drawline(125,10)(125,0)
}
}
\put(125,0){
\thicklines
\color{blue!10}
\polygon*(0,0)(0,37)(10,37)(10,27)(20,27)(20,17)(38,17)(38,10)(58,10)(58,0)(0,0)
\color{gray}
\drawline(0,0)(0,50)
\drawline(0,0)(70,0)
\drawline(40,30)(40,50)
\drawline(40,30)(70,30)
\color{black}
\drawline(30,30)(40,20)
\thinlines
\color{gray}
\multiput(0,0)(0,10){4}{\line(1,0){40}}
\multiput(0,0)(10,0){5}{\line(0,1){30}}
\thicklines
\color{blue}
\put(0,-13){
\drawline(0,50)(10,50)
\drawline(10,40)(20,40)
\drawline(20,30)(40,30)
}
\put(-68,0){
\drawline(105,30)(105,10)
\drawline(125,10)(125,0)
}

}

\end{picture}
}

\caption{\label{fig:GenericYoung} Generalisation of the Young diagram to the case of non-integer weights. {\it Left:} Diagrams with boundaries  outside of the weight lattice. Two examples are shown that define isomorphic $\su(\groupq|\groupm)$ representations; the diagrams differ by the highlighted rectangles, the height of the left rectangle is equal to the width of the right rectangle. {\it Right:} Diagrams with boundaries  entering the weight lattice. The first example defines a unitary representation, the second defines a non-unitary representation (unless intersection happens on the diagonal of a plaquette, in this case the representation becomes shortened).}
\end{center}
\end{figure}
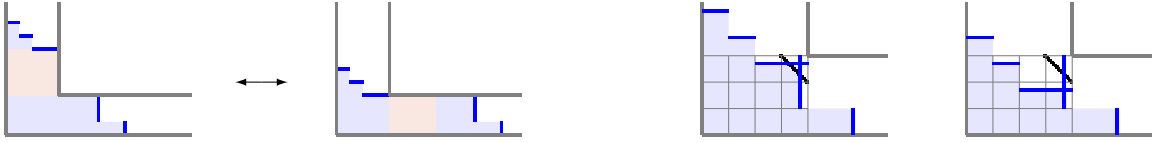

It is also formally possible to define a generalisation of Young diagram when $\wf_a+\wb_i$ take non-integer values. While it is useful for intuitive reasoning, we are not aware of any combinatorial value attached to it as in the integer-weight case. The generalised diagram is a shape inside the fat hook restricted by two boundaries as shown in Fig~\ref{fig:GenericYoung}. The upper boundary is such that the distance to horizontal segments of Kac-Dynkin path is given by $\wf$'s; the right boundary is such that the distance to vertical segments of Kac-Dynkin path is given by $\wb$'s. This time the weight change \eqref{eq:dualityrule1} is impossible, therefore we should allow negative distances in principle. However one should always have  $\wf_a+\wb_i>0$ for all the plaquettes to insure unitarity. This boils down to the following geometrical requirement assuming an appropriate choice of the overall shift $\Lambda$: Either two boundaries do not intersect and are such that the diagram fully covers the weight lattice (hence all the distances in any grading are positive, so their sums are positive as well); or two boundaries intersect, then this should happen above the diagonal of the upper-right plaquette of the weight lattice, see~Fig~\ref{fig:GenericYoung} ($\wf_\groupq$ or $\wb_\groupm$ become negative in certain gradings but their sum remain positive).

\renewcommand{\groupq}{\mathsf{q}}
\renewcommand{\groupm}{\mathsf{m}}
\section{Representations of oscillator algebra and their Hermitian structure}
\label{sec:repos}
The $\gl$ superalgebra can be realised by oscillators and, even more, this realisation can encode any real form of $\mathfrak{u}(\groupp,\groupq|\groupm)$ type. For instance, $\mathfrak{u}(\groupn)$ generators are given by
\be\label{basicoscillator}
\gE_{ij}=a_i^\dagger \cdot a_j^{\vphantom{+}}\,,
\ee
where $a$'s are standard bosonic oscillators, obeying $[a^{\vphantom{\dagger }}_i,a_j^\dagger ]=\delta_{ij}$, and we request that $a^\dagger $ is the conjugate of $a$.

We shall see that the oscillator algebra realisation significantly simplifies the study of unitarity. However, naive identification of the algebra of oscillators with the operators in a Fock module  leads to $\gl$-representations with integer weights only.  In this paper we present a novel generalisation  of the Fock space which we shall call the $\CF_{\gamma}$-module. It depends on a continuous parameter $\gamma$, with whose help we will construct all unitary representations and prove the classification statement in section~\ref{sec:classification}.

The goal of the current section is to introduce the $\CF_{\gamma}$-module and its essential properties. Prior to that, we have to build up an appropriate language and to review the standard knowledge about Fock space and the $\GL(K)-\GL(P)$ dual pair which is relevant to the realisation  \eqref{basicoscillator}.

\subsection{\label{sec:Fock}Fock space and  $\GL(K)-\GL(P)$ Howe dual pair}
Consider the Fock module $\cF^{\otimes(K\times P)}$ which is a representation of the oscillator algebra with $K$ "flavours" and $P$ "colours":
\be\label{comrel}
&& [a_i^{\vphantom{\dagger }}(A),a^\dagger _{j}(B)]=\delta_{ij}\,\delta_{AB}\,,
\ee
for $A,B,\ldots \in\overline{1,P}\,,\ \ i,j,\ldots\in\overline{1,K}\,.$ In the following  we denote this oscillator algebra as $\CH_{K\times P}$.

The Fock module is generated by the commutative action of raising (creation) operators $a^{\dagger }_i(A)$ on the Fock vacuum $\fvac$ distinguished by the property
\be\label{Fockvac}
&& a_i(A)\fvac=0\,,\ \forall\,i,A\,.
\ee
Hence it can be identified, as a vector space, with the space of  polynomials in $K\times P$ variables, i.e. $\mathbb{C}[x_{11},x_{12},\ldots,x_{KP}]$. The action of $a_i(A)$ on this space is fully determined by $\eqref{Fockvac}$ and \eqref{comrel} and it is simply given by the derivative $a_i(A)=\frac{\partial}{\partial x_{iA}}$.

One can introduce a structure of $\mathsf{GL}(K)\times\mathsf{GL}(P)$-module on the Fock space. To this end consider a $K\times P$ matrix ${\bf X}$ with entries ${\bf X}_{jB}=x_{jB}$. Then $\mathsf{GL}(K)\otimes\mathsf{GL}(P)$ simply acts by left$\times$right action on ${\bf X}$:
\be\label{GKGPaction}
{\bf X}\mapsto G_{K}\, {\bf X}\, G_{P}^{\rm T}\,.
\ee
The Fock space is an infinite sum of symmetric powers $\bigoplus\limits_{n=0}^\infty S^n(\mathbb C^{K\times P})$, where $\mathbb C^{K\times P}$ is parameterised by the entries of ${\bf X}$. Hence \eqref{GKGPaction} is extended in an obvious way to $\cF^{\otimes(K\times P)}$.

On the level of algebra, $\gl(K)$ action is explicitly realized by
\be\label{glkaction}
 \gE_{ij}= \sum_A\,a_i^\dagger (A)\,a_j^{\vphantom{\dagger }}(A)\,,
\ee
and $\gl(P)$ action is explicitly realised by
\be\label{glkactionc}
\gE_{AB}=\sum_i a_i^{\dagger }(A)\,a_i^{\vphantom{\dagger }}(B)\,.
\ee

In particular, $\gl(K)$ and $\gl(P)$  share the $\algu(1)$ generator which is simply the number operator
\be
\CN\equiv\sum_{i,A}a_i^{\dagger }(A)\,a_i^{\vphantom{\dagger }}(A)\,.
\ee
The structure of the Fock space as $\mathsf{GL}(K)\times\mathsf{GL}(P)$-module is under perfect control.
It is well-known that the group algebras of $\mathsf{GL}(K)$ and $\mathsf{GL}(P)$ are maximal mutual centralizers in $\cF^{\otimes(K\times P)}$ and are referred to as forming a Howe dual pair \cite{Howe}. The Fock space decomposes into the multiplicity-free direct sum of irreducible representations of $\mathsf{GL}(K)\times\mathsf{GL}(P)$
\be\label{Fockdecomp1}
\CF^{\otimes(K\times P)}=\bigoplus_{\yb}V_{K,\yb}\otimes V_{P,\yb}\,,
\ee
where summation is done over all  integer partitions
\be
\yb=\{\yb_1\geq \yb_2\geq\ldots\geq \yb_K\geq 0\}\,,
\ee
and where $V_{K,\yb}$ and $V_{P,\yb}$ are $\GL$-representations corresponding to the Young diagram $\yb$. It is hence easy to foresee that the non-zero entries in the sum \eqref{Fockdecomp1} should obey  $h_{\yb}\leq\min(K,P)$, where $h_{\yb}$ is the height of the Young diagram (number of non-zero $\yb$'s).

So far colour and flavour indices were treated on equal footing. But now we should choose one of them to describe representations of $\gl$-algebras. It will be our convention to consider $\gl(K)$ action defined by \eqref{glkaction}, which is the precise meaning of \eqref{basicoscillator}. In the following, we will denote an implicit summation over the colour indices by the dot between the operators, as in \eqref{basicoscillator}.

We can give an explicit realisation of any irreducible finite-dimensional representation $\yb$ as a subspace in an appropriate Fock space. Consider first two examples

\ \\
{\it Example 1.} Symmetric square representation is explicitly the space
\be\label{exV1}
V_{\raisebox{-1pt}{\yng(2)}}={\rm span}\left\{a_i^\dagger (1)a_j^\dagger (1)\fvac,\ \forall\,i,j\in\overline{1,K}\right\}.
\ee
Note that if we have more than one colour, there are several non-unique ways to realise such a representation. One can use anything of type ${\rm span}\left\{ \left[a_i^\dagger (A_1)a_j^\dagger (A_2)+a_j^\dagger (A_1)a_i^\dagger (A_2)\right]\!\fvac,\ \forall\,i,j\in\overline{1,K}\right\}$, but \eqref{exV1} suffices for our needs.

\ \\
{\it Example 2.} Antisymmetric square is realised on the space
\be
 V_{\raisebox{-4pt}{\yng(1,1)}}={\rm span}\left\{\left[a_i^\dagger (1)a_j^\dagger (2)-a_i^\dagger (2)a_j^\dagger (1)\right]\!\fvac,\ \forall\,i,j\in\overline{1,K}\right\}\,.
\ee
We need at least two colours to build such a representation. If there are more than two colours then this space, again, is not uniquely defined.

\ \\
In general, the construction  for partition $\yb$  is
\be\label{Vtau}
V_{\yb}^{/ A_{1\vphantom{h_{\yb}}}\ldots\, A_{h_\yb}/}={\rm span}\left\{{\mathsf c}_{\yb}\left[\prod_{x,y}a_{i_{x,y}}^\dagger (A_{x,y})\right]\fvac\,,\  \forall\, {i_{x,y}}\in\overline{1,K}\right\}\,,
\ee
where the product $\prod\limits_{x,y}\equiv \prod\limits_{y=1}^{h_{\mu}}\prod\limits_{x=1}^{\mu_y}$ runs over the boxes of the Young diagram and ${\mathsf c}_{\yb}$ is the Young symmetriser. We can take $A_{1,y}=A_{2,y}=\ldots =A_{\yb_y,y}$, these numbers are denoted by $A_y\equiv A_{1,y}$ on the l.h.s. of \eqref{Vtau}. The values of $A_y$ can be arbitrary but pairwise distinct. Therefore $P\geq h_{\yb}$.

 If the explicit value of $A$'s is irrelevant for the discussion, we will simply denote $V_{\yb}^{/A_{1\vphantom{h_{\yb}}}\ldots\, A_{h_\yb}/}$ as $V_{\yb}$. Note that $V_{K,\yb}\otimes V_{P,\yb}$ from decomposition \eqref{Fockdecomp1} is explicitly realised as ${\rm span}\left\{{\mathsf c}_{\yb}\left[\prod\limits_{x,y}a_{i_{x,y}}^\dagger (A_{x,y})\right]\fvac\,,\  \forall\,{i_{x,y},A_{x,y}}\right\}$. Hence $V_{\yb}^{/ A_{1\vphantom{h_{\yb}}}\ldots\, A_{h_\yb}/}$ is a subspace of $V_{K,\yb}\times V_{P,\yb}$. There is certain abuse of notation here, $V_{\yb}$ denotes either the class of isomorphic $\gl(K)$-representations, as in \eqref{Fockdecomp1}, or an explicit realisation from this class, as in \eqref{Vtau}. In the following, the meaning of $V_{\yb}$ will be clear from the context.

\ \\
Let us now discuss unitarity properties of the constructed representations. Once the norm of the vacuum is fixed to $\langle0\fvac=1$, the scalar product on the Fock space is fully determined by requirement that $a$ and $a^\dagger $ are conjugated with respect to it. Such scalar product is obviously positive-definite. On the other hand, it is obviously invariant under action of $\mathfrak{u}(K)$ algebra:
\be\label{Uinvariance}
 \langle v| \mathfrak{g}\, w \rangle =\langle \mathfrak g^\dagger  \,v|\,w \rangle\,,
\ee
for any states $v,w$ of the Fock space and any $\mathfrak{g}\in\algu(K)$ defined by \eqref{basicoscillator}.
Hence we can conclude immediately  that all  constructed representations are unitary for the $\mathfrak{u}(K)$ choice of the real form.

\subsection{$\CF_\gamma$-module}

The above-presented Fock space approach is a simple and powerful method to demonstrate unitarity of $\mathfrak{u}(K)$-modules. Nevertheless, it has one important shortcoming: It is not clear whether there are other representations which are unitary but not realisable inside the Fock space, hence we cannot use the Fock space alone to prove classification statements. If this question addresses only unitarity of $\su(K)$-modules, the problem is easily resolved. By direct analysis of $\su(2)$ sub-algebras one can quickly conclude that unitarity requires all Dynkin labels $\omega_i=\lambda_i-\lambda_{i+1}$ to be non-negative integers, and hence all such cases are realized in a Fock space with sufficient amount of colours.

However, when we take $\mathfrak{u}(1)$ central element $\CN$ into account, Fock space can produce only integer values for it whereas unitarity is possible for arbitrary real values of $\CN$ \footnote{We remind that integrability of the representation up to the group level is not required; or equivalently, one considers universal covers of the Lie groups which are simply connected, hence there are no topological restrictions which might impose $\CN$ to be integer.}.  This limitation of the Fock space construction is especially critical  in the supersymmetric and non-compact cases where not only the central charge but also some Dynkin labels can be non-integer.  To overcome this limitation we shall introduce the $\CF_\gamma$-module that allows one to control continuity in the spectrum.

As a warm up, consider a special case of the Fock module with $K=P$. In this case we can give a more detailed structure of the decomposition \eqref{Fockdecomp1}. It can be written as follows:
\be\label{Fdecomp}
\CF^{\otimes P^2}=\bigoplus_{n=0}^\infty\bigoplus_{\yb}  (V_{\yb}\otimes V_{\yb})\times(\Delta^\dagger )^n\fvac\,,
\ee
where $\yb$ is now a proper Young diagram, i.e. an integer partition $\{\yb_1\geq \yb_2\geq\ldots\yb_{P}\}$ with $\yb_P=0$, so that $h_\yb<P$. We also have the possibility to construct the determinant
\be\label{detdef}
\Delta^\dagger &\equiv &\det\limits_{1\leq i,A\leq P} a_i^\dagger (A)
\ee
which is invariant under $\sl(P)\oplus\sl(P)$ action but does not commute with the number operator:
\be
[\,\CN,\Delta^\dagger ]=P\,\Delta^\dagger \,.
\ee
Hence we can modify the value of the $\algu(1)$ charge by multiplication with $\Delta^\dagger $.

We now want to consider arbitrary values of the $\algu(1)$ charge. For this, we will need a well-defined object which for all practical purposes behaves as $(\Delta^\dagger )^\gamma$ where the parameter $\gamma$ can be non-integer. As the determinant is a scalar quantity, raising it to non-integer power can be a meaningful operation, one need only a proper mathematical formalisation of this idea which we will now describe.

For the oscillator algebra \eqref{comrel} with $K=P$, one defines the representation $\CF_{\gamma}^{P^2}$ as follows. Consider the vector space $\mathbb{C}[x_{11},x_{12},\ldots,x_{PP}]\otimes \mathbb C[t,t^{-1}]$, where $\mathbb C[t,t^{-1}]$ denotes Laurent polynomials in an additional auxiliary parameter $t$. Introduce, as before, matrix ${\bf X}$ with ${\bf X}_{iA}=x_{iA}$. Then, as a vector space, $\CF_{\gamma}^{P^2}$ is  identified with the following quotient
\be\label{quotient}
\CF_{\gamma}^{P^2} = \frac{\mathbb C[x_{11},x_{12},\ldots,x_{PP}]\otimes\mathbb C[t,t^{-1}]}{\det {\bf X}-t}\,,
\ee
and the action of $a^\dagger $ and $a$ on $\CF_{\gamma}^{P^2}$ is realised by:
\be\label{aaction}
a_i^\dagger (A)&=&x_{iA}\,,
\no\\
a_i(A)&=&\frac{\partial_{}\ \ }{\partial\,x_{iA}}+\left(\frac{\partial\det{\bf X}}{\partial\, x_{iA}\ }\right)\left(\frac{\gamma}{t}+\frac{\partial}{\partial t}\right)\,.
\ee
It is straightforward to check that \eqref{aaction} with arbitrary complex number $\gamma$ indeed defines a representation of the oscillator algebra $\CH_{P^2}$ on $\CF_{\gamma}^{P^2}$. To simplify notations, we will refer to it as the $\CF_{\gamma}$-module when the value of $P$ is clear from the context.

Identification \eqref{quotient} allows one to think about $\CF_{\gamma}$ as a ring of functions\footnote{In fact, this is the coordinate ring of  $\GL(P)$ considered as affine variety.}
 and hence to define  action of $\CF_{\gamma}$ on itself. Then multiplication by $x_{iA}$ is identified with action  by $a_i^{\dagger}(A)$. Multiplication by $t$ should realise the action by determinant $\Delta^\dagger $, and taking the quotient in \eqref{quotient} takes care of that. But additionally we allow for Laurent series in $t$ and formalise in this way action with negative (but integer) powers of $\Delta^\dagger$.

To motivate the definition of $a_i(A)$,  think about the identity monomial in \eqref{quotient} as the state $|\gamma\rangle\equiv(\Delta^\dagger )^\gamma\fvac$, where $\fvac$ is the usual "Fock vacuum".  Then
\be\label{aaction2}
a\act 1 \simeq a|\gamma\rangle=a (\Delta^\dagger )^\gamma\fvac=[a, (\Delta^\dagger )^\gamma]\fvac=\gamma [a,\Delta^\dagger ](\Delta^\dagger )^{\gamma-1}\fvac=[a,\Delta^\dagger ]\frac{\gamma}{t}|\gamma\rangle\simeq \frac{\partial\det{\bf X}}{\partial\, x\ }\frac \gamma t\,.
\ee
Here and below we use notation $\CO\act f$ to denote an operator $\CO$ acting on $f\in\CF_{\gamma}$, where $f$ is considered as a vector. Note that it is possible that $\CO\in\CF_{\gamma}$, then $\CO\act f=\CO\,f$.

The sequence \eqref{aaction2} explains the definition $a\act 1=\frac{\partial\det{\bf X}}{\partial\, x\ }\frac \gamma t$, however \eqref{aaction2}  is a well-defined sequence only for positive integer $\gamma$. As for an arbitrary $\gamma$, it is quite heuristic, as neither Fock vacuum is an element of $\CF_\gamma$ nor $(\Delta^{\dagger})^\gamma$ was rigorously defined, but very suggestive for developing the correct intuition. To get the full expression for $a_i(A)$ in \eqref{aaction}, take any monomial $M(x,t)\in\mathbb{\CF_{\gamma}}$,  consider  $a\act M(x,t)\simeq a\,M(a^{\dagger},\Delta^{\dagger})(\Delta^{\dagger})^
\gamma\fvac$ and repeat the same reasoning as in \eqref{aaction2}.\\[0em]

\subsection{Irreducibility and signature of scalar product}
\label{capelli}
We now proceed to establish the essential ingredient needed to determine the unitarity of a representation: the Hermitian structure of $\CF_{\gamma}$, which should be fixed from the requirement that $a$ and $a^\dagger$ are conjugate to one another. For the case of ordinary Fock space, this is an easy task for two reasons: first, it is a tensor product of one-oscillator Fock space, and, second, the Fock vacuum is a special state that is annihilated by all lowering operators. Neither of these features are present in $\CF_{\gamma}$-module, so our analysis will be more elaborate.

A preliminary question to address is irreducibility of the $\CF_{\gamma}$-module. First we note that this module admits a natural left-right action of $\mathsf{GL}(P)\times \mathsf{GL}(P)$ group, defined precisely by \eqref{GKGPaction}. $t$ transforms under this action in the same way as $\det{\bf X}$. It is then easy to see the following decomposition into $\mathsf{GL}(P)\times \mathsf{GL}(P)$ irreducible pieces, similarly to \eqref{Fdecomp}:
\be\label{Flamdecomp}
\CF_{\gamma}=\bigoplus_{n=-\infty}^\infty\bigoplus_{\yb}\,  (V_{\yb}\otimes
V_{\yb})\times t^n\,.
\ee
In the following we use $\CF_{\mu,\gamma}^{[n]}\equiv  (V_{\yb}\otimes
V_{\yb})\times t^n$.

A priori, the corresponding Lie algebra generators of $\gl(P)\oplus \gl(P)$ are not necessarily given by \eqref{glkaction} and \eqref{glkactionc}, because we modified the definition  of lowering oscillators. However, the explicit computation shows
\be\label{glkaction2}
 \gE_{ij}\equiv a_i^\dagger\cdot a_j=\sum_{A}x_{iA}\frac{\partial}{\partial x_{jA}}+\delta_{ij}\left(\gamma+t\frac{\partial}{\partial t}\right)\,,
\ee
\be\label{glkaction2c}
\gE_{AB}\equiv \sum_i a_i^{\dagger }(A)\,a_i^{\vphantom{\dagger }}(B)=\sum_{i}x_{iA}\frac{\partial}{\partial x_{iB}}+\delta_{AB}\left(\gamma+t\frac{\partial}{\partial t}\right)\,,
\ee
where $\sum_{A}x_{iA}\frac{\partial}{\partial x_{jA}}+\delta_{ij}t\frac{\partial}{\partial t}$ is a natural left $\gl(P)$ action  inherited from \eqref{GKGPaction}, and $\gE_{ij}$ is a $\gl(P)$ algebra that we would like to construct using oscillators. Equivalent comparison holds for right $\gl(P)$ action versus $\gE_{AB}$ generators.

We therefore see that  bilinears in oscillators \eqref{glkaction2} and \eqref{glkaction2c} generate precisely the natural rotation \eqref{GKGPaction} when restricted to restricted to the  $\sl(P)\oplus \sl(P)$ subalgebra, whereas it is only the $\algu(1)$ charge who is  shifted. We are now ready to realise our minimal goal of constructing an oscillator realisation of  a $\algu(P)$ representation which has the proper Young diagram $\mu$ and arbitrary value of the $\algu(1)$ charge. To this end one just needs to restrict to a required $V_{\mu}\times t^n$ subspace; it is, by construction,  invariant under $\gl(P)$ action,  has Hermitian structure consistent with $\algu(P)$ *-algebra, and arbitrariness of the  $\algu(1)$ charge  is controlled by $\gamma$:
\be
\CN \act V_{\mu}\times t^n =(P(\gamma+n)+|\mu|) V_{\mu} \times t^n\,.
\ee
The last result is of course intuitively expected as $\CN$ counts the total number of oscillators and, formally, number of oscillators in $t^n\simeq (\Delta^\dagger)^{\gamma+n}\fvac$ is $P(\gamma+n)$.

Coming back to irreducibility question, we would like to check whether for two arbitrary vectors $v,w\in\CF_{\gamma}$, there exist  $\mathcal{A}\in\mathcal{U}(\mathcal{H}_{P^2})$ (an element in universal enveloping algebra of the oscillator algebra) such that $w=\mathcal{A}\act v$. First note that $\CF_{\mu,\gamma}^{[n]}$ subspaces have distinct $\gl(P)\oplus\gl(P)$ Cartan charges, therefore we can always use combinations of oscillator bilinears \eqref{glkaction} and \eqref{glkactionc} to cook up a projector to a particular $\CF_{\mu,\gamma}^{[n]}$ subspace. Hence the question of irreducibility simplifies to whether oscillators allow one to move between any two $\CF_{\mu,\gamma}^{[n]}$ subspaces.

Now we note that  one has for any $n\in\mathbb{Z}$
\be\label{Fnp}
\CF^{n+}_{\gamma}\equiv \bigoplus_{m=n}^{\infty}\bigoplus_{\mu} \CF_{\mu,\gamma}^{[m]}\simeq \mathbb{C}[{\bf X}]\act t^n\,,
\ee
i.e. we can identify subspaces of type $\CF^{n+}_{\gamma}$ with spaces of polynomials in $P^2$ variables. Then it becomes clear that, starting from any vector $v\in \CF_{\mu,\gamma}^{[m]}\subset \CF^{n+}_{\gamma}$ (which can be viewed as a polynomial in $x$'s), and for any $\mu'$, we can find a polynomial $\CA$ in raising operators $a^\dagger\equiv x$ such that $\CA\act v$  can be projected to $\CF_{\mu',\gamma}^{[m']}$ for sufficiently large $m'$. Indeed, whatever $\mu$ is, one can always find $\tilde\mu$  such that $V_{\mu'}\subset V_{\tilde\mu}\otimes V_{\mu}$. Then we just take generic enough $\CA\in V_{\tilde\mu}$. $\CA\act v$ is just  multiplication of two polynomials, and it is  a vector in the tensor product of two $\sl(P)$ representations. Taking products can increase the $\algu(1)$ charge by some amount, that is why we say that $m'$ should be sufficiently large. Therefore there  exist maps between $\CF_{\mu,\gamma}^{[m]}$ and $\CF_{\mu',\gamma}^{[m']}$ for arbitrary pairs $\mu,\mu'$ but potentially with some restrictions on values of $m,m'$.

Our last step in proving irreducibility is to see whether we can always find a map from $\CF_{\mu,\gamma}^{[m]}$ to $\CF_{\mu,\gamma}^{[m\pm 1]}$ to remove any potential restrictions on $m,m'$. Increasing the value of $m$ is done in the obvious way:
\be
\Delta^{\dagger} \act \CF_{\mu,\gamma}^{[m]}= \CF_{\mu,\gamma}^{[m+1]}\,.
\ee
To lower the value of $m$, we cannot use action by $1/t$ as it is not an element of  $\mathcal{U}(\mathcal{H}_{P^2})$. Purely by Cartan charges counting, it should be that  $\Delta \act \CF_{\mu,\gamma}^{[m]}\propto  \CF_{\mu,\gamma}^{[m-1]}\,.$ But it may happen that $\Delta$ annihilates $ \CF_{\mu,\gamma}^{[m]}\,,$  therefore we need to compute action of $\Delta$ explicitly. This can be done using the Capelli identity
\be
        \Delta^\dagger \,\Delta =\det\limits_{1\leq i,j,\leq P}\left(\ \gE_{ij}+(P-i)\delta_{ij} \right)\,,
\ee
where the determinant of a matrix with non-commuting entries is defined by preserving the column order:  $\det
M\equiv \e_{i_1\ldots i_P} M_{i_11}\ldots M_{i_PP}$.

The Capelli determinant can be written as a polynomial in Casimirs of $\gl(P)$, hence its value can be evaluated by action on any vector of a $\gl(P)$ irrep. Computation of the determinant on the highest-weight vector is extremely simple as all $\gE_{ij}$ with $i<j$ annihilate this state, and the determinant is computed by the product of diagonal entries. Hence we have
\begin{align}\label{Capcons}
\Delta \act \CF_{\mu,\gamma}^{[m]}=\frac 1t \Delta^{\dagger} \Delta \act \CF_{\mu,\gamma}^{[m]}=\frac 1t \det\limits_{1\leq i,j,\leq P}\left(\ \gE_{ij}+(P-i)\delta_{ij} \right) \act \CF_{\mu,\gamma}^{[m]}=\left[\prod_{i=1}^{P}(\hat\mu_i+\gamma+m)\right]\CF_{\mu,\gamma}^{[m-1]}\,,
\end{align}
where we introduced the so-called shifted weights $\hat\mu_i\equiv \mu_i+P-i\geq 0$ that form a strictly decreasing sequence $\hat\mu_1>\hat\mu_2>\ldots $ with $\hat\mu_P=0$.

Note in particular that
\be
\Delta \act t^m =(\gamma+m)(\gamma+m+1)\ldots (\gamma+m+P-1)\, t^{m-1}\,.
\ee

We therefore conclude that  $\Delta\act \CF_{\mu,\gamma}^{[m]}\neq 0$ for any non-integer $\gamma$ and hence the $\CF_{\gamma}$-module is irreducible.

The module becomes reducible if $\gamma \in \mathbb{Z}$. We moreover note that for any given $\gamma$, all $\CF_{\gamma+\mathbb Z}$-modules are isomorphic. The isomorphism between $\CF_{\gamma}$ and $\CF_{\gamma+n}$, $n\in\mathbb{Z}$, is established by the mapping
\be\label{eq:ison}
\mathbb{C}[{\bf X},t,t^{-1}]\mapsto t^{n}\act \mathbb{C}[{\bf X},t,t^{-1}]\,,
\ee
or, informally, by considering $(\Delta^\dagger )^n|\gamma\rangle$ as a new reference state $|\gamma+n\rangle$.
Based on this isomorphism, it is enough to understand reducibility for $\gamma=0$. For this case, the irreducible subspace of $\CF_{\gamma=0}$ is  $\CF^{0+}_{\gamma=0}\equiv \mathbb{C}[{\bf X}]$ is precisely the ordinary Fock space $\CF^{\otimes P^2}$ with standard action of oscillators on it.

By now we learned enough to start analysing the Hermitian structure of the $\CF_{\gamma}$-module. For any real $\gamma$, each subspace $\CF_{\mu,\gamma}^{[n]}$ is a $\algu(P)\oplus \algu(P)$ unitary irreducible representation, its scalar product is uniquely defined up to an overall normalisation by $\algu(P)\oplus \algu(P)$ invariance, and it can be positive-definite with an appropriate sign choice of this normalisation. The question therefore is how to consistently choose normalisations in different subspaces $\CF_{\mu,\gamma}^{[n]}$  as they are related to one another by action of the oscillator algebra.

Define $|v_{\mu,n}\rangle\equiv v_{\mu}\act t^n$, for some $v_{\mu}\in V_\mu$. Then use \eqref{Capcons} to compute
\be\label{normequation}
 ||v_{\mu,{n+1}}||^2 = \langle v_{\mu,{n+1}}|\Delta^{\dagger} |v_{\mu,n}\rangle =\overline{\langle v_{\mu,{n}}|\Delta |v_{\mu,n+1}\rangle}=\left[\prod_{i=1}^{P}(\hat\mu_i+\gamma+n+1)\right] ||v_{\mu,n}||^2\,.
\ee
It is clear now that the Hermitian form of the whole $\CF_{\gamma}$-module cannot be positive-definite, as $\prod\limits_{i=1}^{P}(\hat\mu_i+\gamma+n+1)$  is negative for certain values of $n$. However, we also note the following useful  fact. Define the following subspaces:
\be
\CF_{\mu,\gamma}^+\equiv \bigoplus_{n\in\mathbb Z,n+\gamma>-1}\CF^{[n]}_{\mu,\gamma}\,.
\ee
Hermitian form can be made positive-definite on each of them, as, according to \eqref{normequation}, the sign of relative normalisations does not change among all $\CF_{\mu,\gamma}^{[n]}$ that comprise $\CF_{\mu,
\gamma}^+$.

We can also prove a stronger statement which is the main conclusion of this section:

 \begin{proposition}
It is possible to choose the  overall normalisation such that the subspace
 \be\label{defFgammap}
 \CF^+_{\gamma}\equiv \bigoplus_\mu \CF_{\mu,\gamma}^+
 \ee
 has a positive-definite Hermitian form.
 \end{proposition}
\begin{proof}
Due to isomorphism \eqref{eq:ison}, it is enough to consider $\gamma\in[0,-1)$.  The statement is obvious for $\gamma=0$ because $\CF_{\gamma=0}^+$ is just the ordinary Fock space. Let us now continuously decrease $\gamma$ starting from 0 while keeping the norm of the identity monomial  $1\simeq |\gamma\rangle$ equal to $1$. Suppose we reach a negative $\gamma>-1$ when some  vector $v\in \CF_{\gamma}^+$ has a zero norm. Expand $v$ in any  basis $\{v_{1},v_{2},\ldots\}$ in $\CF_\gamma^+$ with mutually orthogonal basis vectos: $v=\sum\limits_{\alpha}c_{\alpha}v_{\alpha}$, so that $||v||^2=\sum\limits_{\alpha}|c_{\alpha}|^2||v_{\alpha}||^2$. Since there are no vectors in $\CF_{\gamma}^+$ with negative norm yet, we conclude that $||v_{\alpha}||^2=0$ if $c_{\alpha}\neq 0$.  Then it is immediate to conclude that $\langle v|\CF_{\gamma}^+\rangle=0$. But also $v$ should be orthogonal to any vector from $\CF_{\mu,\gamma}^{[n]}\not\subset \CF_{\gamma}^+$, because all spaces $\CF_{\mu,\gamma}^{[n]}$ in decomposition \eqref{Flamdecomp} are mutually orthogonal as they have different Cartan charges. Hence $\langle v|\CF_{\gamma}\rangle=0$.

Let $\CA\in\mathcal{U}(\mathcal{H}_{P^2})$. Then $||\CA|v\rangle||^2=\langle v|\CA^{\dagger}\CA|v\rangle=0$ as $\CA^{\dagger}\CA|v\rangle\in\CF_{\gamma}$. Since $\CF_{\gamma}$ is an irreducible module for non-integer $\gamma$, any element of $\CF_{\gamma}$ can be represented as $\CA|v\rangle$, i.e. we have just derived that any state in $\CF_{\gamma}$ should have zero norm. But, for instance the norms of $t^n$ are non-zero, as it follows from our normalisation $\langle\gamma|\gamma\rangle=1$ and \eqref{normequation}. Hence zero-norm $v$ in $\CF_{\gamma}^+$ could not have existed, hence Hermitian form of  $\CF_{\gamma}^+$  cannot change its signature when performing analytic continuation with $\gamma$ (until $\gamma$ reaches next integer value).
\end{proof}

For practical applications of  the next section we will need the following \\[0em]

\noindent{\bf Corollary:}
{\it With normalisation $\langle \gamma|\gamma\rangle=1$, $\CF^{0+}_{\gamma}$ has  positive-definite Hermitian form for any $\gamma>-1$.}

\section{\label{sec:classification}Sufficient conditions for unitarity and classification theorems}
In this section we  give the complete classification of unitary highest-weight representations (UHW's) of $\su(\groupp,\groupq|\groupm)$. The necessary conditions for unitarity are given in section~\ref{sec:necescond} and, for the special case of $\groupm=0$, in appendix~\ref{sec:appa}. We will prove in this section that these conditions are also sufficient. To this end all the unitary representations are explicitly realised using oscillators based on results of section~\ref{sec:repos}. In the case of undeformed Fock space, the constructed representations are automatically unitary, whereas in the deformed case, we demonstrate that the constructed representations belong to the $\mathcal{F}_{\gamma}^{0+}$ submodule of $\mathcal{F}_{\gamma}$ which is also a sufficient condition for unitarity if $\gamma>-1$.

We will mostly work in the grading $\su(\groupp,\!|\groupm|\groupq)$ throughout this section. Transition to other gradings was discussed before, and we come back to this question once more after introduction of non-compact Young diagrams in subsection~\ref{sec:ncYoung}.

 \subsection{\label{sec:supq}$\su(\groupp,\groupq)$}
A limiting case of $\gl(\groupn|\groupm)$ algebras is a non-supersymmetric $\gl(\groupn)$. It would be instructive to describe UHW representations of its real form $\su(\groupp,\groupq)$, with $\groupn=\groupp+\groupq$, as an opening illustration of our approach. These representations were classified in \cite{MR733809}. Note  that, in contrast to $\gl(\groupn|\groupm)$ case with $\groupn,\groupm\neq 0$, UHW representations of $\su(\groupp,\groupq)$ do not exhaust all possible unitary representations of real forms of $\gl(\groupn)$, where the complete classification is still an open problem.

We think about $\su(\groupp,\groupq)$ as the traceless part of  $\algu(\groupp,\groupq)$, while $\algu(\groupp,\groupq)$ would be naturally constructed with the help of oscillators. One needs two sets of bosonic oscillators: $a_\alpha(A)$, $\alpha\in\overline{1,\groupq}$ and $b_{\dot\alpha}(A)$, ${\dot{\alpha}}\in\overline{1,\groupp}$. The colour index $A$ runs from 1 to $P$.

The $\algu(\groupp)$  subalgebra of $\algu(\groupp,\groupq)$ is realised as $\gE_{\dot\alpha\dot\beta}=-b_{\dot\alpha}^{\vphantom{\dagger}}\cdot b_{\dot\beta}^\dagger $, and the $\algu(\groupq)$ subalgebra is realised as $\gE_{\alpha\beta}=a_{\alpha}^\dagger \cdot a_\beta^{\vphantom{\dagger}}$.   All the  generators of  $\algu(\groupp,\groupq)$ are given by \cite{Gunaydin:1981yq}
\be\label{supqrealisation}
       \mtwo{\gE_{\dot\alpha\dot\beta}}{\gE_{\dot\alpha\beta}}{\gE_{\alpha\dot\beta}}{\gE_{\alpha\beta}}=\mtwo{-b_{\dot\alpha}^{\vphantom{\dagger}}\cdot b_{\dot\beta}^\dagger }{b_{\dot\alpha}\cdot a_{\beta}^{\vphantom{\dagger}}}{-a_{\alpha}^\dagger \cdot b_{\dot\beta}^\dagger }{a_{\alpha}^\dagger
\cdot a_{\beta}^{\vphantom{\dagger}}}\,.
\ee
In \eqref{supqrealisation}, we used two alternative realisations of $\algu(\groupn)$ algebras: $\gE_{ij}=a_i^\dagger \cdot a_j$ and $\gE_{ij}=-b_i\cdot b_j^\dagger $, where the former one was discussed in section~\ref{sec:repos}. They are linked through the automorphism $\phi_{\rm out}:\gE_{ij}\mapsto - \gE_{ji}$, with a minor adjustment: $-\gE_{ji}=-a_{j}^\dagger \cdot a_i=P\,\delta_{ij}-a_i\cdot a_j^\dagger $, and we drop the central element $P\,\delta_{ij}$ which neither affects commutation relations of $\algu(\groupn)$ nor enters the $\su(\groupn)$ generators.

The maximal compact subalgebra of $\su(\groupp,\groupq)$ is $\CK=\su(\groupp)\oplus\su(\groupq)\oplus\algu(1)$, with $\algu(1)$ generator given by $\En = {1 \over \groupq} \CN_a + {1 \over \groupp} \CN_b + P\,.$

The highest-weight representation $U$ is conveniently described as follows \cite{Gunaydin:1981yq}: We decompose $\su(\groupp,\groupq)$ as
\be\label{eq:EKE}
\su(\groupp,\groupq)=\CE^{(-)}\oplus\CK\oplus \CE^{(+)}\,,
\ee
where $\CE^{(+)}$ is spanned by $\gE_{\dot\alpha\beta}$ and $\CE^{(-)}$ is the conjugated subspace. Consider a subspace $U_0\subset U$ which is an irreducible $\CK$-module characterised by the property $\CE^{(+)}U_0=0$. $U_0$ contains the highest-weight state of $U$, hence $U_0$ would uniquely specify the whole representation $U$. In fact, $U$  is the induced representation $\mathcal{U}(\su(\groupp,\groupq))\otimes_{\mathcal{U}(\mathcal K\oplus \CE^{(+)})}U_0$ with subsequent quotient over maximal proper sub-module. The latter quotient  is taken automatically once the representation is realised in a space with positive-definite Hermitian form, in particular this is the case for the oscillator constructions discussed below. Indeed, any proper sub-module, if it were present, should have all states of zero norm which is contradictory.

Following the discussion of section~\ref{sec:repos}, $U_0$ is explicitly realised as
\be\label{eq:U0def}
U_0\simeq V_{\yb_R}^{/{\bf A}/}[a] \, V_{\yb_L}^{/{\bf B}/}[b] \, [\Delta^\dagger_a]^{\gamma_R}\,[\Delta^\dagger_b]^{\gamma_L}\fvac\,.
\ee
Here $\yb_R=\{\yb_R^1\geq\yb_R^2\geq\ldots\geq \yb_R^\groupq\}$, $\yb_L=\{\yb_L^1\geq\yb_L^2\geq\ldots\geq \yb_L^\groupp\}$ are integer partitions. The small letters $a,b$ designate the type of oscillators used to construct $V$ or $\Delta^\dagger$, as defined by \eqref{Vtau} \footnote{To be more precise, we should remove $\fvac$ from the definition \eqref{Vtau} to get  $V_\mu$ that we use here. I.e. we understand $V_{\mu}$ as an appropriate space of polynomials in raising operators.} and \eqref{detdef}. The subsets ${\bf A}\subset \{1,\ldots,P\}$, ${\bf B}\subset \{1,\ldots,P\}$ denote the colours used in the construction; It is important that  ${\bf A}\cap{\bf B}=\emptyset$, otherwise $\CE^{(+)}U_0\neq 0$. Therefore the minimal number of colours we need is $P=h_{\mu_L}+h_{\mu_R}$. Parameters $\gamma_R,\gamma_L$ can be non-integers if the number of colours is sufficiently large as is discussed below.

From \eqref{eq:U0def}, it is not difficult to compute the fundamental weight of the highest-weight state in $U_0$:
\be\label{funweigthsupq}
[\nu_L;\nu_R]=[-\yb_L^\groupp,\ldots,-\yb_L^2,-\yb_L^1;\yb_R^1,\yb_R^2,\ldots \yb_R^\groupq]\ +\
[\mdash (-P-\gamma_L) \mdash;\mdash \gamma_R  \mdash]\,;
\ee
the first $\groupp$ entries define the weight of $\algu(\groupp)$ and the last $\groupq$ entries -- the weight of $\algu(\groupq)$. Here we introduced notation $\mdash \alpha \mdash\equiv \alpha,\alpha,\ldots,\alpha$, where the number of repetitions of $\alpha$  should fill in the pattern dictated by the context.

By adjusting $\gamma_L,\gamma_R$ appropriately, we can always restrict ourselves to consideration of proper integer partitions $\yb_L,\yb_R$, with $\yb_L^\groupp=\yb_R^\groupq=0$. This will be our convention from now on. Also, arbitrary translation $\nu_\mu\to \nu_\mu+\Lambda$ does not affect the irrep of $\su$ algebra. Hence, an arbitrary irrep of $\su(\groupp,\groupq)$ is completely and uniquely parameterised by proper partitions $\yb_L,\yb_R$ and one extra parameter  which we choose to be
\be
\beta=P+\gamma_L+\gamma_R\,.
\ee
In the following we will identify $U$ with the triple $[\yb_L,\yb_R;\beta]$.

The parameter $\beta$ is related to the eigenvalue of the $\algu(1)$ generator on HWS by $\beta=\En-\frac 1\groupq\sum_{\alpha=1}^\groupq \yb_R^\alpha-\frac 1\groupp\sum_{\dot\alpha=1}^\groupp \yb_L^{\dot\alpha}\,.$ In fact, $\beta$  is nothing but the HWS eigenvalue of $\gE_{\groupn\groupn}-\gE_{11}$. As $-(\gE_{\groupn\groupn}-\gE_{11})$ is the Cartan element in the associated non-compact $\su(1,1)$ sub-algebra, we can deduce that $\beta\geq 0$ is a necessary condition for unitarity of representation, however stronger constraints shall be imposed.\\[0em]

The action of $\algu(\groupp,\groupq)$ is realised by \eqref{supqrealisation} inside a certain representation of the oscillator algebra which we will denote as $\CF_{\rm total}$.

Consider first the $\gamma_L=\gamma_R=0$ case.  Then  $\CF_{\rm total}=\CF^{\otimes P\times(\groupp+\groupq)}$, i.e. this is an ordinary Fock space. We will need that $U_0\subset U$  naturally belongs to  a certain smaller subspace of $\CF_{\rm total}$. To demonstrate this fact, represent $\CF_{\rm total}$ as
\be\label{CFtotal}
\CF_{\rm total}=\CF^{\otimes \groupq\, h_{\yb_R}}\otimes\CF^{\otimes \groupp\, h_{\yb_L}}\otimes \CF^{\otimes r}\,,
\ee
where $h_{\yb}$ is the height of Young diagram $\yb$, i.e. number of non-zero $\yb_i$. Then $V_{\yb_R}[a]$ can be realised using only oscillators from the first factor $\CF^{\otimes \groupq h_{\yb_R}}$, $V_{\yb_L}[b]$ -- from the second factor in \eqref{CFtotal}. In this respect, $\CF^{\otimes r}$ denotes oscillators which were not used in \eqref{eq:U0def} but which are switched on by acting with generators from $\CE^{(-)}$ on $U_0$.

$\CF_{\rm total}$ is  a tensor product of elementary non-deformed Fock spaces, hence it  has a positive-definite norm. Therefore any representation  $[\yb_L,\yb_R;\beta=P]$, with integer $P\geq h_{\yb_L}+h_{\yb_R}$ is unitary.

We now generalise to the case of non-zero $\gamma$'s.  Let us focus, for concreteness, on the first factor in \eqref{CFtotal}. If we have enough colours to imbed  $\CF^{\otimes \groupq h_{\yb_R}}$ inside the larger space $\CF^{\otimes \groupq^2}$, then we can $\gamma$-deform the latter and get the possibility to operate with determinants $\Delta^\dagger_a$ in fractional powers. The minimal number of colours for this to happen is $P=\groupq+h_{\yb_L}$, as $h_{\yb_L}$ colours is still needed for $b$-oscillators,  to preserve $\CE^{(+)}U_0= 0$. If this condition is met, we build our representation inside the space
\be\label{eq:Ftotdef}
\CF_{\rm total}=\CF_{\gamma_R}^{\groupq^2}\otimes\CF^{\otimes \groupp\, h_{\yb_L}}\otimes \CF^{\otimes r}\,,
\ee
and it would be parameterised by\footnote{We have chosen precisely $P=\groupq+h_{\yb_L}$, more colours won't be needed for question of classification of unitary representations because the moment $t_a$ is switched on, we can increase $\gamma_R$ instead of increasing $P$ to control the value of $\beta$. This freedom in changing $\gamma$'s vs changing $P$ is emphasised more explicitly in the forthcoming most general case.} $\beta=\groupq+h_{\yb_L}+\gamma_R$. The precise meaning of \eqref{eq:U0def} is the following: The state $\fvac$ in \eqref{eq:U0def} is  the product of Fock vacua $\fvac_{\groupq^2}\otimes\fvac_{\groupp\,h_{\yb_L}}\otimes\fvac_{r}$ according to decomposition \eqref{eq:Ftotdef}. Then $[\Delta^\dagger_a]^{\gamma_R}\fvac_{\groupq^2}\equiv |\gamma_R\rangle\in\CF_{\gamma_R}^{\groupq^2}$.

The whole representation $U$ is {\it guaranteed} to lie within $\CF_{\rm total}^+\equiv\CF_{\gamma_R}^{0+}\otimes\CF^{\otimes \groupp\, h_{\yb_L}}\otimes \CF^{\otimes r}$ space, where $\CF_{\gamma}^{0+}$ was introduced in \eqref{Fnp}. Indeed, elements of $U_0$ space belong to $\CF_{\rm total}^+$ by construction, whereas action of $\CE^{(-)}$ on $U_0$  includes only raising $a_\alpha^\dagger $ operators, for what concerns action on $\CF_{\gamma_R}$ component of the tensor product, and we stay within $\CF_{\gamma}^{0+}$ when acting with raising operators.

As $U$ is restricted to the subspace $\CF_{\rm total}^+$, $U$  has  positive-definite Hermitian form if $\gamma_R>-1$. Hence the representation $[\yb_L,\yb_R;\beta]$ with any real $\beta>\groupq+h_{\yb_L}-1$ is automatically unitary.

Instead of turning on $\Delta^\dagger_a$, one could turn on $\Delta^\dagger_b$ or both of the determinants, assuming we have enough colours\footnote{In the most general case with both $\gamma_L,\gamma_R$ non-zero, we have $\CF_{\rm total}=\CF_{\gamma_R}^{\groupq^2}\otimes\CF_{\gamma_L}^{\groupp^2}\oplus\CF^{\otimes r}$, and the precise meaning of \eqref{eq:U0def} stems from identification $[\Delta^\dagger_a]^{\gamma_R}\,[\Delta^\dagger_b]^{\gamma_L}\fvac\equiv |\gamma_R\rangle\otimes|\gamma_L\rangle\otimes\fvac_{r}$.}. As $\beta$ is determined only by the sum $P+\gamma_L+\gamma_R$, it is not essential who is precisely turned on, and it is convenient to turn only the determinant that requires the least number of colours to use, as it would allow to prove unitarity for  the largest possible continuous range of  $\beta$.

We have outlined a constructive proof of unitarity of a large class of $\su(\groupp,\groupq)$ highest-weight representations, including those with non-integer $\beta$.
That this construction  exhausts all possible UHW representations of $\su(\groupp,\groupq)$ is proved in appendix~\ref{sec:appa}.

In summary, one has proved the following  classification theorem (known already e.g. in \cite{MR733809}):
\begin{theorem} All the unitary highest-weight representations of the Lie algebra $\su(\groupp,\groupq)$  are parameterized by the triple $[\yb_L,\yb_R;\beta]$, where $\yb_L=\{\yb_L^1\geq\yb_L^2\geq\ldots\geq \yb_L^\groupp=0\}$, $\yb_R=\{\yb_R^1\geq\yb_R^2\geq\ldots\geq \yb_R^\groupq=0\}$ are two proper integer partitions and $\beta$ is a real number satisfying the following constraints: $\beta\geq(h_{\yb_L}+h_{\yb_R})$, and  $\beta$ must  also be an integer if $\beta\leq \left( \min(\groupp+h_{\yb_R},\groupq+h_{\yb_L})-1\right)\,.$

The fundamental weight of the $\su(\groupp,\groupq)$ irrep is related to the triple $[\yb_L,\yb_R;\beta]$  by \eqref{funweigthsupq}, where only the sum $\beta=P+\gamma_L+\gamma_R$ (but not $\gamma_R$ and $\gamma_L+P$  separately) matters for defining an $\su(\groupp,\groupq)$ irrep. \label{th:supq}
\end{theorem}

\subsection{$\su(\groupm|\groupq)$ and $\su(\groupp,\!|\groupm)$}
In this subsection we construct unitary representations of  $\su(\groupm|\groupn)$ and $\su(\groupm,\!|\groupn)$ *-algebras with the help of oscillators. Given compactness of the corresponding real form, these representations should be finite-dimensional. The unitary representations of $\su(\groupm|\groupn)$ shall be called covariant representations, due to their transformation properties with respect to bosonic subalgebra action. The unitary rerpresentations of $\su(\groupm,\!|\groupn)$ shall be called contra-variant representations.

The $\su(\groupm|\groupn)$ and $\su(\groupm,\!|\groupn)$ *-algebras are  isomorphic, with isomorphism map $\varphi_{\rm out}$  defined in \eqref{outer}. Hence we do not really need to consider both of them, but the reason we do so nevertheless is that they  appear as subalgebras of non-compact $\su(\groupp,\!|\groupm|\groupq)$ that we are going to analyse next. To emphasize this intended imbedding, we will use notations $\su(\groupm|\groupq)$ and $\su(\groupp,\!|\groupm)$ throughout this subsection.

Also note that covariant and contra-variant representations can be considered as representations of the same Lie algebra if we ignore the *-structure. These representations are not isomorphic because  $\varphi_{\rm out}$ is a Lie algebra outer automorphism. This is another reason to mention both of these cases.

\renewcommand{\groupn}{\mathsf{m}}
\renewcommand{\groupm}{\mathsf{q}}
\paragraph{Oscillator construction for covariant case.} Choose $\su(\groupn|\groupm)$ grading and use the following oscillator representation
\be\label{sunmosci}
       \mtwo{\gE_{ab}}{\gE_{a\beta}}{\gE_{\alpha b}}{\gE_{\alpha\beta}}= \mtwo{f_a^\dagger \cdot f_b^{\vphantom{\dagger}}}{f_a^\dagger \cdot a_\beta^{\vphantom{\dagger}}}{a_\alpha^\dagger \cdot f_b}{a_\alpha^\dagger \cdot a_\beta^{\vphantom{\dagger}}}\,.
\ee
Here one encounters bosonic oscillators $a_\alpha(A)$, the same ones as before, and also fermionic oscillators $f_a(A)$ which obey the anti-commutation relation $\{f_{a}^{\vphantom{\dagger}}(A),f_{b}^{\dagger}(B) \}=\delta_{ab}\delta_{AB}$.

The fermionic Fock space can be viewed as a space of polynomials in Grassmann variables $\theta_{a\,A}$, with generators $f$ realised as
\be
f_{a}^{\dagger}(A) =\theta_{a\,A}\,\ \ \ \ f_{a}^{\vphantom{\dagger}}(A)=\frac{\partial{\hphantom{\theta}}}{\partial\,\theta_{a\,A}}\,.
\ee
The $\algu(\groupn)$ action  is realised on this Fock space by $\gE_{ab}=f_{a}^\dagger \cdot f_{b}^{\vphantom{\dagger}}$, with dot symbolising sum over colours, as before. We can single out an irreducible $\algu(\groupn)$-module $W_{\yf}$ with partition $\yf=\{\yf_1\geq \yf_2\geq \ldots\geq \yf_{\groupn}\}$ as follows
\be\label{Vtauferm}
W_{\yf}^{/ A_{1\vphantom{h_{\yf_1}}}\ldots A_{\yf_1}/}={\rm span}\left\{{\bf P_{\yf}}\left[\prod_{x,y}\theta_{a_{x,y}\,A_{x,y}}\right]\,,\  \forall\, {a_{x,y}}\in\overline{1,\groupn}\right\}\,,
\ee
with the main difference from the bosonic case that  colours should be distinct within the same row and may be the same in different columns. Hence the minimal number of colours required is equal to $\yf_1$. This minimal choice, with $A_{x}\equiv A_{x,1}=A_{x,2}=\ldots$, is denoted as the superscript of $W$.

One cannot construct a determinant from Grassmann variables, the closest analog is the product of all fermions of {\it the same colour}
\be
\Delta_f^\dagger (A)=\theta_{1A}\,\theta_{2A}\,\ldots \,\theta_{\groupn A}\,.
\ee
This pseudo-determinant is invariant under $\su(\groupn)$ action. Although it can be constructed already when only one colour is available, we cannot raise it to any power and hence cannot use it as a source for non-integer weights. It is the determinant of bosonic oscillators that will be the only source of a continuous deformation.

\ \\
The even subalgebra of $\sunm$ is $\CK=\su(\groupn)\oplus\su(\groupm)\oplus\algu(1)$. It is a compact algebra, with $\algu(1)$ generator given by $\gE= \frac 1{\groupn}\sum_{a=1}^{\groupn}\gE_{aa}+\frac 1{\groupm}\sum_{\alpha=1}^{\groupm}\gE_{\alpha\alpha}=\frac{1}\groupn\CN_f+\frac1\groupm\CN_a\,.$

Consider decomposition
\be
\su(\groupn|\groupm)=\CE^{(-)}\oplus\CK\oplus \CE^{(+)}\,,
\ee
where $\CE^{(+)}$ is spanned by $\gE_{a\beta}$ and $\CE^{(-)}$ is the conjugate space. Consider a subspace $U_0$ in the highest-weight representation $U$ which is an irreducible $\CK$-module characterised by the property $\CE^{(+)}U_0=0$. As before, the full representation $U$ can be viewed   as $U=\mathcal{U}(\su(\groupn|\groupm))\otimes_{\mathcal{U}(\mathcal K\oplus \CE^{(+)})}U_0$ with the subsequent quotient by the maximal proper submodule. This quotienting, if any, is done automatically in the oscillator realisation.

The irreducible $\CK$-module $U_0$ is uniquely defined by two partitions $\yf=\{\yf_1,\ldots,\yf_{\groupn}\}$ and $\yb=\{\yb_1,\ldots,\yb_{\groupm}\}$ that label representations with respect to, respectively, $\su(\groupn)$ and $\su(\groupm)$ subalgebras, and the eigenvalue of $\gE$. It is realised in the  Fock space (or its appropriate $\gamma$-deformation) as
\be\label{U0sumq}
U_0\simeq W_{\yf}^{/{\bf F}/}V_{\yb}^{/{\bf A}/} \left[\prod_{C\in {{\bf F}_{\!\Delta}}}\Delta_f^\dagger (C)\right][\Delta^{\dagger}_a]^{\gamma}\fvac\,,
\ee
with precise meaning of $[\Delta^{\dagger}_a]^{\gamma}\fvac$ the same as was explained in the $\su(\groupp,\groupq)$ case.

 ${\bf A},{\bf F},{\bf F}_{\!\Delta}\subset \{1,\ldots, P\}$ are certain subsets of colours. If $\gamma\neq 0$, denote by ${\bf A}_{\!\Delta}$ the subset of colours used in construction of $\Delta^{\dagger}_a$. For transparency of construction, we require  that ${\bf A}\subseteq{\bf A}_{\!\Delta}$. If $\gamma=0$, define ${\bf A}_{\!\Delta}={\bf A}$.

Without loss of generality, we can assume that $\yf$ and $\yb$ are proper integer partitions, and this will be our convention from now on. Indeed, if $\yf_\groupn\neq 0$, the  space $W_{\yb}$ contains $\Delta_f^\dagger (F_1)\ldots \Delta_f^\dagger (F_{\yf_{\groupn}})$ which can be considered as a part of $\prod_{C}\Delta_f^\dagger (C)$. Similarly, $\yb_\groupm\neq 0$ means that the space $V_{\yb}$ contains $\Delta^{\dagger}_a$ which can be absorbed to $[\Delta^{\dagger}_a]^{\gamma}$ by redefinition of $\gamma$. Note that with this convention one gets $|{\bf F}|=\tau_1$, where $|{\bf S}|$ means the number of colours in a colour subset ${\bf S}$.

It should be that ${\bf F}\cap {\bf F}_{\!\Delta}=0$ to get $U_0\neq 0$. It should be that ${\bf A}_{\!\Delta}\subseteq{\bf F}_{\Delta}$ to get $\CE^{(+)}U_0=0$. This implies  $|{\bf F}_{\!\Delta}|\geq h_{\yb}$, or $|{\bf F}_{\!\Delta}|\geq \groupm$ if $\gamma\neq 0$.

The fundamental weight of the constructed representation is
\be
\label{fundaright}
[\wf;\wb]=[\yf_1,\yf_2,\ldots,\yf_{\groupn}=0;\yb_1,\yb_2,\ldots,\yb_{\groupm}=0]+[\mdash |{\bf F}_{\!\Delta}|\mdash; \mdash\gamma\mdash]\,.
\ee
Define $\beta_{R}$ as the eigenvalue\footnote{A more close analog to $\su(\groupp,\groupq)$ case would be the eigenvalue of $\gE_{11}+\gE_{\groupn+\groupm,\groupn+\groupm}$. Our choice has the advantage that the unitarity constraint on $\beta_{R}$ will be given in terms of $\yb$ only.} of  $\gE_{\groupn,\groupn}+\gE_{\groupn+\groupm,\groupn+\groupm}$ on the highest-weight state, i.e.
\be
\label{betaright}
\beta_{R}=|{\bf F}_{\!\Delta}|+\gamma\,.
\ee
It is related to the eigenvalue of $\gE$ as $\gE=\frac 1{\groupn}|\tau|+\frac 1{\groupm}|\mu|+\beta_R\,.$

The number of colours used in \eqref{U0sumq} is  $|{\bf F}|+|{\bf F}_{\!\Delta}|$. The total number of colours in  the oscillator algebra shoud be at least $|{\bf F}|+|{\bf F}_{\!\Delta}|$, but otherwise it is an irrelevant number. Indeed, if there are some colours that are not used in \eqref{U0sumq}, they will not be switched on in the construction of the whole representation $U$, and will not affect the value of the representation weight.

Possible values of ${\beta_R}$ depend only on $\groupm$ and $h_{\yb}$: If $|{\bf F}_{\!\Delta}|<\groupm$, one should have $\gamma=0$ and  $\beta_{R}$ acquires only integer values $h_{\yb},h_{\yb}+1,h_{\yb}+2,\ldots$. The corresponding representations are unitary because they are realised inside a non-deformed Fock space. The moment $|{\bf F}_{\!\Delta}|=\groupm$, we can construct determinant from bosonic oscillators and hence switch on $\gamma$.  By precisely the same argument   as in $\su(\groupp,\groupq)$ case (construction of $\CF_{\gamma}$-module and invariance of $\CF_{\gamma}^{0+}$-subspace under algebra action), we prove unitarity for any real $\beta_R\geq \groupm-1$.

\begin{figure}[t]
\centering
 \begin{minipage}[b]{0.5\textwidth}
\begin{picture}(160,120)(-50,0)
\thicklines
\color{gray!50}
\drawline(-40,0)(155,0)
\drawline(60,40)(155,40)
\multiput(0,10)(0,10){3}{\line(1,0){145}}
\multiput(0,50)(0,10){7}{\line(1,0){60}}
\drawline(0,0)(0,125)
\drawline(60,0)(60,125)
\multiput(10,0)(10,0){5}{\line(0,1){115}}
\multiput(70,0)(10,0){8}{\line(0,1){40}}
\color{blue!5}
\polygon*(0,0)(0,100)(10,100)(10,80)(30,80)(30,70)(40,70)(40,60)(60,60)(60,40)(80,40)(80,30)(110,30)(110,20)(120,20)(120,10)(130,10)(130,0)

\color{blue}
\thicklines
\drawline(0,0)(0,100)(10,100)(10,80)(30,80)(30,70)(40,70)(40,60)(60,60)(60,40)(80,40)(80,30)(110,30)(110,20)(120,20)(120,10)(130,10)(130,0)(0,0)
\color{black}
\dottedline{2}(0,60)(40,60)
\dottedline{2}(60,0)(60,40)
\dottedline{2}(80,0)(80,30)
\dottedline{2}(-20,100)(0,100)
\thinlines
\put(25,120){\vector(-1,0){25}}
\put(35,120){\vector(1,0){25}}
\put(27,118){$\groupn$}
\put(150,25){\vector(0,1){15}}
\put(150,15){\vector(0,-1){15}}
\put(147,19){$\groupm$}
\put(100,5){\vector(-1,0){20}}
\put(113,5){\vector(1,0){17}}
\put(103,3){$\yb_1$}
\put(70,25){\vector(-1,0){10}}
\put(70,25){\vector(1,0){10}}
\put(66,17){$\gamma$}
\put(-15,55){\vector(0,1){45}}
\put(-15,35){\vector(0,-1){35}}
\put(-45,41){$|{\bf F}|\!+\!|{\bf F}_{\!\Delta}|$}
\put(5,80){\vector(0,1){20}}
\put(5,70){\vector(0,-1){10}}
\put(1,79){{\rotatebox{-90}{$\yf_1$}}}
\put(5,30){\vector(0,1){30}}
\put(5,30){\vector(0,-1){30}}
\put(7,28){$\left| {\text{\footnotesize\bf F$_{\!\Delta}$}}\right|$}
\end{picture}
\caption{\label{Fig:partitionfat1} Young diagram for a covariant representation of $\algu(\groupn|\groupm)$. It is assembled from two rectangles,  partition $\yb$, and $\yf\trnp$ -- transposition of partition  $\yf$. When $\gamma$ is integer, the obtained shape is itself  an integer partition  $Y\equiv \{\mu_1+\gamma+\groupn,\ldots,\mu_{\groupm-1}+\gamma+\groupn,\gamma+\groupn,\groupn,\ldots,\groupn,\yf\trnp_1,\yf\trnp_2,\ldots\}$, precisely the one from Fig.~\ref{fig:samplefathook}.}
\end{minipage}
\hfill
\begin{minipage}[b]{0.4\textwidth}
{
\begin{picture}(150,100)(-110,-100)
\thicklines
\color{gray!50}
\drawline(40,0)(-105,0)
\drawline(0,-40)(-105,-40)
\multiput(0,-10)(0,-10){3}{\line(-1,0){95}}
\multiput(0,-50)(0,-10){4}{\line(-1,0){40}}
\drawline(0,0)(0,-95)
\drawline(-40,0)(-40,-95)
\multiput(-10,0)(-10,0){3}{\line(0,-1){85}}
\multiput(-50,0)(-10,0){5}{\line(0,-1){40}}
\color{gray!5}
\polygon*(-40,-30)(-20,-30)(-20,-40)(-10,-40)(-10,-60)(0,-60)(-40,-60)
\color{gray}
\drawline(-40,-30)(-20,-30)(-20,-40)(-10,-40)(-10,-60)(0,-60)(-40,-60)(-40,-30)
\color{blue!5}
\polygon*(0,0)(-80,0)(-80,-10)(-70,-10)(-70,-20)(-40,-20)(-40,-30)(-20,-30)(-20,-40)(-10,-40)(-10,-60)(0,-60)
\color{blue}
\thicklines
\drawline(0,0)(-80,0)(-80,-10)(-70,-10)(-70,-20)(-40,-20)(-40,-30)(-20,-30)(-20,-40)(-10,-40)(-10,-60)(0,-60)(0,0)

\color{black}
\dottedline{2}(0,-60)(40,-60)
\dottedline{2}(-40,0)(-40,-20)
\thinlines
\put(-25,-90){\vector(-1,0){15}}
\put(-15,-90){\vector(1,0){15}}
\put(-24,-92){$\groupn$}
\put(-100,-15){\vector(0,1){15}}
\put(-100,-25){\vector(0,-1){15}}
\put(-103,-22){$\groupp$}
\put(-65,-5){\vector(-1,0){15}}
\put(-52,-5){\vector(1,0){12}}
\put(-63,-7){$\yb_1$}
\put(15,-20){\vector(0,1){20}}
\put(15,-40){\vector(0,-1){20}}
\put(4,-32){$P\!-\!|{\bf F}_{\!\Delta}|$}
\put(-35,-40){\vector(0,1){10}}
\put(-35,-50){\vector(0,-1){10}}
\put(-37,-41){{\rotatebox{-90}{$\yf_1$}}}
\end{picture}
\caption{\label{Fig:partitionfat2} Young diagram for a contravariant representation. More accurately, it describes UHW  of $\algu(\groupp,\!|\groupn)$ realised by \eqref{contraoscillators} with $\gE_{ij}\to \gE_{ij}+P\delta_{ij}$.  \eqref{contraoscillators} with this shift is related to \eqref{sunmosci} by $\varphi_{\rm out}$ and rearrangement of elements. }
}
\end{minipage}
\end{figure}
The constructed unitary representations can be graphically labeled by the Young diagram as shown in Fig.~\ref{Fig:partitionfat1}. This diagram always fits inside the fat hook domain. Also, since only $\beta_R$ defines an irrep of $\su(\groupn|\groupm)$, we can coherently modify $|{\bf F}_{\!\Delta}|$ and $\gamma$ by an integer amount, so as not to change $\beta_R$. The diagrams related by this transformation define isomorphic $\su(\groupn|\groupm)$ representations, cf. Fig~\ref{fig:GenericYoung}. Note that the transformation is only feasible if  restrictions $|{\bf F}_{\!\Delta}|\geq \groupm$ and $\gamma>-1$ are satisfied.

One can now check that we have constructed explicitly all the representations  allowed by the necessary conditions of unitarity from section~\ref{sec:plaquette}. The summarising classification theorem is given below jointly with the contra-variant case that we are going to discuss now.
\renewcommand{\groupn}{\mathsf{p}}
\renewcommand{\groupm}{\mathsf{m}}
\paragraph{Oscillator construction for contra-variant case. } Consider the $\su(\groupn,\!|\groupm)$ grading and the following oscillator realisation
\be\label{contraoscillators}
       \mtwo{\gE_{\dot\alpha\dot\beta}}{\gE_{\dot\alpha b}}{\gE_{a\dot\beta}}{\gE_{ab}}= \mtwo{-b_{\dot\alpha}^{\vphantom{\dagger}}\cdot  b_{\dot\beta}^\dagger }{b_{\dot\alpha}^{\vphantom{\dagger}}\cdot f_b^{\vphantom{\dagger}}}{-f_a^{\dagger }\cdot b^{\dagger }_{\dot\beta}}{f_a^{\dagger }\cdot\,f_b^{\vphantom{\dagger}}}.
\ee
The procedure repeats in full analogy. Consider proper partitions $\yb={\{\yb_1\geq\ldots\geq\yb_\groupn=0\}}$ and
$\yf=\{\yf_1\geq\yf_2\geq\ldots\geq\yf_\groupm=0\}$, and the following subspace in the  Fock space (or its appropriate $\gamma$-deformation)
\be\label{U0supm}
U_0\simeq V_{\yb}^{/{\bf B}/}W_{\yf}^{/{\bf F}/}\left[\prod_{C\in {{\bf F}_{\!\Delta}}}\Delta_f^\dagger (C)\right][\Delta^{\dagger}_b]^{\gamma}\fvac\,.
\ee
The whole representation $U$ is constructed as the induced representation from $U_0$ (in practice, $U$ is generated by action of $\gE_{a\dot\beta}=-f_a^{\dagger }\cdot b^{\dagger }_{\dot\beta}$ on the elements of $U_0$). Construction of such $\su(\groupn,\!|\groupm)$ representation is mapped to a Young diagram which is shown in Fig.~\ref{Fig:partitionfat2}.

In \eqref{U0supm}, $|{\bf B}|=h_{\yb}$, $|{\bf F}|=\yf_1$. The set ${\bf B}_{\!\Delta}\supseteq{\bf B}$ is determined in full analogy to ${\bf A}_{\!\Delta}$ from the covariant case, in particular $|{\bf B}_{\!\Delta}|=h_{\yb}$ if $\gamma=0$  and $|{\bf B}_{\!\Delta}|=\groupn$ if $\gamma\neq 0$. The value of $|{\bf F}_{\!\Delta}|$ can be arbitrary. To satisfy $\CE^{(+)}U_0=0$, the sets ${\bf B}_{\!\Delta},{\bf F},{\bf F}_{\!\Delta}\subset \{1,\ldots,P\}$ should not intersect.

The fundamental weight of the constructed from \eqref{U0supm} representation is given by
\be
\label{fundaleft}
[\wb;\wf]=[-\mu_\groupn=0,\ldots, -\mu_1;\yf_1,\ldots,\yf_\groupm=0]+[\mdash (-P-\gamma)\mdash;\mdash|{\bf F}_{\!\Delta}|\mdash]\,.
\ee
Define $\beta_L$ as the eigenvalue of $-\gE_{11}-\gE_{\groupn+1,\groupn+1}$ on the highest-weight state. One finds
\be
\label{betaleft}
\beta_L=P-|{\bf F}_{\!\Delta}|-\yf_1+\gamma=|{\bf B}_{\!\Delta}|+\delta P+\gamma\,,
\ee
where $\delta P$ is the number of colours that were not used in construction of $U_0$. Differently from the covariant case, these colours will be activated when generating the whole representation $U$, and their existence affects the representation weight.

If  $\gamma=0$  then $\beta_L$ can acquire only integer  values $\beta_L=h_{\yb},h_{\yb}+1,\ldots$.  If $\gamma\neq 0$ then $\beta_L=\groupn+\delta P+\gamma$ can be continuous. The  representations are unitary if $\gamma>-1$. Again, one can check against constraints of section~\ref{sec:plaquette} that all unitary representations are exhausted by the proposed construction.

We summarise our analysis in  the following theorem.

\renewcommand{\groupn}{\mathsf{n}}
\renewcommand{\groupm}{\mathsf{m}}

\begin{theorem}\label{lr} All unitarisable finite-dimensional  representations of  $\sl(\groupn|\groupm)$ complex Lie algebra, i.e. finite-dimensional representations  that can be made  unitary by an  appropriate choice of *-operation and corresponding invariant Hermitian form, are described as follows:

Representations that are unitary with respect to $\su(\groupm|\groupq)$ *-algebra (covariant representations) are parameterized by the triple $[\yf,\yb;\beta_R]$.  $\yb$ is a proper integer partition $\yb={\{\yb_1\geq\yb_2\geq\ldots\geq\yb_\groupq=0\}}$, $\yf$ is a proper integer partition $\yf=\{\yf_1\geq\yf_2\geq\ldots\geq\yf_\groupm=0\}$. The allowed values of  $\beta_R$  are integers in the range  $\groupq-1\geq\beta_R\geq h_{\yb}$ and reals for $\beta_R>\groupq-1$. The fundamental weight is given by \eqref{fundaright}, where only the combination \eqref{betaright} defining $\beta_R$ enters in the  definition of an $\su(\groupm|\groupq)$ irrep.

Representations that are unitary with respect to $\su(\groupp,\!|\groupm)$ *-algebra (contra-variant representations) are parameterized by the triple $[\yb,\yf;\beta_L]$.  $\yb$ is a proper integer partition $\yb={\{\yb_1\geq\yb_2\geq\ldots\geq\yb_\groupp=0\}}$, $\yf$ is a proper integer partition $\yf=\{\yf_1\geq\yf_2\geq\ldots\geq\yf_\groupm=0\}$. The allowed values of  $\beta_L$  are integers in the range  $\groupp-1\geq\beta_L\geq h_{\yb}$ and reals for $\beta_L>\groupp-1$. The fundamental weight is given by \eqref{fundaleft}, where only the combination \eqref{betaleft} defining $\beta_L$ enters in the  definition of an $\su(\groupp,\!|\groupm)$ irrep.
\end{theorem}

Let us illustrate this classification on a simple example. Suppose that all bosonic Dynkin labels are zero, i.e. the fundamental weight is  $[\mdash\omega_{\groupn}\mdash;\mdash0\mdash]$. Then the representation is uniquely determined by the fermionic Dynkin label $\o_{\groupn}=-\beta_{L}=\beta_{R}$. $\o_\groupn=0$ corresponds to a trivial representation. $\o_{\groupn}>0$ (integer for $\o_{\groupn}\leq\groupm-1$ and real otherwise) describes a unitary covariant representation, with positive-definite Hermitian form for  the $\su(\groupn|\groupm)$ choice of grading. $\o_\groupn<0$ (integer for $-\o_{\groupn}\leq\groupn-1$ and real otherwise) describes a contravariant representation, with positive-definite Hermitian form for  the $\su(\groupn,\!|\groupm)$ choice of grading.

Another example is the defining representation with fundamental weight $[1,0,\ldots,0;\mdash 0\mdash]$. It is a covariant representation described by the Young diagram with one box. All covariant representations with integer weights are contained in the tensor powers of the defining representation. This feature is the analog of the fact that all $\su(\groupn)$ representations are contained in the tensor powers of the $\su(\groupn)$ defining representation. The conjugate (to the defining irrep) representation is a contra-variant representation. It has the fundamental weight $[\mdash 0\mdash;0,\ldots,0,1]\simeq [\mdash 1\mdash;1,\ldots,1,0]$, and is also described by the Young diagram with one box, but on the inversed hook of Fig.~\ref{Fig:partitionfat2}. All contra-variant representations with integer weights are contained in the tensor powers of the conjugate representation. Note that the adjoint representation has the fundamental weight $[1,0,\ldots,0;0,\ldots,0,1]$ and this is not a unitary representation.

We emphasise that our discussion and theorem~\ref{lr} use language of *-algebras but not of real forms. On the level of  the $\su(\groupn|\groupm)$  real form, there is no natural distinction between co- and contra-variance. Instead, we should speak about two classes of representations, with self- or anti-self-conjugate odd (fermionic) generators in a Hermitian basis.  There are two inequivalent  ways to restrict to a real form starting from a *-algebra, see section~\ref{sec:defunitary}. One way  would link covariant representations with representations with self-conjugate odd generators, and contra-variant representations  with representations with anti-self-conjugate odd generators. The other way would do the opposite. Both ways are equally good, we can swap them by applying $\varphi_{\rm out}$ to either the real form or to the *-algebra.

\subsection{\label{sec:supqm}$\su(\groupp,\groupq|\groupm)$}
It is most convenient  to work with $\su(\groupp,\!|\groupm|\groupq)$ grading. In this grading, only UHW representations are possible, cf. section \ref{sec:allarehw}. We will construct all the UHW using the following oscillator realisation
\be
\label{oscfull}
\begin{pmatrix}
\gE_{\dot\alpha\dot\beta} & \gE_{\dot\alpha b} & \gE_{\dot\alpha \beta}
\\
\gE_{a\dot\beta} & \gE_{a b} & \gE_{a\beta}
\\
\gE_{\alpha\dot\beta} & \gE_{\alpha b} & \gE_{\alpha \beta}
\end{pmatrix}
=
\begin{pmatrix}
-b_{\dot \alpha}^{\vphantom\dagger}\cdot b^\dagger_{\dot \beta} & b_{\dot\alpha}^{\vphantom\dagger}\cdot f_{b}^{\vphantom\dagger} & b_{\dot\alpha}^{\vphantom\dagger}\cdot a_{\beta}^{\vphantom\dagger}
\\
-f^\dagger _a\cdot b^\dagger _{\dot\beta} & f^\dagger _a\cdot f_b^{\vphantom\dagger} & f^\dagger _a\cdot a_{\beta}^{\vphantom\dagger}
\\
-a^\dagger _{\alpha} \cdot b^\dagger _{\dot\beta} & a^\dagger _{\alpha}\cdot f_{b}^{\vphantom\dagger}& a^\dagger _{\alpha}\cdot a_{\beta}^{\vphantom\dagger} \\
\end{pmatrix}\,.
\ee
The algebra is decomposed as $\su(\groupp,\!|\groupm|\groupq)=\CE^{(-)}\oplus\CK\oplus\CE^{(+)}$, where $\CK$ is spanned by the traceless part of $\algu(\groupp)\oplus\algu(\groupm)\oplus\algu(\groupq)$ subalgebra, $\CE^{(+)}$ is spanned by the generators $\gE_{\dot\alpha b}$, $\gE_{\dot\alpha \beta}$, $\gE_{a\beta}$, and $\CE^{(-)}$ is the space conjugate to $\CE^{(+)}$. We define  $U_0$ as an irreducible $\CK$-module such that $\CE^{(+)}U_0=0$ and the full representation $U$ is constructed as the induced irreducible representation from $U_0$.

The space $U_0$ is realised as follows
\be\label{U0final}
U_0\simeq V_{\yb_L}^{/{\bf B}/}W_{\yf}^{/{\bf F}/}V_{\yb_R}^{/{\bf A}/} \left[\prod_{C\in {{\bf F}_{\!\Delta}}}\Delta_f^\dagger (C)\right][\Delta^{\dagger}_b]^{\gamma_L}[\Delta^{\dagger}_a]^{\gamma_R}\fvac\,,
\ee
where $\yb_L,\yf,\yb_R$ are proper integer partitions describing representations of, respectively, $\su(\groupp)$, $\su(\groupm)$, $\su(\groupq)$ subalgebras.

The colour subsets ${\bf A}_{\!\Delta}$ and ${\bf B}_{\!\Delta}$ are defined as before. To satisfy $\CE^{(+)}U_0=0$ and $U_0\neq 0$, we require that ${\bf A}_{\!\Delta}\subset {\bf F}_{\!\Delta}$ and that ${\bf F}_{\!\Delta},{\bf B}_{\! \Delta},{\bf F}$ do not intersect pairwise.

The number of elements in  ${\bf B},{\bf F},{\bf A}$ is $|{\bf B}|=h_{\yb_L}$, $|{\bf F}|=\yf_1$, $|{\bf A}|=h_{\yb_R}$. The size of ${\bf A}_{\!\Delta}$ is either $\groupq$ ($\gamma_R\neq 0$) or $h_{\yb_R}$ ($\gamma_R=0$). The size of ${\bf B}_{\!\Delta}$ is either $\groupp$ ($\gamma_L\neq 0$) or $h_{\yb_L}$ ($\gamma_L=0$). The size of ${\bf F}_{\!\Delta}$ is not fixed, but it is constrained from below by $|{\bf A}_{\!\Delta}|$. Total number of colours should of course be enough to construct $U_0$, but it can exceed this amount. Any extra colours, even if they are not present in $U_0$, will be switched on in the whole representation $U$ because of action of $\gE_{\alpha\dot\beta}$ generators.

We therefore conclude that the described oscillator realisation of $U$ is controlled by three proper integer partitions $\yb_L,\yf,\yb_R$, and by four extra parameters $\gamma_L,\gamma_R,|{\bf F}_{\!\Delta}|,P$ whose values are constrained but not fully fixed by  $\yb_L,\yf,\yb_R$. We can use this parameterisation to describe representations of $\algu(\groupp,\groupq|\groupm)\oplus \algu(1)$ algebra, with the fundamental weight of $\algu(\groupp,\groupq|\groupm)$ given by
\begin{align}\label{fundafinal}
[\wb_L;\wf;\wb_R]=[0,-\yb^{\groupp-1}_L,\ldots,-\yb_L^1;\yf_1,\ldots,\yf_{\groupm-1},0;\yb^1_R,\ldots,\yb^{\groupq-1}_R,0]+[\mdash (-P-\gamma_L)\mdash;\mdash|{\bf F}_{\!\Delta}|\mdash;\mdash\gamma_R\mdash]\,,
\end{align}
and $P$ being the value of the extra $\algu(1)$ charge.

For what concerns $\su$ algebra, introduce $\beta_L$  as the eigenvalue of $-\gE_{11}-\gE_{\groupp+1,\groupp+1}$, and $\beta_R$  as the eigenvalue of $\gE_{\groupp+\groupm,\groupp+\groupm}+\gE_{\groupp+\groupm+\groupq,\groupp+\groupm+\groupq}$ on the HWS. One gets
\begin{subequations}
\label{betafinal}
\begin{alignat}{5}
& \beta_L &&= \gamma_L+P-|{\bf F}_{\!\Delta}|-\yf_1 &&=\gamma_L+|{\bf B}_{\!\Delta}|+\delta P\,,
\\
&\beta_R &&= \gamma_R+|{\bf F}_{\!\Delta}| &&=\gamma_R+|{\bf A}_{\!\Delta}|+(|{\bf F}_{\!\Delta}|-|{\bf A}_{\!\Delta}|)\,.
\end{alignat}
\end{subequations}
The $\su$ irreps are uniquely defined by the tuple $[\yb_L,\yf,\yb_R;\beta_L,\beta_R]$.

To analyse unitarity of a representation, we can fully benefit from the advantages of the oscillator approach. Clearly, the unitarity constraints on $\beta_L$ and $\beta_R$ as described in theorem~\ref{lr} should be satisfied because $\su(\groupm|\groupq)$ and $\su(\groupp,\!|\groupm)$ are subalgebras of our algebra.  For all $\beta_L$,$\beta_R$ allowed by theorem~\ref{lr}, the construction \eqref{U0final} of $U_0$ can be done, and $\gamma_L,\gamma_R>-1$. By observing the oscillator structure of $\CE^{(-)}$ generators we notice that $U$ still belongs to the $\CF_{\rm total}^+$ subspace even though it is a non-compact infinite-dimensional representation. Hence we immediately deduce that $U$ is a unitary representation, i.e. no new constraints should be imposed on $\beta_L$, $\beta_R$. Hence we've proved the following

\begin{theorem}
\label{th:final}
All possible unitary representations of $\su(\groupp,\groupq|\groupm)$ *-algebra  are labeled by  tuples $[\yb_L,\yf,\yb_R;\beta_L,\beta_R]$, where $\yb_L,\yf,\yb_R$ are  proper integer partitions $\yb_L=\{\yb_L^1,\ldots,\yb^\groupp_L\!=\!0\}$, $\yf=\{\yf_1,\ldots,\yf_{\groupm}\!=\!0\}$, $\yb_R=\{\yb_R^1,\ldots,\yb^{\groupq}_R\!=\!0\}$, and  $\beta_{L},\beta_{R}$ are two numbers that  satisfy the following constraints: $\beta_R\geq h_{\yb_R}$, $\beta_R$ should be  integer for $\beta_R\leq \groupq-1$ and it may be any real number if $\beta_R>\groupq-1$; $\beta_L\geq h_{\yb_L}$, $\beta_L$ should be integer for $\beta_L\leq\groupp-1$ and it may be any real number if $\beta_L>\groupp-1$.

For the $\su(\groupp,\!|\groupm|\groupq)$ choice of grading these representations are of highest-weight type, with the fundamental weight given by \eqref{fundafinal}, where only combinations $\beta_L,\beta_R$ of  $\gamma_L,\gamma_R,|{\bf F}_{\!\Delta}|,P$ defined by \eqref{betafinal} matter for the definition of the representation isomorphism class.

\end{theorem}
\noindent

\paragraph{$\psu(\groupp,\groupq|\groupp+\groupq)$ case.} When $\groupm=\groupp+\groupq$, $\su(\groupp,\groupq|\groupm)$ is not simple, the simple algebra $\psu(\groupp,\groupq|\groupp+\groupq)$ is obtained as a quotient  of $\su(\groupp,\groupq|\groupm)$ by the central charge $\mathcal{C}=\sum\limits_{i=1}^{2\groupm}{E_{ii}}$. Classification of UHW for $\psu(\groupp,\groupq|\groupp+\groupq)$ is  reduced to theorem~\ref{th:final} plus an extra constraint $-|\yb_L|+|\yf|+|\yb_R|-\groupp(\beta_L+\yf_1)+\groupq\,\beta_R=0\,.$

\subsection{\label{sec:ncYoung}Non-compact and extended Young diagrams}
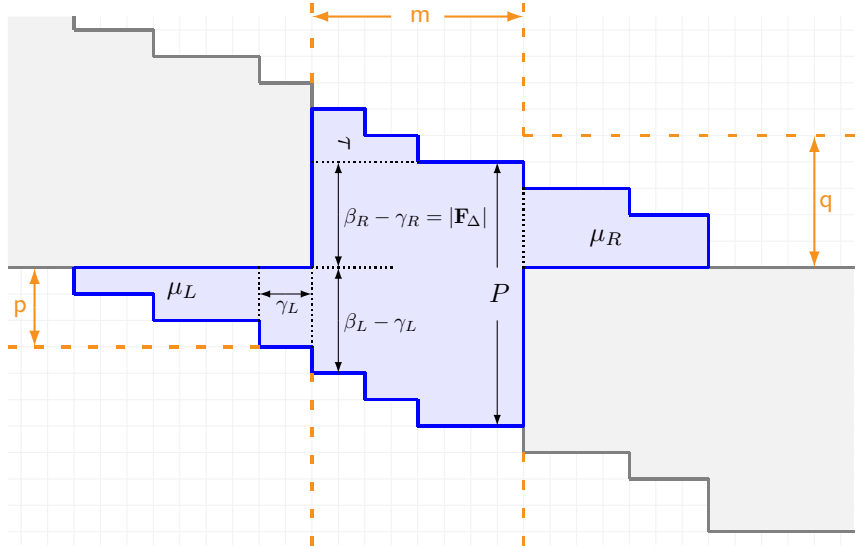
\begin{figure}[t]
\begin{center}
\begin{picture}(330,190)(-120,-100)
\color{gray!10}
\thinlines
\multiput(-115,-100)(0,10){20}{\line(1,0){320}}
\multiput(-110,-105)(10,0){32}{\line(0,1){200}}
\thicklines
\thicklines
\color{gray!10}
\polygon*(-115,95)(-90,95)(-90,90)(-60,90)(-60,80)(-20,80)(-20,70)(0,70)(0,0)(-115,0)
\polygon*(205,0)(80,0)(80,-70)(120,-70)(120,-80)(150,-80)(150,-100)(205,-100)
\color{gray}
\drawline(-90,95)(-90,90)(-60,90)(-60,80)(-20,80)(-20,70)(0,70)(0,0)(-115,0)
\drawline(205,0)(80,0)(80,-70)(120,-70)(120,-80)(150,-80)(150,-100)(205,-100)
%
\color{blue!10}
\polygon*(0,0)(0,60)(20,60)(20,50)(40,50)(40,40)(80,40)(80,30)(120,30)(120,20)(150,20)(150,0)(80,0)(80,-60)(40,-60)(40,-50)(20,-50)(20,-40)(0,-40)(0,-30)(-20,-30)(-20,-20)(-60,-20)(-60,-10)(-90,-10)(-90,0)(0,0)
\color{blue}
\drawline(0,0)(0,60)(20,60)(20,50)(40,50)(40,40)(80,40)(80,30)(120,30)(120,20)(150,20)(150,0)(80,0)(80,-60)(40,-60)(40,-50)(20,-50)(20,-40)(0,-40)(0,-30)(-20,-30)(-20,-20)(-60,-20)(-60,-10)(-90,-10)(-90,0)(0,0)
\color{BurntOrange}
\dashline{4}(80,50)(205,50)
\put(190,0){
\put(0,25){\vector(0,1){25}}
\put(0,25){\vector(0,-1){25}}
\put(2,23){$\groupq$}
}
\dashline{4}(-20,-30)(-115,-30)
\put(-105,0){
\put(0,-15){\vector(0,1){15}}
\put(0,-15){\vector(0,-1){15}}
\put(-8,-17){$\groupp$}
}
\dashline{4}(0,70)(0,100)
\dashline{4}(80,50)(80,100)
\dashline{4}(0,-40)(0,-105)
\dashline{4}(80,-70)(80,-105)
\put(30,95){\vector(-1,0){30}}
\put(50,95){\vector(1,0){30}}
\put(37,93){$\groupm$}
\color{black}
\thinlines
\dottedline{2}(-20,0)(-20,-20)
\dottedline{2}(0,0)(0,-30)
\dottedline{2}(0,0)(30,0)
\dottedline{2}(0,40)(40,40)
\dottedline{2}(80,0)(80,30)

\put(70,-40){
\put(0,40){\vector(0,1){40}}
\put(0,20){\vector(0,-1){40}}
\put(-3,27){$P$}
}
\put(-10,-35){
\put(0,25){\vector(-1,0){10}}
\put(0,25){\vector(1,0){10}}
\put(-3.5,19){\scalebox{0.8}{$\gamma_L$}}
}

\put(5,-60){
\put(5,80){\vector(0,-1){20}}
\put(5,80){\vector(0,1){20}}
\put(7,77){\scalebox{0.8}{$\beta_R-\gamma_R=|{\bf F}_{\!\Delta}|$}}
}
\put(5,-100){
\put(5,80){\vector(0,-1){20}}
\put(5,80){\vector(0,1){20}}
\put(7,77){\scalebox{0.8}{$\beta_L-\gamma_L$}}
}

\put(10,50){\rotatebox{-90}{$\yf$}}
\put(-55,-10){$\yb_L$}
\put(105,10){$\yb_R$}

\end{picture}

\caption[Non-compact Young diagram]{\label{fig:noncYoung}Shaded blue region (in the centre) -- a non-compact Young diagram which is bijectively related to the construction of $U_0$ space \eqref{U0final}. Several different diagrams correspond to the same isomorphism classes of $\su(\groupp,\groupq|\groupm)$ representations, see \eqref{isomoves}.
Both shaded blue and gray regions form the extended Young diagram.}
\end{center}
\end{figure}
%
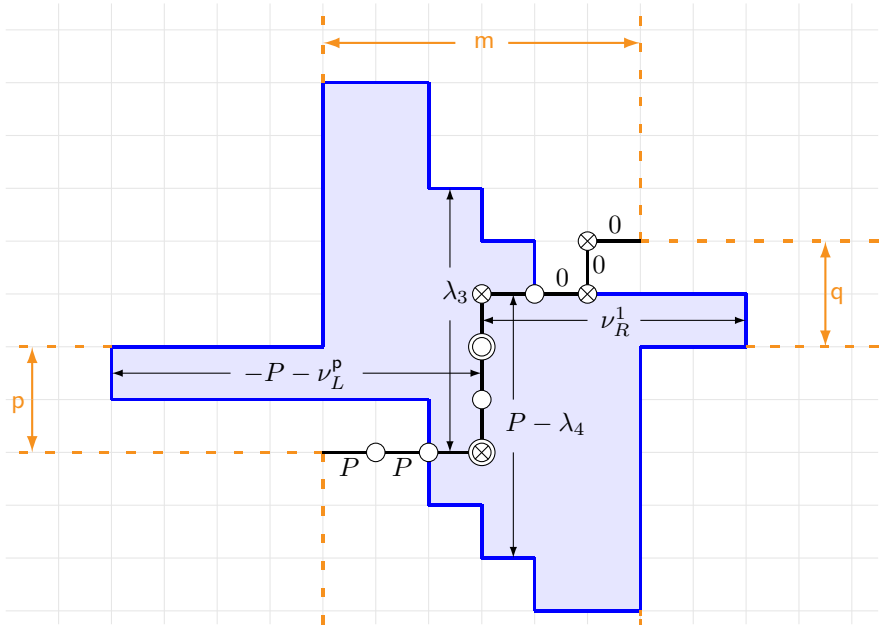
\begin{figure}[t]
\begin{center}
\begin{picture}(330,190)(-120,-120)
\color{gray!20}
\thinlines
\multiput(-120,-100)(0,20){12}{\line(1,0){330}}
\multiput(-100,-105)(20,0){16}{\line(0,1){235}}
\thicklines
\color{blue!10}
\polygon*(0,0)(0,100)(40,100)(40,60)(60,60)(60,40)(80,40)(80,20)(160,20)(160,0)(120,0)(120,-100)(80,-100)(80,-80)(60,-80)(60,-60)(40,-60)(40,-20)(-80,-20)(-80,0)(0,0)
\color{blue}
\drawline(0,0)(0,100)(40,100)(40,60)(60,60)(60,40)(80,40)(80,20)(160,20)(160,0)(120,0)(120,-100)(80,-100)(80,-80)(60,-80)(60,-60)(40,-60)(40,-20)(-80,-20)(-80,0)(0,0)

\color{BurntOrange}
\dashline{4}(160,0)(210,00)
\dashline{4}(120,40)(210,40)
\put(190,0){
\put(0,20){\vector(0,1){20}}
\put(0,20){\vector(0,-1){20}}
\put(2,18){$\groupq$}
}
\dashline{4}(-80,0)(-115,0)
\dashline{4}(00,-40)(-115,-40)
\put(-110,0){
\put(0,-20){\vector(0,1){20}}
\put(0,-20){\vector(0,-1){20}}
\put(-8,-23){$\groupp$}
}
\dashline{4}(0,100)(0,125)
\dashline{4}(120,40)(120,125)
\dashline{4}(0,-40)(0,-105)
\dashline{4}(120,-100)(120,-105)
\put(50,115){\vector(-1,0){50}}
\put(70,115){\vector(1,0){50}}
\put(57,113){$\groupm$}

\color{black}
\drawline(0,-40)(60,-40)(60,20)(100,20)(100,40)(120,40)

\color{white}
\put(20,-40){\circle*{7}}
\put(40,-40){\circle*{7}}
\put(60,-40){\circle*{10}}
\put(60,-40){\circle*{7}}
\put(60,-20){\circle*{7}}
\put(60,0){\circle*{10}}
\put(60,0){\circle*{7}}
\put(60,20){\circle*{7}}
\put(80,20){\circle*{7}}
\put(100,20){\circle*{7}}
\put(100,40){\circle*{7}}

\thinlines
\color{black}
\put(20,-40){\circle{7}}
\put(40,-40){\circle{7}}
\put(60,-40){\circle{10}}
\put(60,-40){\circle{7}}
\put(60,-20){\circle{7}}
\put(60,0){\circle{10}}
\put(60,0){\circle{7}}
\put(60,20){\circle{7}}
\put(80,20){\circle{7}}
\put(100,20){\circle{7}}
\put(100,40){\circle{7}}

\put(60,-40){
\put(-2.6,-2.6){\line(1,1){5.2}}
\put(-2.6,2.6){\line(1,-1){5.2}}
}
\put(60,20){
\put(-2.6,-2.6){\line(1,1){5.2}}
\put(-2.6,2.6){\line(1,-1){5.2}}
}
\put(100,20){
\put(-2.6,-2.6){\line(1,1){5.2}}
\put(-2.6,2.6){\line(1,-1){5.2}}
}
\put(100,40){
\put(-2.6,-2.6){\line(1,1){5.2}}
\put(-2.6,2.6){\line(1,-1){5.2}}
}
\put(60,0){
\put(40,10){\vector(-1,0){40}}
\put(60,10){\vector(1,0){40}}
\put(45,7){$\wb^1_R$}
}
\put(-80,0){
\put(45,-10){\vector(-1,0){45}}
\put(95,-10){\vector(1,0){45}}
\put(50,-13){$-P-\wb^{\groupp}_L$}
}
\put(48,0){
\put(0,10){\vector(0,-1){50}}
\put(0,30){\vector(0,1){30}}
\put(-3,18){$\wf_3$}
}
\put(72,-40){
\put(0,0){\vector(0,-1){40}}
\put(0,20){\vector(0,1){40}}
\put(-3,8){$P-\wf_4$}
}
\put(6,-49){$P$}
\put(26,-49){$P$}
\put(88,23){$0$}
\put(108,43){$0$}
\put(102,28){$0$}

%
\end{picture}

\caption{\label{fig:readingfund}Reading off representation weights}
\end{center}
\end{figure}
It proves useful to introduce  a  generalisation of ordinary Young diagrams which we call non-compact Young diagrams. The definition of these diagrams is explained  in Fig.~\ref{fig:noncYoung}. They bijectively correspond to the construction of the space $U_0$ and hence they label representations of $\algu(\groupp,\!|\groupm|\groupq)\oplus\algu(1)$  algebra. One advantage of such a diagram is that it can be used as an invariant that determines the representation but does not depend on the grading choice. For instance, one can take a diagram and easily read off the fundamental weight of the representation in any grading, as outlined in Fig.~\ref{fig:readingfund}

There are two transformations of the diagrams that preserve the isomorphism class of the $\su(\groupp,\!|\groupm|\groupq)$ representations, according to \eqref{betafinal}. The transformation
\begin{subequations}
\label{isomoves}
\be
\gamma_L,\gamma_R,|{\bf F}_{\!\Delta}|,P \quad\leftrightarrow \quad \gamma_L-1,\gamma_R,|{\bf F}_{\!\Delta}|,P+1
\ee
affects the diagram below the horizontal line only, while the transformation
\begin{align}
\label{isomoveabove}
\gamma_L,\gamma_R,|{\bf F}_{\!\Delta}|,P \quad\leftrightarrow \quad \gamma_L,\gamma_R-1,|{\bf F}_{\!\Delta}|+1,P+1
\end{align}
affects the diagram above the horizontal line only. These transformations generalise the one shown in the left part of Fig.~\ref{fig:GenericYoung}. They are permitted only if they relate admissible values of $\gamma_L,\gamma_R,|{\bf F}_{\!\Delta}|,P$. For instance, in the example of Fig.~\ref{fig:noncYoung} the move \eqref{isomoveabove} is not allowed as $\gamma_R$ is fixed to be zero.
\end{subequations}

There is an alternative diagrammatic description of non-compact representations, in terms of T-hook diagrams \cite{Gromov:2010vb,Tsuboi:2011iz}. Mapping of the non-compact Young diagrams to T-hook diagrams is given in Fig.~\ref{fig:Thookmap}.

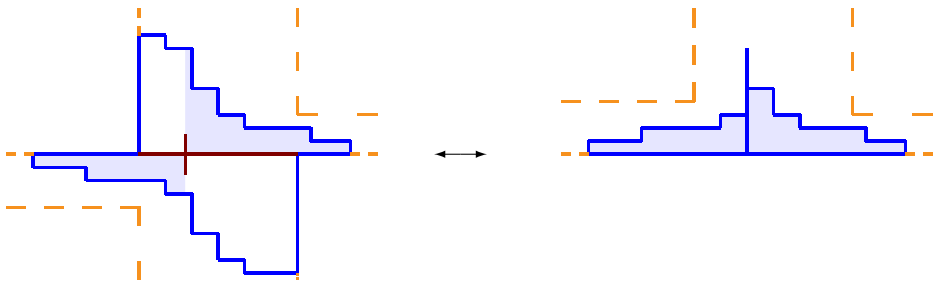
\begin{figure}[t]
\begin{center}
\setlength{\unitlength}{0.5pt}
{
\begin{picture}(680,200)(-340,-90)
\put(-250,0){
\thicklines
\color{blue!10}
\polygon*(0,0)(0,90)(20,90)(20,80)(40,80)(40,50)(60,50)(60,30)(80,30)(80,20)(130,20)(130,10)(160,10)(160,0)(120,0)(120,-90)(80,-90)(80,-80)(60,-80)(60,-60)(40,-60)(40,-30)(20,-30)(20,-20)(-40,-20)(-40,-10)(-80,-10)(-80,0)(0,0)
\color{white}
\polygon*(0,0)(0,90)(20,90)(20,80)(35,80)(35,0)
\polygon*(120,0)(120,-90)(80,-90)(80,-80)(60,-80)(60,-60)(40,-60)(40,-30)(35,-30)(35,0)
\color{blue}
\drawline(0,0)(0,90)(20,90)(20,80)(40,80)(40,50)(60,50)(60,30)(80,30)(80,20)(130,20)(130,10)(160,10)(160,0)(120,0)(120,-90)(80,-90)(80,-80)(60,-80)(60,-60)(40,-60)(40,-30)(20,-30)(20,-20)(-40,-20)(-40,-10)(-80,-10)(-80,0)(0,0)

\color{BurntOrange}
\dashline{16}(160,0)(180,00)
\dashline{16}(120,30)(180,30)
\dashline{16}(-80,0)(-100,0)
\dashline{16}(00,-40)(-100,-40)
\dashline{16}(0,90)(0,110)
\dashline{16}(120,30)(120,110)
\dashline{16}(0,-40)(0,-95)
\dashline{16}(120,-90)(120,-95)

\thicklines
\color{red!50!black}
\drawline(0,0)(120,0)
\drawline(35,15)(35,-15)
}
\put(170,0){
\thicklines
\color{blue!10}
\polygon*(40,80)(40,50)(60,50)(60,30)(80,30)(80,20)(130,20)(130,10)(160,10)(160,0)(120,0)(40,0)
\polygon*(40,30)(20,30)(20,20)(-40,20)(-40,10)(-80,10)(-80,0)(40,0)
\color{blue}
\drawline(40,80)(40,50)(60,50)(60,30)(80,30)(80,20)(130,20)(130,10)(160,10)(160,0)(120,0)(40,0)(40,80)
\drawline(40,30)(20,30)(20,20)(-40,20)(-40,10)(-80,10)(-80,0)(40,0)

\color{BurntOrange}
\dashline{16}(160,0)(180,00)
\dashline{16}(120,30)(180,30)
\dashline{16}(-80,0)(-100,0)
\dashline{16}(00,40)(-100,40)
\dashline{16}(120,30)(120,110)
\dashline{16}(0,40)(0,110)

\thicklines
\color{red!50!black}
}
\put(0,0){\vector(1,0){20}}
\put(0,0){\vector(-1,0){20}}
\end{picture}
}

\caption{\label{fig:Thookmap} Bijection to T-hook diagram. If $P$ is given, the shaded in blue domain fully determines the non-compact Young diagram. The diagram on T-hook is obtained by flipping the left shaded domain. Position of the splitting vertical red line can be chosen arbitrarily. In \cite{Gromov:2010vb}, the diagrams on T-hook appeared for the first time, with the splitting being at the right-most possible place. Possibility of a T-hook diagram with arbitrary splitting was proposed in  \cite{Tsuboi:2011iz}.}
\end{center}
\end{figure}

Finally we note that there is a possibility to extend the non-compact Young diagrams\footnote{This is proposed in the paper \cite{Marboe:2017dmb} by  C.~Marboe and one of the authors. This idea builds on the proposal of non-compact Young diagrams of the current manuscript which was shared for the needs of work \cite{Marboe:2017dmb}. For historical reasons, the diagrams in \cite{Marboe:2017dmb} are reflections in $45^\circ$  mirror of the ones we use here. Our choice was dictated by desire to get conventional ordinary Young diagrams in $\groupp=\groupm=0$ case. The choice of \cite{Marboe:2017dmb} was dictated eventually by the choice of the T-hook orientation in  \cite{Gromov:2009tv}.}. The extended diagram is a union of shaded blue and gray domains of Fig.~\ref{fig:noncYoung}, it is defined by the property that the vertical displacement between its upper and lower boundaries is always $P$. Extended Young diagram naturally emerges if we extend the  representation of $\algu(\groupp,\groupq|\groupm)$ to the representation of $\algu(\infty,\infty|\infty)$ by padding original Dynkin weight $\langle \omega\rangle\equiv \langle \omega_1,\omega_2,\ldots,\omega_{\groupp+\groupq+\groupm-1} \rangle $ with zeros from left and right getting $\langle \ldots 0,0,\omega,0,0,\ldots\rangle$. This should be done in a grading for which the Kac-Dynkin diagram starts, as a two dimensional path of Fig.~\ref{fig:readingfund}, in the far south-west, passes through the original non-compact Young diagram once, and heads towards far north-east.

The extended diagram is not attached to a particular $\su(\groupp,\groupq|\groupm)$ algebra anymore. In fact, we can use it to define representations of different $\su(\groupp,\groupq|\groupm)$ algebras, by choosing different $\groupp,\groupq,\groupm$ and hence by carving out different non-compact diagrams from the same extended diagram, according to Fig.~\ref{fig:noncYoung}. Such an identification predicts a non-trivial map between representations of different $\su(\groupp,\groupq|\groupm)$ algebras, in particular between their characters, and it  was already used for counting purposes in \cite{Marboe:2017dmb}. This map is likely to hide further valuable combinatorial relations which are yet to be explored.  Also extended Young diagrams naturally  emerge as a practical tool for solving rational  Bethe equations \cite{Marboe:2017dmb}.

\section{Conclusions and discussion}
In this paper we found  all the unitary representations of $\su(\groupp,\groupq|\groupm)$ *-algebras with $\groupn=\groupp+\groupq\neq 0$ and $\groupm\neq 0$. Given that these are the only *-algebras of $\sl(\groupn|\groupm)$ that admit nontrivial unitary representations, see section~\ref{sec:otherreal}, we automatically get the  unitary duals for all possible real forms of $\sl(\groupn|\groupm)$. To our knowledge, this is the first time this result was published in the literature in full generality. However, many important partial cases were covered before and hence we believe that the main contribution of this paper is in the development of certain techniques towards the classification and combinatorial goals. This comprises usage of duality transformations to establish very effective necessary unitarity constraints, introduction of deformed Fock modules to produce representations with non-integer weights, and introduction of non-compact and extended Young diagrams that are a natural generalisation of  ordinary Young diagrams for description of unitary representations of non-compact real forms of $\gl$-type Lie superalgebras.

Let us discuss how the proposed techniques compare to the approaches used in the literature. The typical widely-used departing point, which we also employed, is to describe a representation $U$ as an induced representation from $U_0$. One considers decomposition $\algg=\CE^{(-)}\oplus\CK\oplus \CE^{(+)}$, cf. \eqref{eq:EKE}, departs from $U_0$, which should be a unitary $\CK$-module such that $\CE^{(+)}U_0=0$, and then generates $U$ by acting on $U_0$ with elements from $\CE^{(-)}$. The question to answer is what further constraints one should impose on $U_0$ to get a positive-definite norm on $U$.

{\it Necessary conditions and duality transformations.} The first step is to select particular "interesting" vectors or subspaces in $\CE^{(-)}U_0$ and check the positivity of the Hermitian structure restricted to them, thus producing the necessary conditions of unitarity. This was for instance a way to decide about the first possible places (of non-unitarity) and last possible places (of unitarity) in \cite{MR733809}. The difficult part here is to find such vectors  whose norm is easily computable and whose study gives strong enough unitarity constraints. If one were not to use supersymmetry,  one would have to  dig far inside of $\CE^{(-)}U_0$ to decide about the last possible place, where analysis of the positivity of the norm becomes more involved. For instance, one could use Kac determinants \cite{MR1214730}, but other arguments can be used as well\footnote{This is one of the places where the derivation of theorem 5.2. in \cite{MR1214730} was incorrect. A certain vector in $\CE^{(-)}U_0$ there was wrongly assumed to be a highest-weight with respect to $\su(\groupp,\groupq)$ subalgebra action, i.e. being annihilated by certain generators of $\CE^{(+)}$. It was not the case because particular elements of $\CE^{(-)}$ and $\CE^{(+)}$ that were used did not actually commute.}. In our paper we propose to consider only vectors from the critical subspace -- subspace of all possible vectors which can become HWS's with an appropriate choice of grading. One can travel between such vectors using duality transformations, and the signature of the Hermitian form on such subspace is easily computable. Duality transformations can be considered as odd Weyl "reflections". We observe that by using these "reflections" we get powerful constraints on unitarity of the representation and can understand in great detail the structure of shortening conditions. The output is similar in spirit to usage of the (ordinary) Weyl group to decide about reducibility of a Verma module \cite{BGG}. To put it differently,  analysis of reducibility in the absence of supersymmetry often leads to the study of action of particular rank-1 $\su(2)$ subalgebras, whereas in supersymmetric case we can use $\su(1|1)$ subalgebra instead whose unitary representations are simpler. We can benefit from these simplifications even in non-supersymmetric models, by artificially extending the symmetry algebra with odd-graded generators, as we did in appendix~\ref{sec:appa}, see also usage of similar tricks in \cite{Volin:2010xz,Marboe:2016yyn}.

Duality transformations naturally lead to the need to properly deal with various different orders among Lie algebra generators, which we did by introduction of $p$- and $c$-gradings. Having this flexibility considerably enriches our intuition about the structure of Lie superalgebras and  it also simplifies description of possible unitary representations. For instance, cases $N_{st}^+$ and  $N_{nst}^+$ in \cite{MR1214730} correspond  simply to different choices of gradings, $\su(\groupp,\groupq|\groupm)$ and $\su(\groupp,\!|\groupm|\groupq)$ respectively. We now know that one needs not study both of them, but one case will follow from the other.

It is also important to carefully write down relations between gradings, *-structure, and real forms of the Lie algebra, as there are subtleties on how one treats odd generators in the real form, see section~\ref{sec:defunitary}. Being inaccurate here can lead to certain confusions. For instance certain representations of $\su(\groupp,\!|\groupr+\groups|\groupq)$ *-algebra were  considered in \cite{ChenZhang} as novel unitary representations of $\su(\groupp,\groupq|\groupr,\groups)$ real form. As we saw, $\su(\groupp,\groupq|\groupr,\groups)$ has no non-trivial unitary representations if $\groupp,\groupq,\groupr,\groups\neq 0$. In fact, $\su(\groupp,\!|\groupr+\groups|\groupq)$ *-algebra corresponds to $\su(\groupp,\groupq|\groupr+\groups)$ real form. The representations  in \cite{ChenZhang} are, by inspection, not novel but the already known ones in an alternative (and useful) description stemming from a different choice of a fermionic Fock vacuum.

{\it Sufficient condition through oscillator construction in deformed Fock space.} To actually prove unitarity, one needs to prove positive-definiteness of the Hermitian form on the entire representation space. There are various ways to approach this question, and oscillator approach is widely used as was reviewed in the introduction. The challenge is to extend the known results to non-integer weighted representations. In this paper we achieved this goal.

A classical example of how non-integer weights can appear when using oscillators is  the Holstein-Primakof-Dyson  realisation of the Lie algebra
$\sl(2)$ \cite{Holstein:1940zp,Dyson:1956zza}:
\be\label{HPrealisation}
 \mathsf{H}=\frac{\gamma}2-a^{\dagger}\,a\,\quad \gE=a^\dagger(\gamma-a^{\dagger}a)\,,  \quad \mathsf{F}=a\,.
\ee
If one realises $a,a^{\dagger}$ as operators acting on the standard Fock space then the construction describes highest-weight spin-$\frac{\gamma}{2}$ representation of $\sl(2)$ with arbitrary and not necessarily integer spin.
The Holstein-Primakoff-Dyson realisation  of $\sl(2)$ and their various generalisations in the physics literature correspond to special cases of a general realisation of Lie algebras, using the Tits-Kantor-Koecher construction \cite{MR0146231,MR0228554,MR0321986},  that was given in \cite{Gunaydin:1992zh} and extended to Lie superalgebras therein. The resulting Lie algebras and Lie superalgebras can be regarded as generalised conformal or superconformal algebras over the underlying Jordan (super-)algebras (or Jordan triple systems) and admit a one-parameter family of deformations $\gamma$ (unitary character) \cite{Gunaydin:2009dq,Gunaydin:2009zz}.  The concept of a generalised spacetime coordinatized by Jordan algebras was introduced in \cite{Gunaydin:1975mp}.  The Jordan algebra $J_2^{\mathbb{C}}$ of $2\times 2$ Hermitian matrices describes the four-dimensional Minkowski spacetime. The automorphism, reduced structure, and linear fractional group of $J_2^{\mathbb{C}}$
are, respectively, the space rotation group $\SU(2)$, Lorentz group $\SL(2,\mathbb{C})$, and the conformal group $\SU(2,2)$. Similarly one identifies the automorphism, reduced structure, and linear fractional group of a Jordan algebra $J$  with the rotation, Lorentz, and conformal group of the generalised spacetime defined by $J$ \cite{Gunaydin:1975mp,Gunaydin:1999jb}.
The conformal Lie algebras defined by Hermitian Jordan triple systems via the TKK construction exhaust the list of simple Lie algebras that admit unitary highest (or lowest) weight representations \cite{Gunaydin:1992zh}. Some of these Hermitian triple systems are defined by Euclidean Jordan algebras under the Jordan triple product. It was shown in \cite{Mack:2004pv} that the spacetimes defined by Euclidean Jordan algebras are causal. The Jordan algebra $J_2^{\mathbb{C}}$ that defines the four-dimensional Minkowski spacetime belongs to an infinite family of Euclidean Jordan algebras, namely the Jordan algebra $J_{\groupn}^{\mathbb{C}}$ of $\groupn\times \groupn$ complex Hermitian matrices. Their rotation, Lorentz and conformal groups are $\SU(\groupn),\SL(\groupn,\mathbb{C})$ and $\SU(\groupn,\groupn)$. Therefore the results of  this paper can be applied directly to obtain the positive energy unitary representations of the conformal groups $\SU(\groupn,\groupn)$ of the spacetimes defined by $J_{\groupn}^{\mathbb{C}}$ and the entire unitary duals of their supersymmetric extensions, namely $\SU(\groupn,\groupn|\groupm)$.

With the exception of the exceptional Lie algebras $\mathfrak{g}_2, \mathfrak{f}_4$ and $\mathfrak{e}_8$ all simple Lie algebras admit a TKK  construction and hence a generalised conformal realisation.

Conformal realisations of Lie algebras were extended to geometric nonlinear quasiconformal realisations in \cite{Gunaydin:2000xr}. All simple Lie algebras can be realized as  quasiconformal Lie algebras including $\mathfrak{g}_2, \mathfrak{f}_4$ and $\mathfrak{e}_8$ and they  admit deformations by a unitary character \cite{Gunaydin:2007qq,Gunaydin:2009dq,Gunaydin:2009zz}. Quantisation of the geometric quasiconformal actions lead to the minimal unitary representations of the corresponding groups, which can be expressed in terms of ordinary oscillators and a singular oscillator of the Calogero type \cite{Gunaydin:2001bt,Gunaydin:2006vz}. These results extend also to Lie superalgebras \cite{Gunaydin:2006vz,Fernando:2009fq,Fernando:2010ia}.

We stress here that in representations obtained via the conformal or quasiconformal approach  one deals with generators that are, in general, non-quadratic polynomials in oscillators, also the non-integer weight deformation parameter $\gamma$ enters directly in the realisation of the Lie algebra or Lie superalgebra. Hence the *-structure is not transparent in such realisations, for instance $\gE^*\neq \pm \mathsf{F}$ if we consider \eqref{HPrealisation} for standard Hermitian conjugation properties of oscillators\footnote{One can make it Hermitian at the price of abandoning polynomiality: $E=a^{\dagger}\sqrt{\gamma-a^{\dagger}a}$, $F=a\sqrt{\gamma-a^{\dagger}a}$.}. For these reasons, we found the Holstein-Primakof-Dyson  realisation and its conformal generalisations not very convenient for the study of the unitary dual of Lie superalgebras. Instead we propose a different oscillator-based approach, where it is the representation module of the oscillator algebra that is deformed, whereas the basic bilinear realisation of the Lie algebra with transparent conjugation properties is preserved.  For constructing the unitary representations with anomalous dimensions, one needs simply a sufficient number of colours and to choose the value of $\gamma$ within the unitarity range  to construct the corresponding deformed Fock space $\mathcal{F}_\gamma$. Furthermore, tensoring of the resulting representations and their decompositions are very straightforward.  We expect that the proposed approach can be  applied without much modification to construction of highest-weight  representations  and for description of the complete unitary duals in other Lie (super-)algebras where oscillator realisations are known.

{\it Labelling using non-compact and extended Young diagrams.} Finally, one should label somehow the classified unitary representations. Value of the highest weight in a particular grading is a possibility, but it is not an invariant quantity. One can note that, for representations without shortenings, the change of the weight  follows a simple pattern under duality transformations, cf. \eqref{eq:dualityrule2}, which naturally suggests usage of a shifted weight. The shifted weight is formally defined as the weight shifted by the Weyl vector  and its value does not depend on the grading\footnote{In practical terms, it is equal to the distance from the border of the Young diagram to its diagonal.}. Supersymmetric polynomials in shifted weights explicitly realise the Harish-Chandra isomorphism on the level of representations \cite{OkO,Molev1997FactorialSS}.

For the case of short representations the shifted weights are no longer appropriate objects. Therefore we propose to use non-compact Young diagrams that are invariants of a representation and that can be used to describe short as well as long multiplets. Furthermore, non-compact Young diagrams, for the case of integer weights, can be extended \cite{Marboe:2017dmb}, and the extended Young diagrams are not sensitive to the rank of the algebra  playing a very similar role to the role of the ordinary Young diagrams in the study of compact representations. Note that ordinary Young diagrams with $P$ boxes naturally emerge from the Schur-Weyl duality between the $\gl$ Lie algebra and the permutation group $S_P$. In extended Young diagrams, $P$ is the height of a diagram, and it would be an interesting venue to explore how the Schur-Weyl duality is encoded on these diagrams for the purpose of studying representations of non-compact Lie algebras and superalgebras.\\[0em]

{\bf Acknowledgments:} We thank Matthias Staudacher, Carlo Meneghelli, and especially Christian Marboe for useful discussions. The research of M.G. was supported in part by  the National Science Foundation
under Grant Numbers PHY-1213183 and PHY-0855356 and
 the US Department
of Energy under DOE Grant No: DE-SC0010534. The research of D.V. was partially supported by the People Programme (Marie Curie Actions) of the European Union's Seventh Framework Programme FP7/2007-2013/ under REA Grant Agreement No 317089 (GATIS). M.G. thanks Nordita for hospitality where part of the research was done. D.V. thanks IHES and IPhT, C.E.A-Saclay for hospitality where part of the research was done.

\appendix
\section{Technical details}
\subsection{\label{sec:appa}Necessary conditions for UHW of $\su(\groupp,\groupq)$ via a supersymmetric extension trick}
Consider an irrep of $\su(\groupp,\groupq)$ with the fundamental weight
\be
[-\mu_{L},\mu_{R}]=[-\mu_{L}^{(0)},\mu_R^{(0)}]+[\mdash  (-\beta_L)  \mdash,\mdash , \beta_R , \mdash]
\ee
The sufficient condition of unitarity is given in section~\ref{sec:supq}.  We prove in this appendix that they are also necessary.

Given that unitarity depends only on the sum $\beta_L+\beta_R$, we set $\beta_L=h_{\mu_L}$ and focus on discussion of $\beta_R$ only. Let us assume that $\beta_R$ does  not satisfy the sufficient condition requirement. Then we can always find a small number $\epsilon$ such that $\beta_R+\epsilon$ is non-integer and that it also does not satisfy this requirement. Consider  $\su(\groupp,\groupq)$ as a subalgebra of $\su(\groupp,\!|\groupn|\groupq)$ where we chose $\groupn$ to be sufficiently large so that $\groupn+\beta_R$ satisfy the unitarity constraint for unitarity of $\su(\groupp,\groupq)$ representation.

Consider a representation with fundamental weight
\be
[-\mu_{L}^{(0)};\mdash\epsilon\mdash;\mu_R^{(0)}]+[\mdash , (-\beta_L) , \mdash;\mdash , 0 , \mdash;\mdash , \beta_R , \mdash]
\ee
in the $\su(\groupp,\!|\groupn|\groupq)$ grading. Denote the corresponding highest-weight vector as $|\Omega\rangle$.

The same representation in $\su(\groupp,\groupq|\groupn)$ grading has the weight
\be
[-\mu_{L}^{(0)};\mu_R^{(0)};\ldots]+[\mdash , (-\beta_L) , \mdash;\mdash , \beta_R+\groupn , \mdash;\mdash , 0 , \mdash]\,,
\ee
where $\ldots$ denotes the eigenvalues of $E_{aa}$ generators which are not important for this proof. Denote the corresponding highest-weight vector for this grading as $|\Omega'\rangle$. The two highest-weight vectors are related by
\be
|\Omega'\rangle=\prod_{a,\alpha}E_{\alpha\,a}|\Omega\rangle\,,
\ee
where, we recall, $a$ is a $p$-odd $c$-even index and $\alpha$ is a $p$-even $c$-even index. In particular, $E_{\alpha\,a}$ is a $p$-odd generator. The product runs over all the indices of this parity.

We will also need the following simple fact:
\be
E_{\dot \groupp\groupq}^k|\Omega'\rangle=\prod_{a,\alpha}E_{\alpha\,a}(E_{\dot \groupp\groupq}^k|\Omega\rangle)\,.
\ee
Given that the representation is that of a non-compact real form we have $E_{\dot \groupp\groupq}^k|\Omega'\rangle\neq 0$  for any $k\in\mathbb{Z}_{\geq 0}$.

Action of $\su(\groupp,\groupq)$ generators on $|\Omega'\rangle$ generates a UHW because $\beta_{R}+\groupn$ satisfies the sufficient unitary constraint. Hence all the vectors $E_{\dot \groupp\groupq}^k|\Omega'\rangle$ have the norm of the same sign. Using the analysis introduced in section~\ref{sec:plaquette}, we deduce the signs of $E_{\dot \groupp\groupq}^k|\Omega\rangle$. For sufficiently large $k$ it should be the same sign as $E_{\dot \groupp\groupq}^k|\Omega'\rangle$, but at certain finite $k$ the sign would change because $\beta_R$ does not satisfy the sufficient unitarity constraints, change of sign happens at least once, but it can happen several times as well, depending on the actual value of $\beta_R$. Hence, in the module generated by action of $\su(\groupp,\groupq)$ on $|\Omega\rangle$ there are states with norm of different sign. Hence this module is not unitary. \qed

\subsection{Structure of monomial shortenings}
\label{sec:monoshortproof}
Here we prove the statements formulated after \eqref{shorteningweightvalues}. We use the same conventions as in section~\ref{sec:shortening}.

First consider the case when all $a_i$ are equal and hence $\prod\limits_{i=1}^{r}\gE_{\mu_ia}|\es\rangle=0$ such that none of the subproducts of $\prod\limits_{i=1}^{r}\gE_{\mu_ia}$ annihilates $|\es\rangle$. Without loss of generality we can assume $\mu_1<\mu_2<\ldots<\mu_r$. Restrict to the subalgebra $\su(1|r)$ which includes all the generators $\gE_{\mu_ia}$. Then $|\Omega\rangle=\prod\limits_{i=1}^{r-1}\gE_{\mu_ia}|\es\rangle$ is the HWS for $\su(0|r-1|1|1)$ grading. Condition $\gE_{\mu_ra}|\Omega\rangle$ implies $\wb_{\mu_r}+\wf_a=0$ in this grading. We can choose the overall  shift $\Lambda$ such that $\wb_{\mu_r}=0$. Then $\wf_a=0$ in this grading, hence $\wf_a=r-1$ in the distinguished $\su(1|r)$ grading, and for the distinguished grading of the full $\su(\groupq|\groupm)$ algebra as well. Then unitarity constraints together with $\wb_{\mu_r}=0$ imply that $\wb_i=0$ for all $i\geq r$, in particular the shift $\Lambda$ we have chosen here is the same as was defined at the beginning of section~\ref{sec:shortening}. We therefore proved the r.h.s. of \eqref{shorteningweightvalues}.

Now consider a generic monomial $\mathcal{M}=\prod\limits_{i=1}^{r}\gE_{\mu_ia_i}$ that annihiliates $|\es\rangle$ (with no subproducts of $\mathcal{M}$ doing the same). Represent it in the form where  the terms with the same value of $a_i$ are collected: $\prod\limits_{i=1}^{k}\prod\limits_{j=1}^{r_i}{\gE_{\mu_{j,i}a_i}}\,,$ such that $r=\sum\limits_{i=1}^k r_i$. Restrict onself from $\su(\groupq|\groupm)$ to the algebra $\su(k|\groupm)$ which includes $\gE_{\mu a_i}$, so effectively we can label $a_i=i$ in this subalgebra. Using the fact that $|\es\rangle$ is a HWS, act on $\prod\limits_{i=1}^{k}\prod\limits_{j=1}^{r_i}{\gE_{\mu_{j,i}i}}|\es\rangle$ with $\gE_{\mu'\mu_{j,i}}$, $\mu'<\mu_{j,i}$ sufficiently many times  to conclude $\prod\limits_{a=1}^{k}\prod\limits_{\mu=1}^{r_a}{\gE_{\mu a}}|\es\rangle=0$. We therefore got a more regular monomial to analyse, however we should be careful as it is not guaranteed now that there are no subproducts annihilating $|\es\rangle$.

If  $\prod\limits_{\mu=1}^{r_k}\gE_{\mu k}|\es\rangle=0$ then we are in the situation  considered above and conclude that $\lambda_k\leq r_k-1$ (It can be a strict inequality if some subproducts of $\prod\limits_{\mu=1}^{r_k}\gE_{\mu k}$ alreday annihiliate $|\es\rangle$). Then, coming back to the full $\su(\groupq|\groupm)$ algebra, we conclude that $\prod\limits_{j=1}^{r_k}\gE_{\mu_{j,k}a_k}|\es\rangle=0$ which contradicts original assumption about subproducts of $\mathcal{M}$. Therefore $\prod\limits_{\mu=1}^{r_k}\gE_{\mu k}|\es\rangle\neq 0$.

Now consider $|\es'\rangle$ -- the HWS with respect to the $\su(k-1|\groupm|1)$ grading. It is related to $|\es\rangle$ in particular by action of $\prod\limits_{\mu=1}^{r_k}\gE_{\mu k}$, hence $\prod\limits_{a=1}^{k-1}\prod\limits_{\mu=1}^{r_a}{\gE_{\mu a}}|\es'\rangle=0$. We  repeat the same arguments to prove that $\prod\limits_{\mu=1}^{r_{k-1}}{\gE_{\mu ,k-1}}|\es'\rangle\neq 0$ etc, and by recursion  come to the conclusion that $\prod\limits_{\mu=1}^{r_{1}}{\gE_{\mu1}}|\Omega\rangle\neq 0$, where $|\Omega\rangle$ is the HWS in $\su(1|\groupm|k-1)$ grading. But again, this annihilation leads to conclusion that the subproduct $\prod\limits_{j=1}^{r_1}{\gE_{\mu_{j,1}a_1}}$ of $\mathcal{M}$ annihilates $|\es\rangle$. This is not contradictory only in the case if  $\prod\limits_{j=1}^{r_1}{\gE_{\mu_{j,1}a_1}}=\mathcal{M}$ meaning that all $a_i$ are equal, as required. \qed


\section{\label{sec:confa}Conformal algebra in four dimensions and its positive energy unitary representations}
The non-compact group $\SU(2,2)$ appears in physics in various applications. The most known ones originate from the fact that $\SU(2,2)$ is the double covering group of  $\SO(4,2)$, while the latter can be interpreted  as the group of conformal transformations in four-dimensional Minkowskian space-time, or as the isometry group of the AdS$_5$ space-time. It is also relevant for the description of Euclidean conformal field theory in four dimensions as  will be explained below. Here we review unitary representations of its Lie algebra $\su(2,2)\simeq \so(4,2)$ using the approach developed in the main text.

\subsection{Compact and non-compact bases  of $\su(2,2)$}
$\su(2,2)$, as a real Lie algebra in its Hermitian basis, can be viewed as an algebra of $4\times 4$ matrices $X$ that satisfy conjugation properties
\be
X^\dagger\, H=H\, X\,,
\ee
where $H$ is a matrix with signature $(+,+,-,-)$. Two useful choices for $H$ are  $H_{\rm c}=\left(\begin{smallmatrix}{\mathbbm{1}_2}&0\\0&{-\mathbbm{1}_2}\end{smallmatrix}\right)$ and $H_{\rm nc}=\left(\begin{smallmatrix}{0}&{\mathbbm{1}_2}\\{\mathbbm{1}_2}&{0}\end{smallmatrix}\right)$. They lead to the following structure of matrices $X$:
\be
X_{\rm c}=\mtwo{A}{C}{-C^\dagger}{B}\,,\quad X_{\rm nc}=\mtwo{C}{A}{B}{C^\dagger}\,,
\ee
where $A,B$ are Hermitian $2\times 2$ matrices and $C$ is an arbitrary complex $2\times 2$ matrix.

We will think about $\su(2,2)$ as a complex *-algebra. For $\sl(4)$ generators $\gE_{ij}$   satisfying \eqref{comrel1}, the choice of $H=H_{\rm c}$ implies the conjugation properties:
\be\label{compactstar}
\gE_{\dot\alpha\dot\beta}^*=\gE_{\dot\beta\dot\alpha}\,,\quad \gE_{\alpha\beta}^*=\gE_{\beta\alpha}\,,\quad \gE_{\dot\alpha\beta}^*=-\gE_{\beta\dot\alpha}\,,
\ee
i.e. it is precisely the *-structure for $\su(\groupp,\groupq)$ algebras used in the main text, cf. \eqref{supqrealisation}.

The choice $H=H_{\rm nc}$ implies different conjugation properties. We will label corresponding $\sl(4)$ generators as ${\cE}_{ij}$. They also satisfy \eqref{comrel1}, but form the following *-algebra:
\be\label{noncompactstar}
\cE_{\dot\alpha\dot\beta}^*=\cE_{\vphantom{\dot\beta}\beta\alpha}\,,\quad \cE_{\vphantom{\dot\beta}\dot\alpha\beta}^*=\cE_{\dot\beta\alpha}\,,\quad \cE_{\alpha\dot\beta}^*=\cE_{\vphantom{\dot\beta}\beta\dot\alpha}\,.
\ee

The conformal algebra $\su(2,2)$ admits a
three-graded decomposition
\be
\su(2,2)=\CE^{(-)}\oplus\CK\oplus \CE^{(+)}\,,
\ee
where we  define the subalgebras $\CE^{(-)},\CK, \CE^{(+)}$  through packing their generators into $4\times 4$ matrices:  $\left(\begin{smallmatrix}{\CK}&{\CE^{(+)}}\\{\CE^{(-)}}&{\CK}\end{smallmatrix}\right)$. In the case of *-structure \eqref{compactstar}, the subalgebra $\CK\equiv \CK_{\rm c}$ is $\so(4)\oplus \algu(1)\simeq \su(2)\oplus\su(2)\oplus \algu(1)$ which is the maximal compact subalgebra of $\su(2,2)$. For this reason the choice $H=H_{\rm c}$ leads to what is called the compact basis. In the case of *-structure \eqref{noncompactstar}, $\CK\equiv \CK_{\rm nc}\simeq\sl(2,\mathcal{C})\oplus \so(1,1)\simeq \so(3,1)\oplus \so(1,1)$, this is the non-compact basis.

The isomorphism between two bases is established by the linear transformation
\be\label{comptononcomp}
\cE_{ij}=T\,\gE_{ij}\,T^{-1}\,,\quad{\rm where}\quad T=e^{\frac \pi 4\sum\limits_{\alpha=1}^2(\gE_{\alpha\dot\alpha}-\gE_{\dot\alpha\alpha})}\,,
\ee
as is easy to deduce from $H_{\rm nc}=e^{\frac{\pi}{4}I}\,H_{\rm c}\,e^{-\frac{\pi}{4}I}$, $I=\left(\begin{smallmatrix}{0}&{-\mathbbm{1}_2}\\{\mathbbm{1}_2}&{0}\end{smallmatrix}\right)$. Note that $T$ is not unitary but satisfies $T^*=T$ for either choice of the *-structure.

\subsection{Isomorphism $\su(2,2)\simeq \so(2,4)$ and relation to physics}
When $\so(2,4)$ algebra is considered as the algebra of infinitesimal conformal  transformations of 1+3 {Minkowski} space-time, the  commutation relations of its generators  are typically written in a Lorentz-covariant manner:
\begin{eqnarray}\label{comLor}
[M_{\mu\wb},M_{\rho\sigma}] & = & - i (\eta_{\wb\rho}M_{\mu\sigma}-
\eta_{\mu\rho}M_{\wb\sigma} -\eta_{\wb\sigma}M_{\mu\rho}+
\eta_{\mu\sigma}M_{\wb\rho})
\,,\cr
[ P_{\mu}, M_{\rho\sigma} ] & = & - i (\eta_{\mu\rho}P_{\sigma}-\eta_{\mu\sigma}
P_{\rho})
\,,\cr
[K_{\mu},M_{\rho\sigma}]& = & - i (\eta_{\mu\rho}K_{\sigma}-\eta_{\mu\sigma}
K_{\rho})
\,,\cr
[D,M_{\mu\wb}]& = & [P_{\mu},P_{\wb}] = [K_{\mu},K_{\wb}]=0
\,,\cr
[D,P_{\mu}] & = & iP_{\mu}; \quad [D,K_{\mu}]=-iK_{\mu}
\,,\cr
[P_{\mu},K_{\wb}]& = &2i(\eta_{\mu\wb}D-M_{\mu\wb})
\,,
\end{eqnarray}
where $M_{\mu\wb}$ are the  Lorentz algebra
generators,  $P_{\mu}$ the translation generators, $D$ the dilatation
generator and
  $K_{\mu}$ are  the special conformal generators.  We use a mostly
plus  metric,  $\eta_{\mu\wb}=\textrm{diag}(-,+,+,+)$ with ($\mu, \wb,
\dots = 0, 1, 2, 3$).

The above-listed commutation relations  are written in a Hermitian basis, i.e. all generators are {self-conjugate}:
\be\label{selfconjugate}
D^* = D\,,\ \ \ P_{\mu}^*= P_{\mu}\,,\ \ \ K_{\mu}^* = K_{\mu}\,,\ \ \ M_{\mu\wb}^* = M_{\mu\wb}\,.
\ee

This Lie algebra can be written also in a manifestly $\SO(4,2)$-covariant form
($\eta_{44}=-\eta_{55}=1$; $\quad a,b,\dots=0,1,2,3,4,5$)
\be
[M_{AB}, M_{CD}] = - i(\eta_{BC}M_{AD} - \eta_{AC}M_{BD}
-\eta_{BD}M_{AC} + \eta_{AD}M_{BC})\label{SO42app}
\ee
 by identifying
\be
M_{\mu 4} = -{1 \over 2} (P_{\mu} - K_{\mu}), \quad
M_{\mu 5} =-{1 \over 2} (P_{\mu} + K_{\mu}), \quad M_{45} = D\,,
\ee
and in this form it is directly related to the isometries of the AdS$_5$ space-time.

\paragraph{Positive-energy representations.}
The maximal compact subalgebra $\so(4)\oplus \mathfrak{o}(2) \simeq \su(2)_{L}\oplus \su(2)_{R}\oplus
 \algu(1)_{E}$
is spanned by
\begin{align}
&& L_{m}&=\frac{1}{2} \left( \frac{1}{2} \varepsilon_{mnl} M_{nl}-M_{m4}\right)\,
&\leftrightarrow  && \su(2)_{L}\,,&&\cr
&& R_{m}&=\frac{1}{2} \left( \frac{1}{2} \varepsilon_{mnl} M_{nl}+M_{m4}\right)\,
&\leftrightarrow  &&  \su(2)_{R}\,,&&\cr
&& H&=\frac{1}{2}(P_{0}+K_{0})&\leftrightarrow && \algu(1)_E\,,&&
\label{transformation}
\end{align}
where $m,n=1,2,3$, and the non-zero commutators are  $[L_{m},L_{n}]=i \varepsilon_{mnp}L_{p}$, $[R_{m},R_{n}]=i \varepsilon_{mnp}R_{p}$.
The $\algu(1)_{E}$ generator $H$ plays the role of conformal Hamiltonian or the AdS energy. \\[0 em]

The unitary positive-energy representations is a natural question to study for the quantum theories on the AdS space-time, where the conformal algebra plays the role equivalent to the Poincare algebra of the Minkowski spacetime. Positive-energy representations are necessarily of the highest-weight type \cite{Mack:1975je}, hence one can label them by the Cartan weights of the highest-weight vector. A natural  selection is to employ two spins $(j_L,j_R)$  of $\su(2)_L\oplus\su(2)_R$ and the AdS energy $E_0$ (eigenvalue of $H$ on the HWS) as these labels\cite{Gunaydin:1984fk} and we  will use them in the next section.

Isomorphism between $\so(2,4)$ and $\su(2,2)$ is established using the standard twistor relations. Define
\be
(\sigma^{\mu})_{\alpha\dot\beta}=(\mathbbm{1}_2,\overrightarrow{\sigma})
\quad{\rm and }\quad
(\bar\sigma^{\mu})^{\dot\alpha\beta}=(\mathbbm{1}_2,-\overrightarrow{\sigma})\,,
\ee
and then
\be
(\sigma^{\mu\nu})_{\alpha}{}^{\beta}=\frac 14(\sigma^{\mu}\bar\sigma^{\nu}-\sigma^{\nu}\bar\sigma^{\mu})\,,
\quad
(\bar\sigma^{\mu\nu})^{\dot\alpha}{}_{\dot\beta}=\frac 14(\bar\sigma^{\mu}\sigma^{\nu}-\bar\sigma^{\nu}\sigma^{\mu})\,.
\ee
With these conventions $P,K,M,D$ are expressed through $\gl(4)$ generators as follows
\begin{subequations}
\label{PKMD2}
\begin{gather}
P^\mu= -(\bar\sigma^{\mu})^{\dot\beta\alpha}\cE_{\alpha\dot\beta}\,,\quad K^{\mu} =-\sigma^{\mu}_{\beta\dot\alpha}\cE^{\dot\alpha\beta}\,,
\\
M^{\mu\nu}=-\ii\left[(\sigma^{\mu\nu})_{\beta}{}^{\alpha}\cE_{\alpha}{}^{\beta}+(\bar\sigma^{\mu\nu})^{\dot\alpha}{}_{\dot\beta}\cE^{\dot\beta}{}_{\dot\alpha}\right]\,,
\\
D=-\frac {\ii}2\left(\cE^{\dot\alpha}{}_{\dot\alpha}-\cE_\alpha{}^\alpha\right)\,.
\end{gather}
\end{subequations}
Conversely, $\gl(4)$ generators are expressed in terms of $P,K,M,D$ and the central charge $C\equiv \cE^{\dot\alpha}{}_{\dot\alpha}+\cE_\alpha{}^\alpha$ as
\begin{subequations}
\label{defE}
\begin{gather}
\label{PK}
\cE_{\alpha\dot\beta}=\frac 12\sigma_{\alpha\dot\beta}^{\mu}P_{\mu}\,,\quad \cE^{\dot\alpha\beta}=\frac 12(\bar\sigma^{\mu})^{\dot\alpha\beta}K_{\mu}\,,
\\
\cE_{\alpha}{}^{\beta}=-\frac {\ii}2M_{\mu\nu}(\sigma^{\mu\nu})_{\alpha}{}^{\beta}+\frac 12\delta_{\alpha}^{\beta}\left(\frac 12{C}-\ii\,D\right)\,,
\\
\cE^{\dot\alpha}{}_{\dot\beta}=-\frac {\ii}2M_{\mu\nu}(\bar\sigma^{\mu\nu})^{\dot\alpha}{}_{\dot\beta}+\frac 12\delta^{\dot\alpha}_{\dot\beta}\left(\frac 12{C}{}+\ii\,D\right)\,.
\end{gather}
\end{subequations}
In these relations, summation conventions are as usual. We wrote certain indices of  generators $\cE$ as superscripts, but this is done only to clarify the covariance structure. We mean $\cE_\alpha{}^\beta\equiv \cE_{\alpha\beta}$ etc. that satisfy the $\gl(4)$ commutation relations \eqref{comrel1} \footnote{Note also that in the main text $\mu,\nu\ldots$ means either $\alpha$ or $\dot\alpha$, more precisely $\alpha=1,2$ corresponds to $\mu=3,4$ and $\dot\alpha=\dot 1,\dot 2$ corresponds to $\mu=1,2$. In this appendix however $\mu,\nu,\ldots$, labels Minkowski space-time coordinates and is related to a pair $\alpha\dot\alpha$ through Pauli matrices.}.

The outlined relations pack $\so(2,4)$ generators into $4\times 4$ matrix as $\frac 12\left(\begin{smallmatrix} -\ii(M+D) & K\\ P & -\ii(\bar M-D)\end{smallmatrix}\right)$, where $\bar M=-M^*$, hence $\cE_{ij}$ defined by \eqref{defE} are  $\su(2,2)$ generators in the non-compact basis. To benefit from the oscillator formalism developed in the main text, one still needs to transform to the compact basis which is done by means of the isomorphism \eqref{comptononcomp}. Hence, if one has $M_{AB}=\Sigma_{AB}^{ij}\cE_{ij}$ (relations \eqref{PKMD2}), the relation to the compact basis is given by $M_{AB}=\Sigma_{AB}^{ij}T\gE_{ij}T^{-1}$. For instance, one gets the following explicit expressions for the generators of the compact subalgebra:
\begin{subequations}
\label{HLRcomp}
\begin{gather}
\label{Hosc}
H=-\frac 12 T(\delta^{\dot\beta\alpha}\gE_{\alpha\dot\beta}+\delta_{\beta\dot\alpha}\gE_{\dot\alpha\beta})T^{-1}=\frac 12(-E_{11}-E_{22}+E_{33}+E_{44})=\frac 12(\CN_{b}+\CN_{a})+P\,.
\\
L_m=-\frac 12\sigma^{m}_{\dot\alpha\dot\beta}\gE_{\dot\beta\dot\alpha}\,,\quad R_m=-\frac 12\sigma^{m}_{\alpha\beta}\gE_{\beta\alpha}\,.
\end{gather}
\end{subequations}
Realisation of $H$ explicitly by oscillators in the last equality of \eqref{Hosc} explicitly demonstrates that the highest-weight representations constructed in section~\ref{sec:supq} are positive-energy representations.

\paragraph{Local fields in a four-dimensional CFT.}
Four-dimensional  conformal field theories  are generally formulated in terms of local fields with finite number of components which then implies that they  transform covariantly under the Lorentz group $\SL(2,\mathbb{C})$ and have a definite scale dimension. Time honored method \cite{Mack:1969rr,Mack:1969dg,Mack:1975je} for constructing corresponding representations  of the conformal group is by induction from the parabolic subgroup
 \be
\mathcal{H} = \left(\SL(2,\mathbb{C}) \times \mathcal{D}\right) \ltimes
 \mathsf{K}_4\,,
 \ee
 where $\mathcal{D}$ is the dilatation group  $\SO(1,1)$ and  $\mathsf{K}_{4}$ is the Abelian group of special
conformal
transformations.  Namely, for $g$ being an element of the conformal group, the field $\Phi_{a}(x)$ should transform as $g:\Phi_{a}(x)\mapsto \Lambda_a^b(g,x)\Phi_b(g^{-1}x)$. The parabolic subgroup $\mathcal{H}$
is  the  stability group of the origin $x^{\mu}=0$, and it plays the role of  Wigner's little group. The representation is decided by determining $\Lambda_a^b(h,0)$, for $h\in H$. To have bounded spectrum of scaling dimensions, it is easy to conclude that there should be a primary conformal field which is annihilated by the generators $K_{\mu}$ of special conformal transformations. To have a field with finite-number of components, $\Lambda_a^b$ should realise a finite-dimensional representation of $\SL(2,\mathbb{C})$ \footnote{this  property actually already follows from the bound on the spectrum requirement \cite{Mack:1975je} .}. Hence, these representations are labelled
by  Lorenz spins  $(j_{M},j_{N})$, and the scale
dimension $\Delta$. Scale dimension is defined as $D\Phi_a(0)=\ii\Delta\Phi_a(0)$, and Lorenz spins originate from the quadratic Casimirs of two complex $\sl(2)$ subalgebras of the complexification of $\so(2,4)$:
\eqn
M_{m}=\frac{1}{2} \left( \frac{1}{2} \varepsilon_{mnl} M_{nl}+ \ii M_{0m}
\right)
\ , \qquad
N_{m}=\frac{1}{2} \left( \frac{1}{2} \varepsilon_{mnl} M_{nl}- \ii M_{0m}
\right) \ ,
\enn
where the commutation relations are $[M_{m},M_{n}]=\ii \varepsilon_{mnl}M_{l}$, $[N_{m},N_{n}]=\ii
\varepsilon_{mnl}N_{l}$ , $[M_{m},N_{n}]=0$.

The constructed induced representations are not unitary of course. However, if appropriate bounds on $\Delta$ are respected, they are mapped to the unitary ones using $T$ now not as an isomorphism between two bases but as an intertwiner between two representations. Indeed, from the fact that $M_{AB}=\Sigma_{AB}^{ij}T\gE_{ij}T^{-1}$ it immediately follows that appropriate combinations of $T^{-1}M_{AB}T$ would be directly related to the generators of the $\su(2,2)$ compact basis. Explicitly
\begin{subequations}
\label{intetwinedHMN}
\begin{gather}
T^{-1}(-\ii\,D) T        =   H \label{UL}\,,
\\
T^{-1}M_{m} T  =   L_{m}\,,\quad
T^{-1}N_{m} T  =   R_{m}\,,
\end{gather}
\end{subequations}
and, cf. \eqref{PK},
 \be\label{intertwinedPK}
T^{-1}\big(\sigma_{\alpha\dot\beta}^{\mu}P_{\mu} \big) T =  2\gE_{\alpha\dot\beta}\,, \quad
T^{-1}\big(\bar\sigma_{\mu}^{\dot\alpha\beta}K^{\mu}\big) T = 2\gE_{\dot\alpha\beta} \,.
\ee
Consider highest-weight irrep of $\su(2,2)$ in the compact basis. Its HWS is labelled by $|(j_L,j_R);E_0\rangle$ as was discussed above. This state is annihilated by $\gE_{\dot\alpha\beta}$, but from \eqref{intertwinedPK} it follows that
\be
\gE_{\dot\alpha\beta} |(j_L,j_R);E_0)\rangle =0 \quad \Longrightarrow  K_{\mu} T|
(j_L,j_R); E_0
)\rangle =0\,.
\ee
Furthermore, from \eqref{intetwinedHMN}, we see that the irreducible  module generated by action with generators of $\mathcal{H}$ on  $| \Phi^\Delta_{j_M,j_N}(0) \rangle :=
T| (j_L,j_R);E_0 )\rangle$ has  Lorentz spins  $(j_{M},j_{N})=(j_{L},j_{R})$, and conformal dimension  $\Delta=E_0$. Acting  on  $|\Phi^\Delta_{j_M,j_N}(0) \rangle $ with the translation operators  $e^{-i x^{\mu}P_{\mu}}$, one obtains a coherent state labelled by the coordinate $x_\mu$
\be
e^{-i x^{\mu}P_{\mu}}|\Phi^\Delta_{j_M,j_N}(0) \rangle \equiv |\Phi^{\Delta}_{j_M,j_N}(x_\mu)
\rangle \ .
\ee
 They correspond to states created by the action
of
conformal fields $\Phi^{\Delta}_{j_M,j_N}(x_\mu) $ acting on the vacuum vector
$|0\rangle $   \cite{Gunaydin:1998jc,Chiodaroli:2011pp}\,.
 It is necessary that the bound on $\Delta=E_0$ satisfies the unitarity constraints for positive-energy representations to be able use  these conformal fields in defining a unitary CFT.

\paragraph{Euclidean CFT and radial quantisation.} So far we demonstrated two scenarios where unitary highest-weight representations of $\su(2,2)$ are relevant for physics. In both of them we used $T$, in slightly different roles, to relate the studied representations to the ones of $\su(2,2)$ in the compact basis. There is yet another way to end up with the compact basis. In it, we will not use $T$, but instead will modify the conjugation properties of the conformal algebra generators.

This approach is very natural for the conformal algebra of an {\it Euclidean} conformal field theory in four dimensions\footnote{Discussion is generic for any dimensionality ${\rm d}$. We consider only ${\rm d}=4$ because we need isomorphism $\su(2,2)\simeq \so(4,2)$ to apply the methods of this paper.}. One may  depart from the same commutation relations \eqref{comLor}, now with $\eta_{\mu\wb}=\textrm{diag}(+,+,+,+)$. However, it will be more convenient to redefine slightly the generators to get rid of imaginary units and use the commutation relations
\begin{eqnarray}\label{comrel2}
[\CM_{\mu\wb},\CM_{\rho\sigma}] & = & \delta_{\wb\rho}\,\CM_{\mu\sigma}-
\delta_{\mu\rho}\,\CM_{\wb\sigma} -\delta_{\wb\sigma}\,\CM_{\mu\rho}+
\delta_{\mu\sigma}\,\CM_{\wb\rho}
\,,\cr
[ \CP_{\mu}, \CM_{\rho\sigma} ] & = & \delta_{\mu\rho}\,\CP_{\sigma}-\delta_{\mu\sigma}\,
\CP_{\rho}
\,,\cr
[\CK_{\mu},\CM_{\rho\sigma}]& = & \delta_{\mu\rho}\,\CK_{\sigma}-\delta_{\mu\sigma}\,
\CK_{\rho}
\,,\cr
[\CD,\CM_{\mu\wb}]& = & [\CP_{\mu},\CP_{\wb}] = [\CK_{\mu},\CK_{\wb}]=0
\,,\cr
[\CD,\CP_{\mu}] & = &\CP_{\mu}; \quad [\CD,\CK_{\mu}]=-\CK_{\mu}
\,,\cr
[\CP_{\mu},\CK_{\wb}]& = &-2(\delta_{\mu\wb}\,\CD+\CM_{\mu\wb})\,.
\end{eqnarray}
If we consider these commutation relations as defining a real Lie algebra, then it is $\so(1,5)$. However, we will think about this as a complex *-algebra with the  following conjugation properties
\be\label{conjugation}
\CD^*=\CD\,,\ \ \  \CP_{\mu}^*=\CK_{\mu}\,,\ \ \  \CM_{\mu\wb}^*=-\CM_{\mu\wb}\,.
\ee
The real form defined from such a $^*$-algebra structure is $\so(2,4)$. Indeed, by doing the Euclidean analog of \eqref{defE}, we can pack the generators into $4\times 4$ matrices as $\frac 12\left(\begin{smallmatrix}\ii\CM+\CD& \CK\\ \CP & \ii\bar\CM-\CD\end{smallmatrix}\right)$, which is clearly the $\su(2,2)$ in the compact basis. Note that for UHW representations, eigenvalues of the dilatation operator play the role of energy.

Choosing  $\CP$ and $\CK$ to be conjugates of each other may appear strange. However we recall that in 2-dimensional CFT's, this is the conjugation property one would deduce from the widely used *-structure of the Virasoro algebra ($L_n^*=L_{-n}$).

The conjugation \eqref{conjugation} is indeed a natural one to impose when we work in the operator-based formalism and choose to do radial quantisation. The physical vacuum $|\Omega\rangle$ is invariant under the action of  symmetry generators. By operator-state correspondence, action of local fields evaluated at the origin creates states in the Hilbert space: $\Phi(0)|\Omega\rangle=|\Phi\rangle$. The coordinate dependence, i.e. transformation properties of  states like $\Phi(x)|\Omega\rangle\equiv |\Phi(x)\rangle$, is inferred from the commutation relations (given for a primary scalar field of dimension $\Delta$ for simplicity):
\begin{subequations}
\begin{align}
[\CP_{\mu},\Phi(x)]&=\partial_{\mu}\Phi(x)\,,
\\
\label{KPhi}
[\CK_{\mu},\Phi(x)]&=(-x^2I_{\mu\wb}\partial_{\wb}+2\Delta x_{\mu})\Phi(x)\,,
\\
[\CM_{\mu\wb},\Phi(x)]&=(x_{\wb}\partial_{\mu}-x_{\mu}\partial_{\wb})\Phi(x)\,,
\\
[\CD,\Phi(x)]&=(x_{\mu}\partial_{\mu}+\Delta)\Phi(x)\,,
\end{align}
\end{subequations}
where $I_{\mu\wb}=\delta_{\mu\wb}-2\frac{x_{\mu}x_{\wb}}{x^2}$. We have for instance $|\Phi(x)\rangle=e^{x_{
\mu}P_{\mu}}|\Phi\rangle$ etc.

Furthermore one would like to have   the correlation function $\langle \Phi(x)\Phi(0) \rangle$ to be  equal to the vacuum average  $\langle \Omega | \Phi(x)\Phi(0)|\Omega\rangle$, and choose the normalisation  $\langle \Phi |\Phi \rangle=1$ . To this end, one identifies Hermitian conjugation as an inversion of the spacial coordinate

\be
x_{\mu}^\dagger =y_{\mu}\equiv\frac{x_{\mu}}{|x|^2}\,,
\ee
and conjugation of quantum fields  as $\Phi^{\dagger }=\Phi'$, where prime denotes that the field should be computed in the inverted reference frame (labeled by $y$) according to its covariance properties. E.g. for a primary scalar field of dimension $\Delta$:  $(\Phi(x))^{\dagger }=\Phi^{\dagger }(x^{\dagger })=\Phi'(y)=\left|\frac{\partial x}{\partial y}\right|^{\Delta/4}\Phi(x)=|x|^{2\Delta}\Phi(x)$. This allows to consistently identify \eqref{conjugation} as the conjugation properties of the symmetry generators\footnote{Another way to interpret these conjugation properties as introduction of reflection positivity  when time is radial.}. Consider for instance
\begin{eqnarray}
[\CD,\Phi(x)]^{\dagger } &=&-[\CD^{\dagger },(\Phi(x))^{\dagger }]=-|x|^{2\Delta}[\CD^{\dagger },\Phi(x)]\,,
\\ \vphantom{}
[\CD,\Phi(x)]^{\dagger } &=& ((x\partial+\Delta)\Phi(x))^{\dagger }=(-x\partial+\Delta)|x|^{2\Delta}\Phi(x)=|x|^{2\Delta}(-x\partial-\Delta)\Phi(x)=-|x|^{2\Delta}[\CD,\Phi(x)]\,,\nonumber
\end{eqnarray}
which implies $[\CD^\dagger ,\Phi(x)]=[\CD,\Phi(x)]$, in agreement with \eqref{conjugation}.

Note that it would be tempting  to think along the lines: $\CD=x\partial+\ldots $, and hence $\CD^\dagger =(x\partial+\ldots)^\dagger =-x\partial+\ldots=-\CD$.  $\CD$ is, however, not a differential operator but rather an operator acting on the Hilbert space, Noether charge, whose commutation with another operator $\Phi(x)$ produces $(x\partial+\ldots) \Phi(x)$. By itself, $\CD$ does not act on the coordinate $x_{\mu}$ which is just a label in this context. Taking this into account gives correct commutation relations \eqref{comrel2} and correct signs under conjugation.

We confirm now our identifications by computation of the 2-point function: $\langle \Omega |\Phi(x)=(|x|^{2\Delta}\Phi(x)|\Omega\rangle)^{\dagger }=\frac 1{|x|^{2\Delta}}(e^{x_{\mu}P_{\mu}}|\Phi))^{\dagger }=\frac 1{|x|^{2\Delta}}\langle\Phi| e^{y_{\mu}K_{\mu}}$, and thus $\langle \Phi(x)\Phi(0) \rangle=\langle \Omega | \Phi(x)\Phi(0)|\Omega\rangle=\frac 1{|x|^{2\Delta}}\langle\Phi| e^{y_{\mu}K_{\mu}}|\Phi\rangle=\frac 1{|x|^{2\Delta}}$, as $K_{\mu}|\Phi\rangle=0$ according to \eqref{KPhi} and invariance of the vacuum.  In similar fashion one can compute two-point functions of fields with arbitrary spin.

\subsection{Oscillator construction  of the UHW representations of $\SU(2,2)$ }
The most general  highest-weight state of UHW representation of $\su(2,2)$ is of the form
\be
|{\rm HWS}\rangle=(b_1^\dagger)^n(a_1^\dagger)^m(\Delta_b^\dagger)^{\gamma_L}(\Delta_a^\dagger)^{\gamma_R}|0\rangle\,,
\ee
where $n,m$ are non-negative integers and $\gamma_L,\gamma_R>-1$, with certain restrictions on when these four parameters are non-zero to be specified below.  The fundamental weight $[\nu_L,\nu_R ]$  is given by
\be
[\nu_L;\nu_R] = [0,- n; m ,0] + [-(P+\gamma_L), -(P+\gamma_L ); \gamma_R, \gamma_R]\,,
\ee
where $P$ is the number of colours in the oscillator algebra.

From \eqref{HLRcomp} we conclude that $|{\rm HWS}\rangle= |(j_L,j_R) ; E_0 \rangle $, where
\be
j_L=\frac n2\,\quad j_R=\frac m2\,,\quad E_0= j_L +   j_R + P +\gamma_L +\gamma_R = j_L +j_R + \beta\,.
\ee
One should also mention  the operator
\be
\mathcal{Z}= \frac 12(\CN_a - \CN_b) = \frac 12C-P\,,
\ee
where $C$ is the central charge. This operator, although it is not a part of $\su(2,2)$, is a part of $\algu(2,2)$ up to the addition of $P$ and is well-defined in the oscillator formalism. For massless representations ($P=1$, see below) the value of $\mathcal{Z}$ coincides with the helicity of the corresponding particles \cite{Mack:1969dg}, and it is common to call $\mathcal{Z}$ as the helicity operator even in the case of other representations \cite{Govil:2013uta}. It can be used to further distinguish between isomorphic $\su(2,2)$ representations that have a different oscillator content.

Classification theorem~\ref{th:supq} specified for this case tells us that the representation is unitary in the following cases: if $j_L=j_R=0$ then for $\beta=0$ (trivial representation) and real $\beta\geq 1$; if $j_Lj_R=0$ but $j_L+j_R>0$ then for real $\beta\geq 1$; if $j_Lj_R\neq 0$ then for real $\beta\geq 2$. This is of course in complete agreement with Mack  \cite{Mack:1975je} who first classified the positive energy unitary representations of $\SU(2,2)$.

We will now analyse explicitly what representations can be constructed for various values of $P$. The Fock vacuum itself is a highest-weight state of non-trivial representation which depends on $P$
\be
|0\rangle =  |(j_L= 0, j_R=0); E_0=P \rangle
\ee
Consider now HWS's obtained by non-empty oscillator content acting on the Fock vacuum.

For $P=1$ we must have $\gamma_L =\gamma_R=0$ as there are not enough colours to construct determinants, and the possible HWS's are
\begin{eqnarray}
(b_{\dot{1}}^\dagger )^n \fvac &=& |(j_L= \frac{n}{2}, j_R=0); E_0=\frac{n}{2} +1 \rangle\,, \\
(a_1^\dagger )^m \fvac  &=& |(j_L= 0, j_R=\frac{m}{2}); E_0=\frac{m}{2} +1 \rangle\,.
\end{eqnarray}
These representations describe massless conformal fields \cite{Gunaydin:1984fk,Gunaydin:1998jc,MR848089}. Indeed, the   four-dimensional Poincar\'e mass is given by $m^{2}=P_{\mu}P^{\mu}$, but $T^{-1}P_{\mu}P^{\mu}T\propto \det\limits_{1\leq \alpha,\beta\leq 2} \gE_{\alpha\dot\beta}=\det\limits_{1\leq \alpha,\beta\leq 2} a_{\alpha}^\dagger b_{\dot\beta}^{\dagger}=0$. Note that if $m^2\neq 0$ then mass is not an invariant of the conformal algebra.  $P=1$ representations  were called doubleton representations in \cite{Gunaydin:1984fk,Gunaydin:1998jc,Gunaydin:1998sw}. They were also shown to
describe  the minimal unitary representation (minrep) of $\SU(2,2)$ and its  deformations labelled by helicity \cite{Fernando:2009fq}.

Since the minimum number of colours needed to construct an oscillator determinant is two, we find that doubleton representations can not acquire anomalous dimensions since $\gamma_L,\gamma_R$ are the only continuous parameters available. As a consequence, if a field is massless then it must be free  in a conformal field theory. A proof of this using different methods was given by Weinberg relatively recently \cite{Weinberg:2012cd}.

For $P=2$ one can form $\su(2)_L$ or $\su(2)_R$ invariant  oscillator determinants  involving  $a$ or $b$ type oscillators:
\begin{eqnarray} \label{oscdeterminant}
\Delta_a^\dagger &=& a_1(1)^\dagger  a_2(2)^\dagger  - a_2(1)^\dagger  a_1(2)^\dagger\,,  \nn \\
\Delta_b^\dagger &=& b_{\dot{1}}(1)^\dagger  b_{\dot{2}}(2)^\dagger  - b_{\dot{2}}(1)^\dagger  b_{\dot{1}}(2)^\dagger\,.
\end{eqnarray}

 Hence the possible  highest weights of irreducible UHW representations for $P=2$  are of three  types.
First is the set of highest-weight states of the form
 \be
 (a_1(1)^\dagger )^m (b_{\dot{1}}(2)^\dagger )^n \fvac = |(j_L=\frac{n}{2},j_R=\frac{m}{2} ); E_0= \frac{m+n}{2} + 2 \rangle\,,
 \ee
 or the equivalent set with the colour indices interchanged.

Second is the set of UHW reps with the HWS
\be
(a_1(i)^\dagger )^m [\Delta_a^\dagger]^{\gamma_R}  \fvac = |( j_L=0,
j_R=\frac{m}{2} ; E_0 = \frac{m}{2}+2+\gamma_R \rangle\,,
\ee
where the colour index $i$ could be 1 or 2 and $\gamma_R  > -1$.

The third  set of irreducible UHW representations is the parity conjugate of the second set
 with the highest-weight states
\be
(b_{\dot{1}}(i)^\dagger  )^n [ \Delta_b^\dagger]^{\gamma_L}  \fvac =
|(j_L=\frac{n}{2},
j_R=0); E_0 =  \frac{m}{2}+2+ \gamma_L
\rangle\,.
\ee
The unitary irreducible representations of $\SU(2,2)$ for $P=2$ describe massless fields in five-dimensional anti-de Sitter spacetime \cite{Gunaydin:1984fk,Gunaydin:1998sw}.

The above highest-weight states remain highest weight when we increase the number of colours $P$ (with the difference that $E_0$ is increased by the same integer amount). For $P= 3$  the additional HWS's of new type are of the form (modulo permutation of colour indices)
\be
(a_1(i)^\dagger  )^m  (b_{\dot{1}}(3)^\dagger  )^n \, [\Delta_a^\dagger]^{\gamma_R} \fvac\,,
\ee
or of the form
\be
(a_1(3)^\dagger  )^m  (b_{\dot{1}}(i)^\dagger  )^n \, [\Delta_b^\dagger]^{\gamma_L} \fvac
\ee
where $mn \neq 0$, $i=1$ or $2$  and $\gamma_L,\gamma_R >-1$. The resulting irrep is labeled by
\be
|( j_L=\frac{n}{2},
j_R=\frac{m}{2} ) ; E_0 = 3+ \frac{m +n}{2}+\gamma \rangle
\ee
where $\gamma=\gamma_L$ or $\gamma=\gamma_R$.

Again, if we increase the number of colours $P$ the above highest-weight states remain highest weight and lead to infinite families of irreducible UHW representations. For $P > 3$, one does not find any new $\su(2,2)$ representations that is not obtainable from the set of highest weights for $ P\leq 3$ by increasing the number of colours, as we can control now $E_0$ by the value of $\gamma$ as well. Hence the above UHW irreps  exhaust the list of all positive energy unitary irreps of $\su(2,2)$.

\section{\label{sec:superconformal} Unitary representations of the superconformal algebra $\su(2,2|4)$ }

In this section we shall specialize the general results obtained in the main text to the
 conformal superalgebra  $\su(2,2|\mathcal{N})$. To keep things concrete, we will explicitly discuss the $\mathcal{N}=4$ case only, and will sometimes impose vanishing of the central charge as well, to restrict to representations of $\psu(2,2|4)$. However, the discussion is easily generalisable to arbitrary $\mathcal{N}$.

 We recall that the unitary supermultiplets of $\su(2,2|4)$ were first obtained using
the oscillator method by Gunaydin and Marcus (GM)  \cite{Gunaydin:1984fk} in their work
on the Kaluza-Klein spectrum of   type IIB supergravity on AdS$_5\times$S$^5$.
The symmetry superalgebra of   AdS$_5\times$S$^5$  background of type IIB supergravity is $\psu(2,2|4)$ with the even subalgebra
$\su(2,2) \oplus \su(4)$ where $\su(2,2)$ is isomorphic to the Lie algebra of
the isometry group of AdS$_5$ and $\su(4)$ is isomorphic to the Lie algebra
of the isometry group of S$^5$. The full spectrum of type IIB supergravity was obtained
 by a simple tensoring of an ultrashort CPT self-conjugate supermultiplet
(doubleton) and restricting to the CPT self-conjugate massive BPS multiplets.
Authors of  \cite{Gunaydin:1984fk} pointed out that the doubleton supermultiplet decouples from the spectrum
as gauge modes and its field theory lives on the boundary of AdS$_5$, and
this boundary theory is simply the four dimensional $\mathcal{N}\!=\!4$ super Yang-Mills
theory which was known to be conformally invariant. Boundary was identified with Minkowski
space, AdS$_5$ group $\SO(4,2)$ acts on the boundary as a conformal group, and the CPT self-conjugate doubleton supermultiplet of $\psu(2,2|4)$ is simply
the $4d$ $\mathcal{N}\!=\!4$ Yang-Mills supermultiplet. These results represent some of the earliest work on  AdS/CFT duality in
a true Wignerian sense within the framework of Kaluza-Klein supergravity.

\subsection{Five-graded decomposition}
We shall use the $\su(2,\!|4|2)$ grading of $\su(2,2|4)$ and the matrix \eqref{oscfull} to realise generators of the algebra using oscillators.  There is a natural "compact" five grading of $\su(2,\!|4|2)$ that extends the three grading of $\su(2,2)$ with respect to its maximal compact subalgebra:
\be
\su(2,\!|4|2)= \CE^{-1} \oplus \mathcal{F}^{-1/2} \oplus \mathcal{K}_{\rm c} \oplus \mathcal{F}^{1/2} \oplus \CE^{+1}\,,
\ee
where $\mathcal{F}^{\pm 1/2} $ denote the odd  generators which contain the bilinears
\be
\mathcal{F}^{-1/2} = \begin{pmatrix} -Q_L^{a\dot\beta}=  -f^\dagger _a\cdot b^\dagger _{\dot\beta} \\ Q_{Ra}^{\beta}=a^\dagger _{\beta}\cdot f_{a} \end{pmatrix}  \quad , \quad \mathcal{F}^{1/2} = \begin{pmatrix} S_{L\, a \dot\alpha}=b_{\dot\alpha}\cdot f_{a}\, \\  S_{R\beta}^a= f^\dagger_a\cdot a_{\beta} \end{pmatrix}\,.
\ee
Even generators are
\be
\CE^{-1} =-a^\dagger _{\alpha} \cdot b^\dagger _{\dot\beta}  \quad , \quad \mathcal{K}_{\rm c} = \begin{pmatrix} -b_{\dot\alpha}\cdot b^\dagger _{\dot\beta} \\f^\dagger _a\cdot f_b \\a^\dagger _{\alpha}\cdot a_{\beta} \end{pmatrix} \quad , \quad \CE^{+1} = b_{\dot\alpha}\cdot a_{\beta}\,.
\ee
One  can pack the above generators into an $8\times 8$ matrix in the following way
\be
\left(
\begin{matrix}
\CK_{\rm c} & \CF^{1/2} & \CE^{+1} \\
\CF^{-1/2} & \CK_{\rm c} & \CF^{1/2} \\
\CE^{-1} & \CF^{-1/2} & \CK_{\rm c}
\end{matrix}
\right)
=
\left(
\begin{matrix}
\CK_{\rm c} & S_L & \CE^{+1} \\
Q_L & \CK_{\rm c} & S_R \\
\CE^{-1} & Q_R & \CK_{\rm c}
\end{matrix}
\right)\,.
\ee
The highest-weight vectors of irreducible unitary representations are annihilated by all the 16 odd generators in the subspace $\mathcal{F}^{1/2}$. Under the isomorphism $T$ \eqref{comptononcomp} these generators get mapped into the special conformal supersymmetry generators $ S_{a}^{\alpha} $ and $ \bar{S}^{a\dot{\alpha}}$
\be \begin{pmatrix} \bar{S}_b^{\dot{\alpha}} \\ S^{a\alpha} \end{pmatrix} = T \begin{pmatrix} S_{Lb\dot\alpha} \\  S_{R\alpha}^a \end{pmatrix} T^{-1}
\ee
which close into special conformal generators $K^{\dot\alpha{\alpha}} = \bar\sigma_{\mu}^{\dot\alpha\alpha}K^{\mu} $ under anticommutation.
 On the other hand,  the odd generators in $\mathcal{F}^{-1/2}$ get mapped into the Poincare supersymmetry generators $Q_{a\alpha}$ and $\bar{Q}^{a}_{\dot{\alpha}}$
 \be \begin{pmatrix} \bar{Q}^b_{\dot{\alpha}} \\ Q_{a\alpha} \end{pmatrix} = T \begin{pmatrix} Q_{L}^{b\dot\alpha} \\  Q_{Ra}^{\alpha}\end{pmatrix} T^{-1} \ee
which close into the generators of translations $P_{\alpha \dot{\alpha}}= \sigma_{\alpha\dot\alpha}^{\mu}P_{\mu}$ under anticommutation.

\subsection{\label{sec:Genusho}Generic highest-weight state, unitarity constraints, and shortenings}
The most general highest-weight state has the form
\be
\HWS = f_1^{\lambda_1}f_2^{\lambda_2}f_3^{\lambda_3}f_4^{\lambda_4}(b_1^{\dagger})^{n_L}(a_1^{\dagger})^{n_R}(\Delta_b^\dagger)^{\gamma_L}(\Delta_a^\dagger)^{\gamma_R}\fvac\,, \label{HWS}
\ee
with appropriate dependence of oscillators on colours which is not explicitly written. Conventions used in section~\ref{sec:supqm} are restored by noticing that $\lambda_4=|{\bf F}_{\!\Delta}|$ and $\tau_a=\lambda_a-\lambda_4$. The fundamental weight $[\nu_L;\lambda; \nu_R]$ is given by
\be
[\nu_L;\lambda; \nu_R]=[0,-n_L;\lambda_1,\lambda_2,\lambda_3,\lambda_4;n_R,0]+[-(P+\gamma_L),-(P+\gamma_L);0,0,0,0;\gamma_R,\gamma_R]\,.
\ee
The tuple $[\mu_L,\tau,\mu_R;\beta_L,\beta_R]$ would be explicitly written as $[n_L,\ \tau_1,\tau_2,\tau_3,\ n_R;\beta_L,\beta_R]$, where $\beta_L=\gamma_L+P-\lambda_1$, and $\beta_R=\gamma_R+\lambda_4$. The corresponding non-compact Young diagram is shown in Fig~\ref{fig:ncYoungapp}.
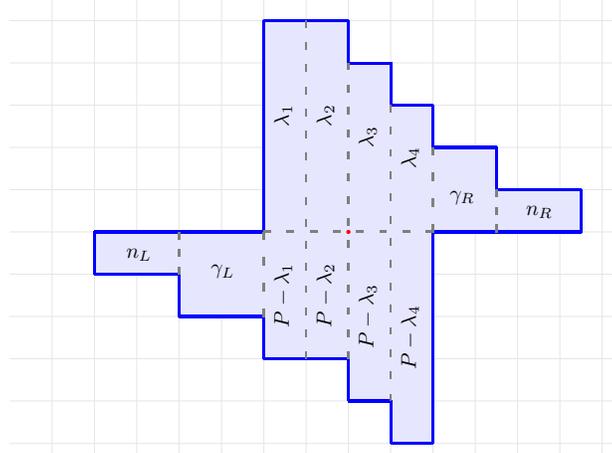
\begin{figure}[t]
\begin{center}
\scalebox{0.8}{
\begin{picture}(300,180)(-130,-100)
\color{gray!20}
\thinlines
\multiput(-120,-100)(0,20){11}{\line(1,0){290}}
\multiput(-100,-105)(20,0){14}{\line(0,1){215}}
\thicklines
\color{blue!10}
\polygon*(0,0)(0,100)(20,100)(20,100)(40,100)(40,80)(60,80)(60,60)(80,60)(80,40)(110,40)(110,20)(150,20)(150,0)(80,0)(80,-100)(60,-100)(60,-80)(40,-80)(40,-60)(0,-60)(0,-40)(-40,-40)(-40,-20)(-80,-20)(-80,0)(0,0)
\color{blue}
\drawline(0,0)(0,100)(20,100)(20,100)(40,100)(40,80)(60,80)(60,60)(80,60)(80,40)(110,40)(110,20)(150,20)(150,0)(80,0)(80,-100)(60,-100)(60,-80)(40,-80)(40,-60)(00,-60)(0,-40)(-40,-40)(-40,-20)(-80,-20)(-80,0)(0,0)

\color{gray}
\dashline{4}(0,0)(80,0)
\dashline{4}(20,-60)(20,100)
\dashline{4}(40,-60)(40,80)
\dashline{4}(60,-80)(60,60)
\dashline{4}(60,-80)(60,60)
\dashline{4}(80,0)(80,40)
\dashline{4}(0,-40)(0,0)
\dashline{4}(110,0)(110,20)
\dashline{4}(-40,0)(-40,-20)

\color{red}
\put(40,0){\circle*{2}}
\color{black}
\put(5,50){{\rotatebox{90}{$\wf_1$}}}
\put(25,50){{\rotatebox{90}{$\wf_2$}}}
\put(45,40){{\rotatebox{90}{$\wf_3$}}}
\put(65,30){{\rotatebox{90}{$\wf_4$}}}
\put(5,-45){{\rotatebox{90}{$P-\wf_1$}}}
\put(25,-45){{\rotatebox{90}{$P-\wf_2$}}}
\put(45,-55){{\rotatebox{90}{$P-\wf_3$}}}
\put(65,-65){{\rotatebox{90}{$P-\wf_4$}}}
\put(88,15){$\gamma_R$}
\put(124,8){$n_R$}
\put(-25,-20){$\gamma_L$}
\put(-65,-12){$n_L$}
%
\end{picture}
}

\caption{\label{fig:ncYoungapp}Non-compact Young diagram for a long $\su(2,2|4)$ representation, labels of weights in $\su(2,\!|4|2)$ grading. The representation has zero central charge if number of boxes in the first quadrant is equal to the number of boxes in the third quadrant \cite{Marboe:2017dmb} (quadrants are defined w.r.t. the centre of the weight lattice marked by a red dot).}
\end{center}
\end{figure}

The unitarity constraints can be read off for instance from the classification theorem \eqref{th:final}. In this context it defines the constraints on $\beta_{L/R}$: if $n_{L/R}=0$ then $\beta_{L/R}=0$ or $\beta_{L/R}\geq 1$; if $n_{L/R}\neq 0$ then $\beta_{L/R}\geq 1$. This can be shown to be equivalent to the classification of UHW's by Dobrev and Petkova \cite{Dobrev:1985vh}.

In the absence of zero central charge restriction, $\beta_L$ and $\beta_R$ are two independent (possibly continuous) parameters. But if one wants to consider representations of $\psu(2,2|4)$, one should additionally impose the zero-charge constraint which is equivalent to
\be
2(\beta_R-\beta_L)=\lambda_1-\lambda_2-\lambda_3+\lambda_4+n_L-n_R\,.
\ee
Hence if it is forbidden to change continuously one of $\beta$'s, the second cannot be continuous either.

One can also discuss unitarity constraints from the point of view of shortening conditions. We remind that the multiplet gets shortened if HWS gets annihilated by some of the odd generators in the subspace $\mathcal{F}^{-1/2}$ (or by their combination).  Analysis of what shortenings can happen was done in detail in section~\ref{sec:shortening} for the compact case, but it generalises without any modifications to the non-compact case as well. Action of $Q_R$ is analysed by considering the upper half-plane part of the Young diagram, and action of $Q_L$ is analysed by considering the lower half-plane part of the Young diagram. The final outcome is controlled by $P$ and the values of $\lambda_a$, $a=1,\ldots,4$:

\be
\label{RBPS}
\lambda_a=0\quad &\Rightarrow&  Q_{Ra}^{\beta}\HWS=0\,,\ \ \beta=1,2\,,
\\
\label{Rsemishort}
\lambda_a=1\quad &\Rightarrow&  Q_{Ra}^{1}Q_{Ra}^{2}\HWS =0\,,
\\
\hspace{1em}\nonumber
\\
\label{LBPS}
P-\lambda_a=0\quad &\Rightarrow&  Q_{L}^{a\dot\beta}\HWS =0\,,\ \ \dot\beta=\dot 1,\dot 2\,,
\\
\label{Lsemishort}
P-\lambda_a=1\quad &\Rightarrow&  Q_{L}^{a\dot 1}Q_{L}^{a\dot 2}\HWS =0\,.
\ee
Shortening of type $Q_{L/R}\HWS =0$ is called the BPS shortening condition. If it is satisfied then it must be that $\gamma_{L/R}=0$ and $n_{L/R}=0$. Shortening of type $Q_{L/R}Q_{L/R}\HWS  =0$ is called semi-shortening condition. If it is satisfied then it must be that $\gamma_{L/R}=0$, but no restrictions on $n_{L/R}$ are imposed.

Although shortened multiplets are formally protected from  quantum corrections, it is possible that several (two or four) short multiplets related in a special way combine together into a long one which is no longer protected from quantum corrections. This is known as the Konishi mechanism. Below is an example of the Konishi multiplet which is a long one at finite coupling but decomposes into four short multiplets as coupling goes to zero ($\gamma\to -1$):
\be
\label{Konishidecomposition}
\raisebox{-0.4\height}{
\begin{picture}(40,40)(0,0)
\color{gray}
\polygon(0,0)(0,40)(40,40)(40,0)
\color{blue!10}
\polygon*(0,20)(0,40)(35,40)(35,20)(40,20)(40,0)(5,0)(5,20)(0,20)
\color{blue}
\drawline(0,40)(40,40)
\drawline(35,40)(35,20)
\drawline(0,0)(40,0)
\drawline(5,0)(5,20)
\color{gray}
\drawline(0,20)(40,20)
\drawline(0,10)(40,10)
\drawline(0,30)(40,30)
\drawline(10,0)(10,40)
\drawline(20,0)(20,40)
\drawline(30,0)(30,40)
\color{black}
\put(7,27){$4+\gamma$}
\put(11,8){$4+\gamma$}
\end{picture}
}
\quad
\xrightarrow{\gamma\to -1}
\quad
\raisebox{-0.4\height}{
\begin{picture}(40,40)(0,0)
\color{gray}
\polygon(0,0)(0,40)(40,40)(40,0)
\color{blue!10}
\polygon*(0,20)(0,60)(10,60)(10,40)(30,40)(30,20)(40,20)(40,-20)(30,-20)(30,0)(10,0)(10,20)(0,20)
\color{magenta!10}
\polygon*(0,0)(0,20)(10,20)(10,0)
\polygon*(30,20)(30,40)(40,40)(40,20)
\color{gray}
\drawline(0,20)(40,20)
\drawline(0,10)(40,10)
\drawline(0,0)(40,0)
\drawline(0,30)(40,30)
\drawline(0,40)(40,40)
\drawline(10,0)(10,40)
\drawline(20,0)(20,40)
\drawline(30,0)(30,40)
\drawline(40,0)(40,40)
\drawline(0,50)(10,50)
\drawline(30,-10)(40,-10)

\color{blue}
\drawline(0,20)(0,60)(10,60)(10,40)(30,40)(30,20)(40,20)(40,-20)(30,-20)(30,0)(10,0)(10,20)(0,20)

\end{picture}
}
\quad\oplus\quad
\raisebox{-0.4\height}{
\begin{picture}(40,40)(0,0)
\color{gray}
\polygon(0,0)(0,40)(40,40)(40,0)
\color{blue!10}
\polygon*(0,30)(40,30)(40,10)(0,10)(0,30)
\color{magenta!10}
\polygon*(0,30)(0,40)(40,40)(40,30)
\polygon*(0,0)(0,10)(40,10)(40,0)
\color{blue}
\drawline(0,30)(40,30)(40,10)(0,10)(0,30)
\color{gray}
\drawline(0,20)(40,20)
\drawline(10,00)(10,40)
\drawline(20,00)(20,40)
\drawline(30,00)(30,40)
\end{picture}
}
\quad\oplus\quad
\raisebox{-0.2\height}{
\begin{picture}(40,80)(0,0)
\color{magenta!10}
\polygon*(0,0)(0,40)(40,40)(40,0)
\color{gray}
\polygon(0,0)(0,40)(40,40)(40,0)
\color{blue!10}
\polygon*(0,20)(0,50)(10,50)(10,30)(40,30)(40,0)(10,0)(10,20)(0,20)
\color{blue}
\drawline(0,20)(0,50)(10,50)(10,30)(40,30)(40,0)(10,0)(10,20)(0,20)
\color{gray}
\drawline(0,40)(10,40)
\drawline(0,30)(10,30)
\drawline(10,20)(40,20)
\drawline(0,10)(40,10)
\drawline(10,20)(10,30)
\drawline(20,0)(20,40)
\drawline(30,0)(30,40)
\end{picture}
}
\quad\oplus\quad
\raisebox{-0.2\height}{
\begin{picture}(40,80)(0,0)
\color{magenta!10}
\polygon*(0,0)(0,40)(40,40)(40,0)
\color{gray}
\polygon(0,0)(0,40)(40,40)(40,0)
\color{blue!10}
\polygon*(0,40)(30,40)(30,20)(40,20)(40,-10)(30,-10)(30,10)(0,10)(0,40)
\color{blue}
\drawline(0,40)(30,40)(30,20)(40,20)(40,-10)(30,-10)(30,10)(0,10)(0,40)
\color{gray}
\drawline(0,30)(40,30)
\drawline(0,20)(30,20)
\drawline(30,10)(40,10)
\drawline(30,0)(40,0)
\drawline(10,00)(10,40)
\drawline(20,00)(20,40)
\drawline(30,10)(30,20)
\end{picture}
}\,.
\ee
\\[1em]
The left diagram is the long multiplet with $\gamma_L=\gamma_R=\gamma<0$. It is drawn in the same style as the right diagram of Fig.~\ref{fig:GenericYoung}. The four right figures are short multiplets. The plaquettes in blue are the non-compact Young diagrams. The plaquettes of the weight lattice outside the diagrams, that is the plaquettes  where shortenings occur, are highlighted in pink, cf. Fig.~\ref{fig:samplefathook}.

Conditions for a short multiplet be capable of joining with other ones to form a long multiplet is the following one (given, for simplicity, under zero central charge assumption): If $\lambda_4\neq 0$ or if $\lambda_4=0$ and $\lambda_3>1$ then it can join. Otherwise, it cannot. Explanation of this criterium using Young diagrams is given in \cite{Marboe:2017dmb}.

\subsection{Relation to some notations used in literature}
The conventions of Dolan and Osborn \cite{Dolan:2002zh} are often used for classification purposes, in particular for short and semi-short supermultiplets. The notation  $[k,p,q]_{j,\bar j}^{\Delta}$ and related notations are used to label multiplets. Here $[k,p,q]$ are Dynkin labels given by $k=\lambda_1-\lambda_2$, $p=\lambda_2-\lambda_3$, $q=\lambda_3-\lambda_4$; $j,\bar j$ are Lorenz spins given by $j=j_L=n_L/2$ and $\bar j=j_R=n_R/2$. Finally, $\Delta$ is the conformal dimension which is given by
\be
\Delta=E_0=\frac{n_L+n_R}{2}+\gamma_L+\gamma_R+P\,.
\ee

Also,  different classes of multiplets are distinguished in \cite{Dolan:2002zh}: $\CA,\CB,\mathcal{C},\mathcal{D},\bar{\mathcal{D}}$. We now review what  these multiplets are in our formalism.

$\CA_{[k,p,q](j,\bar j)}^{\Delta}$ labels generic long multiplets with no shortenings present.

$\CB_{[k,p,q](0,0)}^{s,\bar s}$ labels ($s$,$\bar s$)-BPS multiplets -- those who have BPS shortenings for both $Q_L$ and $Q_R$ generators. Therefore, at least conditions $\lambda_4=P-\lambda_1=0$ are satisfied.  $s$ denotes fraction of $Q_L$ charges annihilating $|{\rm HWS}\rangle$ and $\bar s$ denotes a similar fraction for $Q_R$, they must be multiples of $\frac 14=\frac 1{\mathcal{N}}$.  As discussed, these multiplets must have zero spins and vanishing $\gamma_L,\gamma_R$.  Hence the conformal dimension is fixed by $\su(4)$ Dynkin labels: $\Delta=P=\lambda_1=k+p+q$.

There are only two types of BPS multiplets with zero central charge. The first one is (1/2,1/2)-BPS, with $\lambda_1=\lambda_2=P$ and $\lambda_3=\lambda_4=0$. Such multiplets are protected from quantum corrections (cannot join with others). The second one is (1/4,1/4)-BPS, with $[\lambda_1,\lambda_2,\lambda_3,\lambda_4]=[P,P-q,q,0]$. In the case of $q=1$ the multiplet enjoys extra semi-shortening condition, and it is protected from quantum corrections as well. For $q>1$ it can join with other appropriate multiplets and is in this sense not protected.

The first diagram in the r.h.s of \eqref{Konishidecomposition} is a (1/4,1/4)-BPS multiplet with $P=4$ and $q=2$.

$\mathcal{C}_{[k,p,q](j,\bar j)}^{(t,\bar t)}$ labels multiplets that have semi-shortenings for both $Q_L$ and $Q_R$. Therefore, at least conditions $\lambda_4=P-\lambda_1=1$ are satisfied. $t$ denotes fraction of \eqref{Lsemishort} to hold (out of maximal possible number, which is 4). $\bar t$ denotes fraction of $\eqref{Rsemishort}$ to hold. As discussed, these multiplets have $\gamma_L,\gamma_R=0$ but they may have non-zero Lorentz spins. Their conformal dimension is also fixed by other charges: the spins and Dynkin labels. These multiplets can always join to form a long one.

An example of a multiplet from this class is the second diagram in the r.h.s of \eqref{Konishidecomposition}. In the notation of Dolan and Osborn, it is $\mathcal{C}_{[0,0,0](0,0)}^{(1,\bar 1)}$ - a multiplet with maximal possible number of semi-shortenings.

$\mathcal{D}_{[k,p,q](0,\bar j)}^{(s,\bar t)}$ is the case when $Q_L$ charges act with BPS shortenings and $Q_R$ charges act with semi-shortenings. Example is the third diagram on the r.h.s. of \eqref{Konishidecomposition}, which is $\mathcal{D}_{[2,0,0](0,0)}^{(\frac 14,\bar{\frac 34})}$.

$\bar{\mathcal{D}}_{[k,p,q](j,0)}^{(t,\bar s)}$ is the case when $Q_L$ charges act with semi-shortenings and $Q_R$ charges act with BPS shortenings. Example is the last diagram on the r.h.s. of \eqref{Konishidecomposition}, which is $\bar{\mathcal{D}}_{[0,0,2](0,0)}^{(\frac 34,\bar{\frac 14})}$.

\subsection{Tables of multiplets}
We finalise our discussion by explicitly constructing various supermultiplets. More of multiplet combinatorics using the Young diagram technique is available in \cite{Marboe:2017dmb}.

 For $P=1$ we list the possible highest-weight vectors of $\su(2,\!|4|2)$, their fundamental weights $[\nu_L;\lambda; \nu_R]$, the corresponding tuples $[\mu_L,\tau,\mu_R;\beta_L,\beta_R]$ labelling the unitary highest-weight representations (doubletons) and their BPS properties in Table
 \ref{tableP1}\footnote{ We recall  that $\mu_L= 2 j_L$ and $\mu_R=2j_R$ where $j_L$ and $j_R$ label the spins under $\su(2)_L$ and $\su(2)_R$, respectively.}. We should note that in all the tables below the first column labelled $\HWS$  lists the irreducible $\mathcal{K}$-modules containing the highest-weight vectors which must be of the form \ref{HWS}. 

\begin{table}[htp]
\caption{Below we list the possible highest-weight vectors of $\su(2,\!|4|2)$ irreducible unitary supermultiplets obtained with one colour $P=1$ (doubletons), their fundamental weights, $[\mu_L,\tau,\mu_R;\beta_L,\beta_R]$ labels and their BPS properties.  }
\begin{center}
\begin{tabular}{|c|c|c|c|} \hline
 $\HWS$ & $[E_{11},\cdots, E_{88}] $ &$[\mu_L,\tau,\mu_R;\beta_L,\beta_R] $ & BPS  \\ \hline
$ \fvac $ & $ [-1,-1; 0, 0, 0,0; 0, 0] $& $ [0,000,0;1,0] $ & (0,1)  \\
$f_a^\dagger \fvac $ &  $ [-1,-1; 1,0, 0,0; 0, 0] $ & $ [0,100,0;0,0] $& (1/4,3/4) \\
$f_a^\dagger  f_b^\dagger  \fvac $   &   $ [-1,-1; 1, 1, 0,0;
0, 0] $ & $ [0,110,0;0,0] $ & (1/2,1/2) \\
$f_a^\dagger  f_b^\dagger  f_c^\dagger \fvac $  & $ [-1,-1; 1,
1, 1,0;
0, 0] $ & $ [0,111,0;0,0] $ & (3/4,1/4)  \\
$(a_{\alpha}^\dagger )^m \Delta_f^\dagger  \fvac $  & $ [-1,-1;
1,
1, 1,1;
m, 0] $ & $ [0,000,m;0,1] $ & (1,0)   \\
$ (b_{\dot\alpha}^\dagger )^n \fvac $ & $  [-1,-(1+n);
0,
0, 0,0;
0, 0] $& $[n,000,0;1,0] $ & (0,1) \\ \hline

\end{tabular}
\end{center}
 \label{tableP1}
\end{table}%
To obtain the highest-weight vectors for two colours ($P=2$), one needs simply tensor  the irreducible $\mathcal{K}$-modules for the first colour given in  Table~\ref{tableP1}  with the irreducible $\mathcal{K}$-modules  of the second colour, and then to decompose the result into irreducible $\mathcal{K}$-modules. Furthermore, for two colours one can form the bosonic oscillator determinants given in equation \ref{oscdeterminant}
with which to create additional highest-weight vectors with continuous (anomalous) scaling dimensions. In Tables \ref{tableP2a}-\ref{tableP2f}  below we list the highest-weight vectors obtained by this tensoring procedure for $P=2$ and the resulting irreducible unitary representations of $\su(2,\!|4|2)$ labelled by the tuple $[\mu_L,\tau,\mu_R;\beta_L,\beta_R] $.

\begin{table}[htp]
\caption{Below we list the possible highest-weight vectors of $\su(2,\!|4|2)$ irreducible unitary supermultiplets obtained by tensoring those with one colour $P=1$ with the vacuum state $\fvac_2$ of the second colour.   }
\begin{center}
\begin{tabular}{|c|c|} \hline
$\HWS = \fvac_2 \otimes $ &$[\mu_L,\tau,\mu_R;\beta_L,\beta_R] $ \\ \hline
$ \fvac_1 $ & $ [0,000,0;2,0] $   \\ \hline
$f_a^\dagger \fvac_1 $ & $ [0,100,0;1,0] $ \\ \hline
$f_a^\dagger  f_b^\dagger  \fvac_1 $   & $ [0,110,0;1,0] $  \\ \hline
$f_a^\dagger  f_b^\dagger  f_c^\dagger \fvac_1 $  & $ [0,111,0;1,0] $   \\ \hline
$(a_{\alpha}^\dagger )^m \Delta_f^\dagger  \fvac_1$  & $ [0,000,m;1,1] $   \\ \hline
$ (b_{\dot\alpha}^\dagger )^n \fvac_1 $ & $[n,000,0;2,0] $\\ \hline

\end{tabular}
\end{center}
 \label{tableP2a}
\end{table}%
\begin{table}[htp]
\caption{Below we list the possible highest-weight vectors of $\su(2,\!|4|2)$ irreducible unitary supermultiplets obtained by tensoring those with one colour $P=1$ with the states $ f_d^\dagger  \fvac_2$ of the second colour.   }
\begin{center}
\begin{tabular}{|c|c|} \hline
$\HWS = f_d^\dagger \fvac_2 \otimes $ &$[\mu_L,\tau,\mu_R;\beta_L,\beta_R] $ \\ \hline
$ \fvac_1 $ & $ [0,100,0;1,0] $   \\ \hline
$f_a^\dagger \fvac_1 $ & $ [0,200,0;0,0] $ \\
& [0,110,0;1,0]  \\ \hline
$f_a^\dagger  f_b^\dagger  \fvac_1 $   & $ [0,210,0;0,0] $ \\
& [0,111,0;1,0] \\ \hline
$f_a^\dagger  f_b^\dagger  f_c^\dagger \fvac_1 $  & $ [0,211,0;0,0] $ \\
& [0,000,0;1,1]  \\ \hline
$(a_{\alpha}^\dagger )^m \Delta_f^\dagger  \fvac_1$  & $ [0,100,m;0,1] $   \\ \hline
$ (b_{\dot\alpha}^\dagger )^n \fvac_1 $ & $[n,100,0;1,0] $\\ \hline

\end{tabular}
\end{center}
 \label{tableP2b}
\end{table}%

\begin{table}[htp]
\caption{Below we list the possible highest-weight vectors of $\su(2,\!|4|2)$ irreducible unitary supermultiplets obtained by tensoring those with one colour $P=1$ with the states $ f_d^\dagger  f_e^\dagger  \fvac_2$ of the second colour.   }
\begin{center}
\begin{tabular}{|c|c|} \hline
$\HWS = f_d^\dagger f_e^\dagger \fvac_2 \otimes $ &$[\mu_L,\tau,\mu_R;\beta_L,\beta_R] $ \\ \hline
$ \fvac_1 $ & $ [0,110,0;1,0] $   \\ \hline
$f_a^\dagger \fvac_1 $ & $ [0,210,0;0,0] $ \\
& [0,111,0;1,0]  \\ \hline
$f_a^\dagger  f_b^\dagger  \fvac_1 $   & $ [0,220,0;0,0] $ \\
& [0,211,0;1,0] \\
&[0,000,0;1,1] \\ \hline
$f_a^\dagger  f_b^\dagger  f_c^\dagger \fvac_1 $  & $ [0,221,0;0,0] $ \\
& [0,100,0;0,1]  \\ \hline
$(a_{\alpha}^\dagger )^m \Delta_f^\dagger  \fvac_1$  & $ [0,110,m;0,1] $   \\ \hline
$ (b_{\dot\alpha}^\dagger )^n \fvac_1 $ & $[n,110,0;1,0] $\\ \hline

\end{tabular}
\end{center}
 \label{tableP2c}
\end{table}%
\begin{table}[htp]
\caption{Below we list the possible highest-weight vectors of $\su(2,\!|4|2)$ irreducible unitary supermultiplets obtained by tensoring those with one colour $P=1$ with the states $ f_d^\dagger  f_e^\dagger f_f^\dagger  \fvac_2$ of the second colour.   }
\begin{center}
\begin{tabular}{|c|c|} \hline
$\HWS = f_d^\dagger f_e^\dagger f_f^\dagger \fvac_2 \otimes $ &$[\mu_L,\tau,\mu_R;\beta_L,\beta_R] $ \\ \hline
$ \fvac_1 $ & $ [0,111,0;1,0] $   \\ \hline
$f_a^\dagger \fvac_1 $ & $ [0,211,0;0,0] $ \\
& [0,110,0;0,1]  \\ \hline
$f_a^\dagger  f_b^\dagger  \fvac_1 $   & $ [0,221,0;0,0] $ \\
& [0,100,0;0,1] \\
 \hline
$f_a^\dagger  f_b^\dagger  f_c^\dagger \fvac_1 $  & $ [0,222,0;0,0] $ \\
& [0,110,0;0,1]  \\ \hline
$(a_{\alpha}^\dagger )^m \Delta_f^\dagger  \fvac_1$  & $ [0,111,m;0,1] $   \\ \hline
$ (b_{\dot\alpha}^\dagger )^n \fvac_1 $ & $[n,111,0;1,0] $\\ \hline

\end{tabular}
\end{center}
 \label{tableP2d}
\end{table}%

\begin{table}[htp]
\caption{Below we list the possible highest-weight vectors of $\su(2,\!|4|2)$ irreducible unitary supermultiplets obtained by tensoring those with one colour $P=1$ with the states $(a_i^\dagger )^m \Delta_f^\dagger  \fvac_2$ of the second colour ($m \geq n$). }
\begin{center}
\begin{tabular}{|c|c|} \hline
$\HWS = (a_i^\dagger )^m \Delta_f^\dagger  \fvac_2  \otimes $ &$[\mu_L,\tau,\mu_R;\beta_L,\beta_R] $ \\ \hline
$ \fvac_1 $ & $ [0,000,m;1,1] $   \\ \hline
$f_a^\dagger \fvac_1 $ & $ [0,100,m;0,1] $  \\ \hline
$f_a^\dagger  f_b^\dagger  \fvac_1 $   & $ [0,110,m;0,1] $ \\
 \hline
$f_a^\dagger  f_b^\dagger  f_c^\dagger \fvac_1 $  & $ [0,111,m;0,1] $  \\ \hline
$(a_{\alpha}^\dagger )^n \Delta_f^\dagger  \fvac_1$  & $ [0,000,m+n;0,2] $ \\
&  $ [0,000,m+n-1;0,3] $ \\ &$ \cdots $ \\ & $ [0,000,(m-n);0,2+n] $ \\
\hline
$ (b_{\dot\alpha}^\dagger )^n \fvac_1 $ & $[n,000,m;1,1] $\\ \hline

\end{tabular}
\end{center}
 \label{tableP2e}
\end{table}%

\begin{table}[htp]
\caption{Below we list the possible highest-weight vectors of $\su(2,\!|4|2)$ irreducible unitary supermultiplets obtained by tensoring those with one colour $P=1$ with the states $ (b_r^\dagger )^n \fvac_2 $  of the second colour ($ m \geq n$).   }
\begin{center}
\begin{tabular}{|c|c|} \hline
$\HWS =  (b_r^\dagger )^n \fvac_2  \otimes $ &$[\mu_L,\tau,\mu_R;\beta_L,\beta_R] $ \\ \hline
$ \fvac_1 $ & $ [n,000,0;2,0] $   \\ \hline
$f_a^\dagger \fvac_1 $ & $ [n,100,0;1,0] $  \\ \hline
$f_a^\dagger  f_b^\dagger  \fvac_1 $   & $ [n,110,0;1,0] $ \\
 \hline
$f_a^\dagger  f_b^\dagger  f_c^\dagger \fvac_1 $  & $ [n,111,0;1,0] $  \\ \hline
$(a_{\alpha}^\dagger )^m \Delta_f^\dagger  \fvac_1$  & $ [n,000,m;1,1] $ \\
\hline
$ (b_{\dot\alpha}^\dagger )^m \fvac_1 $ & $[m+n,000,0;2,0] $\\
 & $[m+n-1,000,0;3,0] $ \\ & $\cdots $ \\
& $[(m-n),000,0;2+n,0] $ \\
 \hline

\end{tabular}
\end{center}
 \label{tableP2f}
\end{table}%

\newpage

With two colours of oscillators one can form additional highest-weight vectors of unitary representations of $\su(2,\!|4|2)$ with anomalous dimensions using the oscillator determinants given in equation  \ref{oscdeterminant} :
\be \label{deformedP2}
 [\Delta_b^\dagger ]^{\gamma_L} [b_{\dot\alpha}^\dagger(1)]^m \fvac \Longrightarrow  [\mu_L,\tau,\mu_R;\beta_L,\beta_R] = [m,000,0;(2+\gamma_L),0] \\
{[}\Delta_a^\dagger{]}^{\gamma_R} [a_{\alpha}^\dagger(1)]^n \Delta_f^\dagger (1) \Delta_f^\dagger (2)\fvac \Longrightarrow [\mu_L,\tau,\mu_R;\beta_L,\beta_R]= [0,000,n;0,2+\gamma_R]
\ee
They lead to chiral supermultiplets which are CPT-conjugates of each other for $m=n$ and $\gamma_L=\gamma_R$.

 For $P=3$ possible highest-weight vectors with anomalous dimensions all  involve the action of a single bosonic determinant,  either $[\Delta_a^\dagger]^{\gamma_R}  $ or $[\Delta_b^\dagger]^{\gamma_L}$ ,  modulo the permutation of colour indices. For four colours $P=4$ one can construct highest-weight vectors which involve the action of both types of bosonic determinants, namely

\be \label{deformedP4}
 [\Delta_b(1,2)^\dagger ]^{\gamma_L} [b_{\dot\alpha}^\dagger(1)]^m[\Delta_a(3,4)^\dagger{]}^{\gamma_R} [a_{\alpha}^\dagger(3)]^n \Delta_f^\dagger (3) \Delta_f^\dagger (4)\fvac \ee
corresponding to the unitary irreps with the labels

\be
 [\mu_L,\tau,\mu_R;\beta_L,\beta_R] = [m,000,n;(2+\gamma_L),(2+\gamma_R)]\,.
 \ee

By increasing the number of colours, one can construct the most general irreducible $\mathcal{K}$-module with the highest-weight vector given in equation \ref{HWS}.

As stated earlier, the full spectrum of IIB supergravity was first obtained  by tensoring of CPT self-conjugate doubleton supermultiplet of $\su(2,2|4)$ with itself repeatedly and restricting to the CPT self-conjugate supermultiplets in~\cite{Gunaydin:1984fk}. CPT self-conjugate doubleton supermultiplet is simply the $4D$ Yang-Mills supermultiplet and has the highest-weight vector
\be
[\mu_L,\tau,\mu_R;\beta_L,\beta_R] = [0, (1,1,0), 0; 0,0] \, \Leftrightarrow \, \textrm{ 4d, $\mathcal{N}$=4 \, Yang-Mills \, supermultiplet. }
\ee
For two colours $P=2$ the CPT self-conjugate supermultiplet is  the   massless graviton supermultiplet in AdS$_5$ with  the
highest-weight vector
\be
[\mu_L,\tau,\mu_R;\beta_L,\beta_R] = [0, (2,2,0), 0; 0,0] \, \Leftrightarrow \, \textrm{Massless  AdS$_5$,  $\mathcal{N}$=8 \, graviton \, supermultiplet.}
\ee
The full spectrum of the IIB supergravity over AdS$_5\times$S$^5$ consists of the following infinite tower of CPT self-conjugate  $(1/2,1/2)$-BPS supermultiplets:
\be
\textrm{Spectrum of IIB Supergravity over} \,\,  {\rm AdS}_5\times {\rm S}^5 \,\, = \sum_{n=2}^{\infty} \bigoplus [ 0,(n,n,0),0;0,0]\,.
\ee
They correspond to $(1/2,1/2)$-BPS supermultiplets of $\psu(2,2|4)$ with vanishing central charge. We should also  
note that we did not include the Yang-Mills multiplet in the spectrum since it decouples as gauge modes and its field theory lives on the boundary of AdS$_5$ as the $\mathcal{N}$=4  superconformal Yang-Mills theory \cite{Gunaydin:1984fk}.


\newpage
\bibliography{combined}
\bibliographystyle{bibstyle}

\end{document}